\documentclass[%
preprint,
superscriptaddress,
 amsmath,amssymb,
 aps,
prb,
citeautoscript,
]{revtex4-2}

\usepackage[english]{babel}

\usepackage[letterpaper,top=2cm,bottom=2cm,left=2cm,right=2cm,marginparwidth=1.75cm]{geometry}
\usepackage{graphicx} 
\usepackage{placeins}
\usepackage{amsmath}
\usepackage{amssymb}
\usepackage{physics}
\usepackage{xcolor}
\DeclareRobustCommand{\TDK}[1]{{\color{blue}#1}}
\usepackage{amsthm}
\usepackage[colorlinks=true, allcolors=blue]{hyperref}
\usepackage{quantikz}
\usepackage{appendix}
\usepackage{algorithm}
\usepackage{algpseudocode}

\usepackage{tikz}
\usepackage{diagbox}
\usepackage{mathtools} 

\newtheoremstyle{defnopunct}%
  {3pt}{3pt}
  {}
  {}
  {\bfseries}
  {}
  {0.5em}
  {}

\theoremstyle{defnopunct}

\newtheorem{lemma}{Lemma}[section]
\newtheorem{theorem}{Theorem}[section]

\newtheorem{remark}{Remark}[section]
\newtheorem{assumption}{Assumption}[section]
\usepackage{comment} 
\usepackage{parskip}
\usepackage{csquotes}
\newtheorem*{mainresult}{Main result.}
\newcommand{\aQa}{\affiliation{%
    $\langle a Q a^L \rangle$ Applied Quantum Algorithms,
    Leiden University, the Netherlands}
}
\newcommand{\inner}[2]{\langle #1, #2\rangle}
\newcommand{\Ran}{\operatorname{Ran}}

\begin{document}
\title{
From Wavefunction Regularity to Eigenvector Conditioning:\\
Accuracy--Cost Trade-offs of Quantum Algorithms for Non-Hermitian Transcorrelated Hamiltonians\\
}

\author{Cheng-Lin Hong}
\affiliation{Center for Advanced Systems Understanding (CASUS), 02826 Görlitz, Germany}
\affiliation{Helmholtz-Zentrum Dresden-Rossendorf (HZDR), 01328 Dresden, Germany}

\author{Emiel Koridon}
 \aQa
 \affiliation{Covestro Deutschland AG, Leverkusen, Nordrhein-Westfalen 51373, Germany}

\author{Paul~K.~Faehrmann}
 \affiliation{Dahlem Center for Complex Quantum Systems, Freie Universit\"{a}t Berlin, 14195 Berlin, Germany}
 \affiliation{Walrus Computing, Munich, Germany}

\author{Thomas D. K\"uhne}
\affiliation{Center for Advanced Systems Understanding (CASUS), 02826 Görlitz, Germany}
\affiliation{Helmholtz-Zentrum Dresden-Rossendorf (HZDR), 01328 Dresden, Germany}
\affiliation{Institute of Artificial Intelligence, Technische Universit\"at Dresden, Helmholtzstra{\ss}e 10, 01069 Dresden, Germany}

\author{Stefano Polla}
 \affiliation{QuSoft, HIMS \& IvI, University of Amsterdam, the Netherlands}
 \aQa

\author{Werner Dobrautz}
\affiliation{Center for Advanced Systems Understanding (CASUS), 02826 Görlitz, Germany}
\affiliation{Helmholtz-Zentrum Dresden-Rossendorf (HZDR), 01328 Dresden, Germany}
\affiliation{Center for Scalable Data Analytics and Artificial Intelligence (ScaDS.AI) Dresden/Leipzig, 01069 Dresden, Germany}
\affiliation{Technical University Dresden, 01069 Dresden, Germany}

\begin{abstract}
The electron--electron Coulomb singularity forces cusp structure in electronic
wavefunctions. This cusp structure limits wavefunction regularity and
slows the convergence of expansions in smooth bases, requiring large basis sets
to reach a target accuracy. The transcorrelated (TC) method incorporates this
short-range behavior into the Hamiltonian, thereby improving the regularity of
the transformed wavefunction and accelerating convergence toward the
complete-basis-set limit. The resulting Hamiltonian is non-Hermitian and
non-normal, which limits the direct applicability of standard quantum
algorithms. For such non-normal operators, the complexity of several recent
quantum algorithms depends explicitly on quantities related to eigenvector
conditioning and eigenvalue sensitivity. Understanding these quantities is
essential for assessing both algorithmic complexity and quantum resource
requirements. In this work, we address these issues in two stages. We first
study a solvable model, which provides a controlled setting for isolating how
transcorrelation affects basis convergence and eigenvector conditioning. Within
this setting, we investigate two distinct one-parameter Jastrow correlators. We
derive  the asymptotic energy-error scaling: $N^{-1/2}$ for the bare projection
and $N^{-3/2}$ for both TC correlators at generic fixed parameters, with $N$
denoting the basis size. Furthermore, we show that this faster convergence can
yield a polynomial improvement in the asymptotic query-complexity upper bound.
We then extend the analysis to second-quantized electronic-structure
Hamiltonians and use particle-number-sector diagnostics to distinguish
target-sector from auxiliary-sector conditioning.  Our work provides a framework
for understanding the central accuracy--cost trade-off in TC Hamiltonians.
\end{abstract}

\date{\today}

\maketitle
\section{Introduction}
\label{sec:intro}
High-accuracy quantum chemistry simulations remain challenging due to
the computational cost of accurately describing electron correlation effects.
While quantum computing offers a
promising framework for simulating many-electron systems, the accuracy
achievable in both classical and quantum calculations depends on how
efficiently the electronic Hamiltonian and its states can be represented. In particular,
both classical and quantum calculations 
typically operate
on Hamiltonians projected onto finite one-electron basis sets, introducing a
systematic approximation. The corresponding many-electron space is built from
antisymmetrized products of the selected orbitals, and therefore does not
incorporate the interelectronic coordinate $r_{12}$ explicitly. The slow
convergence of the correlation energy reflects the difficulty of representing
the correct short-range behavior of the electronic wave function at
electron--electron coalescence points, known as Kato's cusp
conditions~\cite{kato_eigenfunctions_1957}.

To address this convergence issue, explicitly correlated methods incorporate
the interelectronic distance $r_{12}$ directly into the
wave-function ansatz. The idea goes back to the pioneering work of Hylleraas on
the helium atom~\cite{hylleraas_neue_1929} and was later developed by Kutzelnigg
and coworkers~\cite{kutzelnigg_r12-dependent_1985,
klopper_mibller-plesset_1987}. For detailed reviews, see
Refs.~\cite{klopper_r12_2006,kong_explicitly_2012,hattig_explicitly_2012,
gruneis_perspective_2017}. While these approaches have been highly successful,
an alternative perspective is to incorporate the effects of electron correlation directly into
the Hamiltonian via a similarity transformation, rather than modifying the
wave-function ansatz itself. This approach, known as the transcorrelated (TC)
method, was initially proposed by Hirschfelder~\cite{hirschfelder_removal_1963}
and subsequently refined by Boys and Handy~\cite{boys_condition_1969,
boys_handy_TC_1969}. In this framework, the Hamiltonian is
similarity-transformed using a Jastrow
factor~\cite{PhysRev.98.1479}, which is chosen to encode the cusp structure.

This transformation introduces two main challenges. First,
this similarity transformation generates
effective interactions up to the three-body level, which
increases the cost of constructing and
representing the Hamiltonian. In practice, this difficulty can be mitigated by
the xTC approximation~\cite{christlmaier_xtc_2023}. Second, the resulting
transcorrelated Hamiltonian is non-Hermitian
and non-normal in the standard inner
product.
Unlike normal operators, non-normal operators can have
non-orthogonal eigenvector bases, and the resulting eigenvector conditioning is
closely connected to their spectral sensitivity~\cite{higham_accuracy_2002,
trefethen_spectra_2005}. Basis-set acceleration must therefore be assessed
together with the representation and conditioning costs introduced by the
transformation.

Despite these challenges, the TC method has seen renewed attention in recent
years~\cite{ten-no_feasible_2000,hino_biorthogonal_2001,
hino_application_2002,jeszenszki_accelerating_2018,cohen_similarity_2019,
giner2021new,baiardi_transcorrelated_2020,PhysRevResearch.3.033072,
schraivogel_transcorrelated_2021,schraivogel_transcorrelated_2023, Haupt2023-uo, Liao2026}. On the
classical side, transcorrelated Hamiltonians have been treated successfully in
transcorrelated coupled-cluster approaches~\cite{PhysRevResearch.3.033072,
schraivogel_transcorrelated_2021,schraivogel_transcorrelated_2023}, in full
configuration interaction quantum Monte Carlo
(FCIQMC)~\cite{luo_combining_2018, Dobrautz2019, Guther2020, guther_binding_2021, Dobrautz2022}, selected configuration interaction~\cite{Ammar2024}, and in density-matrix
renormalization group (DMRG) methods~\cite{baiardi_transcorrelated_2020,
liao_density_2023}. More recently, TC Hamiltonians have also been studied in quantum
computing~\cite{motta_quantum_2020,
mcardle2020improvingaccuracyquantumcomputational,kumar_quantum_2022,
PhysRevResearch.5.023174,dobrautz_toward_2024,magnusson_towards_2024,
feniou_real-space_2025,uvarov_accuracy_2025}.

Earlier strategies for integrating the TC framework with quantum computing have
primarily followed two distinct directions. The first approach is based on
canonical transcorrelation \cite{yanai_canonical_2012}, in which 
the associated Baker--Campbell--Hausdorff expansion is truncated, yielding an approximate Hermitian Hamiltonian compatible with standard quantum algorithms~\cite{motta_quantum_2020,
kumar_quantum_2022}. The second focuses on variational approaches tailored to
near-term devices, which rely on algorithms such as variational quantum
imaginary-time evolution (VarQITE) to cope with the non-Hermitian
structure~\cite{mcardle2020improvingaccuracyquantumcomputational,
PhysRevResearch.5.023174,dobrautz_toward_2024,magnusson_towards_2024}.

As quantum computing progresses toward fault-tolerant regimes, new frameworks
have emerged for eigenvalue transformation and estimation beyond the Hermitian
setting. Early progress in this direction was made by
Shao~\cite{shao_computing_2022}, who developed algorithms for diagonalizable
matrices with real spectra and also considered more general settings under
additional assumptions. More recently, Low and Su introduced the quantum eigenvalue processing (QEP) framework~\cite{Low_2026}, which enables eigenvalue transformation (QEVT) and estimation (QEVE) beyond the Hermitian setting. QEVE has since been applied to transcorrelated Hamiltonians. Feniou \emph{et al.}~\cite{feniou_real-space_2025} proposed a first-quantized real-space framework in which transcorrelated Hamiltonians can be treated with QEVE, while Uvarov and Izmaylov~\cite{uvarov_accuracy_2025} subsequently provided the first resource estimates for applying QEVE to second-quantized TC Hamiltonians and compared these estimates with those for standard qubitization. Their study also emphasized that estimating the relevant condition number and understanding its behavior remain key challenges for assessing the practical success of QEVE. These challenges raise three questions that have not yet been systematically investigated in the quantum-algorithmic setting: (i) how do finite-basis projection and choice of correlator affect accuracy and eigenvector conditioning in TC Hamiltonians, (ii) how does the resulting conditioning evolve as the basis is enlarged, and (iii) under what conditions does this faster basis convergence translate into an improved query-complexity bound for a target total accuracy?

In this work, we address these questions in two stages. We first analyze an exactly solvable model, the one-dimensional harmonic oscillator with a delta-function potential. This model provides a clean, analytically tractable setting in which all questions can be addressed: we compare two different one-parameter correlators, track basis-set convergence and eigenvector conditioning as the basis is enlarged, and analyze the resulting quantum query complexity.
We then extend the analysis to our main object of interest, second-quantized electronic Hamiltonians, to understand how the mechanisms identified in the solvable model appear in electronic-structure calculations. By combining these analyses with diagnostics relevant to quantum algorithms, we provide a physically motivated perspective on the trade-off between the accuracy gains from transcorrelation and the associated algorithmic costs in electronic-structure calculations.

The rest of this paper is organized as follows. Sec.~\ref{sec:theory} develops the TC operator in real space and finite-dimensional representations and introduces
the conditioning measures used throughout. Sec.~\ref{sec:SHO-delta} applies this framework
to the exactly solvable contact-interaction model. Sec.~\ref{sec:molecular} examines accuracy--cost trade-offs and algorithm-relevant diagnostics for second-quantized electronic Hamiltonians. Sec.~\ref{sec:conclusion} concludes and outlines future directions.

\section{Theoretical Framework}
\label{sec:theory}
\subsection{Transcorrelated Hamiltonian in Real Space}
\label{sec:TC-real-space}
Consider a Schr\"odinger-type operator of the form
\begin{equation}
    \hat H=\hat T+\hat V,
\end{equation}
where
\begin{equation}
    \hat T=-\frac{1}{2}\sum_i\nabla_i^2
\end{equation}
is the kinetic-energy operator, and $\hat V$ denotes the coordinate-dependent
part of the Hamiltonian, including the external potential and interparticle
interactions. The corresponding eigenvalue problem is
\begin{equation}
    \hat H\Psi_k=E_k\Psi_k.
\end{equation}

Coulombic or contact singularities in $\hat V$ can induce nonanalytic
structure in the eigenfunctions of $\hat H$
\cite{kato_eigenfunctions_1957,albeverio_solvable_1988}. Finite expansions in
smooth basis functions do not represent this structure explicitly, which can
limit the convergence of the resulting wave functions and Ritz eigenvalues
\cite{hill_rates_1985}.

For Coulombic electronic Hamiltonians, the eigenfunctions admit a factorized
representation
\begin{equation}
    \Psi_k=\mathcal F\,\varphi_k,
\end{equation}
where $\mathcal F$ is an explicit universal factor containing the nonanalytic
coalescence structure, while the residual function $\varphi_k$ has higher
regularity~\cite{fournais_sharp_2005,yserentant_regularity_2010}. Motivated by this factorization, we
introduce a parametrized Jastrow factor $e^{\hat\tau_{\boldsymbol\theta}}$, where
$\hat\tau_{\boldsymbol\theta}$ is a real-valued, coordinate-dependent
multiplicative correlator, and define the residual function $\Phi_k^R$ by
\begin{equation}
    \Psi_k=e^{\hat\tau_{\boldsymbol\theta}}\Phi_k^R.
    \label{eq:tc-state-factorization}
\end{equation}
The parametrized factor need not coincide with the 
universal factor
$\mathcal F$; it is chosen to encode selected nonanalytic features.

Substitution of Eq.~\eqref{eq:tc-state-factorization} into the original
eigenvalue equation, followed by left multiplication by
$e^{-\hat\tau_{\boldsymbol\theta}}$, gives
\begin{equation}
    e^{-\hat\tau_{\boldsymbol\theta}}
    \hat H
    e^{\hat\tau_{\boldsymbol\theta}}
    \Phi_k^R
    =E_k\Phi_k^R,
    \label{eq:transformed-eigenvalue-equation}
\end{equation}
where the superscript $R$ identifies $\Phi_k^R$ as a right eigenfunction of the
transformed operator. The transcorrelated Hamiltonian is thus defined at the
operator level by
\begin{equation}
    \hat H_{\mathrm{TC}}(\boldsymbol\theta)
    =e^{-\hat\tau_{\boldsymbol\theta}}
     \hat H
     e^{\hat\tau_{\boldsymbol\theta}},
    \qquad
    S_{\boldsymbol\theta}=e^{\hat\tau_{\boldsymbol\theta}}.
    \label{eq:similarity_transform}
\end{equation}
Here $\boldsymbol\theta$ parametrizes the chosen correlator and
determines properties such as its functional form or range. These parameters
characterize the transformation rather than the bare Hamiltonian $\hat H$, and for admissible $\boldsymbol\theta$ the spectrum of
$\hat H_{\mathrm{TC}}(\boldsymbol\theta)$ coincides with that of $\hat H$.

For electronic-structure applications, $\hat\tau_{\boldsymbol\theta}$ is usually
chosen as a symmetric pairwise Jastrow correlator,
\begin{equation}
    \hat \tau_{\boldsymbol\theta}
    = \sum_{i<j} u_{\boldsymbol\theta}(\mathbf r_i,\mathbf r_j),
    \qquad
    u_{\boldsymbol\theta}(\mathbf r_i,\mathbf r_j)
    =u_{\boldsymbol\theta}(\mathbf r_j,\mathbf r_i).
    \label{eq:pairwise_tau}
\end{equation}
Throughout the following discussion, we write
$u_{ij} \equiv u_{\boldsymbol\theta}(\mathbf r_i,\mathbf r_j)$, where
$\mathbf r_i$ denotes the coordinate of electron $i$.
In the one-dimensional model considered later, $\hat \tau_{\boldsymbol\theta}$
becomes the
multiplication operator generated by a scalar function of one coordinate.

To obtain the explicit form of
$\hat{H}_{\mathrm{TC}}(\boldsymbol\theta)$, we employ the
Baker--Campbell--Hausdorff (BCH) expansion.
Since
$\hat{\tau}_{\boldsymbol\theta}$ is multiplicative, it commutes with all
multiplicative terms in $\hat{V}$, so that the nontrivial commutators arise from the
kinetic operator $\hat{T}$. Because $\hat{T}$
contains derivatives only up to second order, the BCH series
truncates exactly after the second commutator (see
Appendix~\ref{app:BCH_derivation}), yielding
\begin{equation}
    \label{eq:Htc-commutator}
    \hat{H}_{\mathrm{TC}}(\boldsymbol\theta)= \hat{H} + \comm{\hat{T}}{\hat{\tau}_{\boldsymbol\theta}} + \frac{1}{2} \comm{\comm{\hat{T}}{\hat{\tau}_{\boldsymbol\theta}}}{\hat{\tau}_{\boldsymbol\theta}}.
\end{equation}

The resulting Hamiltonian contains a generally non-Hermitian
first-order differential operator and a scalar correction generated by the
similarity transformation.
Substituting the explicit commutators gives
\begin{equation}
    \label{eq:tc-real-space}
    \hat{H}_{\mathrm{TC}}(\boldsymbol\theta)=
    \hat {H}
    - \sum_i (\nabla_i  \tau_{\boldsymbol\theta})\!\cdot\!\nabla_i
    - \frac{1}{2}\sum_i \nabla_i^2  \tau_{\boldsymbol\theta}
    - \frac{1}{2}\sum_i |\nabla_i  \tau_{\boldsymbol\theta}|^2,
\end{equation}
where $\tau_{\boldsymbol\theta}$ denotes the coordinate-space representation
of $\hat{\tau}_{\boldsymbol\theta}$, and $\mathbf r$ collectively denotes all particle coordinates.

It is convenient to distinguish the first-order differential operator,
the scalar correction generated by the TC transformation, and the complete
transformed scalar potential.
We define
\begin{equation}
    \mathbf{F}_{i,\boldsymbol\theta}(\mathbf{r}) \equiv \nabla_i \tau_{\boldsymbol\theta}(\mathbf{r}),
    \label{eq:first_order_field_def}
\end{equation}

\begin{equation}
    V_{\mathrm{TC,eff},\boldsymbol\theta}(\mathbf{r})
    \equiv
    - \frac{1}{2}\sum_i \left[
        \nabla_i^2 \tau_{\boldsymbol\theta}(\mathbf{r})
        + |\nabla_i \tau_{\boldsymbol\theta}(\mathbf{r})|^2
    \right].
    \label{eq:Veff_def}
\end{equation}
The complete transformed scalar potential is then
\begin{equation}
    V_{\mathrm{eff},\boldsymbol\theta}(\mathbf r)
    \equiv V(\mathbf r)
    +V_{\mathrm{TC,eff},\boldsymbol\theta}(\mathbf r),
    \label{eq:Veff_total_def}
\end{equation}

Here $V(\mathbf r)$ is the coordinate-space representation of $\hat V$.
Thus, $V_{\mathrm{TC,eff},\boldsymbol\theta}$ contains only the scalar terms
generated by the similarity transformation, whereas
$V_{\mathrm{eff},\boldsymbol\theta}$ also contains the original potential.
With these definitions, the transcorrelated Hamiltonian can be
written as
\begin{equation}
    \hat{H}_{\mathrm{TC}}(\boldsymbol\theta)
    =
    \hat H
    - \sum_i \mathbf{F}_{i,\boldsymbol\theta}(\mathbf{r}) \cdot \nabla_i
    + V_{\mathrm{TC,eff},\boldsymbol\theta}(\mathbf{r}).
    \label{eq:TC-first-order-form}
\end{equation}

Formally, a sufficient condition for admissibility is that
$S_{\boldsymbol\theta}$ is boundedly invertible on the Hilbert space
of interest.
Under this assumption, the transformation in
Eq.~\eqref{eq:similarity_transform}
can equivalently be written as
\begin{equation}
    \hat{H}_{\mathrm{TC}}(\boldsymbol\theta)
    = S_{\boldsymbol\theta}^{-1}\hat{H}S_{\boldsymbol\theta},
\end{equation}
which defines a genuine similarity transformation; consequently,
$\hat{H}_{\mathrm{TC}}(\boldsymbol\theta)$ and $\hat H$ share the
same spectrum.  We adopt this assumption throughout the operator-level
discussion; more general unbounded similarity transformations need
not preserve the spectrum~\cite{lowdin_change_1988}.

Beyond preserving the spectrum,
the bounded invertibility of $S_{\boldsymbol\theta}$ gives
$\hat{H}_{\mathrm{TC}}(\boldsymbol\theta)$ an additional structural property.
Assuming that $\hat H$ is self-adjoint, define the positive metric operator as
\begin{equation}
    \eta_{\boldsymbol\theta}
    = S_{\boldsymbol\theta}^\dagger S_{\boldsymbol\theta} .
\end{equation}
Provided the relevant domains are compatible, the TC
Hamiltonian is
\textit{quasi-Hermitian} with respect to
$\eta_{\boldsymbol\theta}$~\cite{MR187086,antoine_remarks_2014,scholtz1992quasi},
in the sense that
\begin{equation}
    \hat{H}_{\mathrm{TC}}(\boldsymbol\theta)^\dagger \eta_{\boldsymbol\theta}
    = \eta_{\boldsymbol\theta} \hat{H}_{\mathrm{TC}}(\boldsymbol\theta).
    \label{eq:quasi-hermitian-eq}
\end{equation}
Under this domain-compatibility assumption,
$\hat{H}_{\mathrm{TC}}(\boldsymbol\theta)$
may equivalently be regarded as self-adjoint with respect to the
modified inner product induced by $\eta_{\boldsymbol\theta}$.

\subsection{Finite-Dimensional Transcorrelated Hamiltonians}
\label{sec:finite-dimension-tc}
In practice, neither classical nor quantum computations are performed directly
in the infinite-dimensional Hilbert space.
Instead, one works in a finite-dimensional subspace obtained by
discretization or basis truncation. We denote this subspace by $\mathcal B$ and
let $P_{\mathcal B}$ be the orthogonal projector onto it. For a
selected correlator parameter
$\boldsymbol\theta$, the object used in computation is therefore not
$\hat H_{\mathrm{TC}}(\boldsymbol\theta)$ itself but
its finite-basis representation
\begin{equation}
    H_{\mathrm{TC}}^{\mathcal B}(\boldsymbol\theta)
    = P_{\mathcal B}\hat H_{\mathrm{TC}}(\boldsymbol\theta)P_{\mathcal B},
    \label{eq:projected_tc}
\end{equation}
with associated right eigenpairs
\begin{equation}
    H_{\mathrm{TC}}^{\mathcal B}(\boldsymbol\theta)
    |\Phi_k^{\mathcal B}(\boldsymbol\theta)\rangle
    = E_k^{\mathcal B}(\boldsymbol\theta)
    |\Phi_k^{\mathcal B}(\boldsymbol\theta)\rangle .
    \label{eq:projected_tc_eigenproblem}
\end{equation}

We illustrate the resulting structure in the molecular case, where
$\mathcal B$ is the many-electron space generated by a finite set of orthonormal
spatial orbitals $\{\phi_p(\mathbf r)\}_{p=1}^M$. The correlator
$\hat\tau_{\boldsymbol\theta}$ takes the pairwise Jastrow form of
Eq.~\eqref{eq:pairwise_tau}.
For this pairwise correlator, the additional terms in
Eq.~\eqref{eq:tc-real-space} can be grouped into the
$\boldsymbol\theta$-dependent effective two- and three-body operators
(see Appendix~\ref{app:pairwise_derivation}):
\begin{subequations}
\begin{align}
    \label{eq:K_operator_exact}
    \hat{K}_{\boldsymbol\theta}(\mathbf{r}_i, \mathbf{r}_j)
    &= \frac{1}{2}\Big(
        \nabla_i^2 u_{ij}
        + \nabla_j^2 u_{ij}
        + |\nabla_i u_{ij}|^2
        + |\nabla_j u_{ij}|^2
    \Big) \nonumber \\
    &\quad + \nabla_i u_{ij}\cdot\nabla_i
        + \nabla_j u_{ij}\cdot\nabla_j , \\
    \label{eq:L_operator_exact}
    \hat{L}_{\boldsymbol\theta}(\mathbf{r}_i, \mathbf{r}_j, \mathbf{r}_k)
    &= \nabla_i u_{ij}\cdot\nabla_i u_{ik}
     + \nabla_j u_{ji}\cdot\nabla_j u_{jk}
     + \nabla_k u_{ki}\cdot\nabla_k u_{kj},
\end{align}
\end{subequations}
where $u_{ij} \equiv u_{\boldsymbol\theta}(\mathbf r_i, \mathbf r_j)$
as defined in Eq.~\eqref{eq:pairwise_tau}.

Let $a_{p\sigma}^\dagger$ and $a_{p\sigma}$ denote the
fermionic creation and annihilation operators associated with orbital
$\phi_p$ and spin $\sigma$. Taking the matrix elements of the one-, two-, and
three-body operators in this orbital basis gives the second-quantized
representation~\cite{cohen_similarity_2019}
\begin{equation}
\begin{aligned}
    H_{\mathrm{TC}}^{\mathcal B}(\boldsymbol\theta)
    =\;
    &\sum_{pq}\sum_{\sigma}
        h_{pq}\, a_{p\sigma}^\dagger a_{q\sigma} \\
    &+ \frac{1}{2}\sum_{pqrs}\sum_{\sigma\tau}
        \bigl(V_{pqrs}-K_{pqrs}(\boldsymbol\theta)\bigr)\,
        a_{p\sigma}^\dagger a_{q\tau}^\dagger a_{s\tau} a_{r\sigma} \\
    &- \frac{1}{6}\sum_{pqrstu}\sum_{\sigma\tau\lambda}
        L_{pqrstu}(\boldsymbol\theta)\,
        a_{p\sigma}^\dagger a_{q\tau}^\dagger a_{r\lambda}^\dagger
        a_{u\lambda} a_{t\tau} a_{s\sigma},
\end{aligned}
\end{equation}
where $h_{pq}$ and $V_{pqrs}$ denote the usual one- and two-electron
integrals, while
\begin{align}
    K_{pqrs}(\boldsymbol\theta)
    &= \mel{\phi_p\phi_q}{\hat{K}_{\boldsymbol\theta}}{\phi_r\phi_s}, \\
    L_{pqrstu}(\boldsymbol\theta)
    &= \mel{\phi_p\phi_q\phi_r}{\hat{L}_{\boldsymbol\theta}}{\phi_s\phi_t\phi_u},
    \label{eq:3BL}
\end{align}
are the effective two- and three-body integrals. 

The finite-basis representation $H_{\mathrm{TC}}^{\mathcal
B}(\boldsymbol\theta)$ depends on two distinct inputs: (i) the finite space
$\mathcal B$, which specifies the representable subspace, and (ii) the
correlator parameter $\boldsymbol\theta$, which determines the underlying
non-unitary transformation. Both inputs affect the structure of the
finite-dimensional representation.
At finite $\mathcal B$, the
eigenvalues $E_k^{\mathcal B}(\boldsymbol\theta)$ need not coincide with those
of $\hat H$. Under the usual spectral-convergence assumptions, they converge to
the corresponding eigenvalues of $\hat H$ along a basis sequence approaching
the complete-basis-set limit.
The transcorrelated Hamiltonian is designed to accelerate this convergence.
The underlying mechanism is the regularity gain of
Eq.~\eqref{eq:tc-state-factorization}, which can improve approximability in
suitable finite-dimensional spaces~\cite{yserentant_mixed_2011}, although
the realized convergence rate still depends on $\mathcal B$ and
$\boldsymbol\theta$.

However, this acceleration comes at a structural cost. The projected TC
operator $H_{\mathrm{TC}}^{\mathcal B}(\boldsymbol\theta)$
is generally non-Hermitian and non-normal. In contrast, the bare projection
$P_{\mathcal B}\hat H P_{\mathcal B}$ remains Hermitian and therefore normal.
It has orthonormal eigenvectors
and real eigenvalues, and its Ritz values satisfy Rayleigh--Ritz upper bounds.
These properties need not hold in the projected TC setting, where the
eigenbasis is generically non-orthogonal.

\subsection{Conditioning and Eigenvalue Sensitivity}
\label{sec:cond}
A finite-dimensional operator $A$ admits an orthonormal eigenbasis if and
only if it is normal ($A^\dagger A = A A^\dagger$). The projected operator
$H_{\mathrm{TC}}^{\mathcal B}(\boldsymbol\theta)$ is generally non-normal, not
merely non-Hermitian, and therefore admits no orthonormal eigenbasis. The
departure from normality has direct
numerical consequences and can amplify
eigenvalue sensitivity to small perturbations, including numerical
errors~\cite{higham_accuracy_2002,trefethen_spectra_2005}.
This feature is relevant
to recent quantum algorithms that estimate
eigenvalues of non-normal operators: a measure of non-normality
enters their
query-complexity bounds. Such bounds have been
established in
\cite{shao_computing_2022,Low_2026,3n8f_k8pl,zhang2025heisenberglimit} under
suitable assumptions on the input operator.

The quantity appearing in several of these bounds is the condition number
of an eigenvector matrix.
For a diagonalizable matrix $A \in \mathbb{C}^{N\times N}$, we
write
\begin{equation}
  A=V\Lambda V^{-1},
\end{equation}
where $\Lambda$ is diagonal and $V$ is a matrix whose columns are right
eigenvectors of $A$. Throughout this work the condition number of an
eigenvector matrix $V$ is taken in the spectral norm,
\begin{equation}
    \kappa(V) \equiv \kappa_2(V)=\|V\|_2\,\|V^{-1}\|_2.
\end{equation}

The eigenvector matrix $V$ is not unique: its columns may be rescaled
independently, and within a degenerate eigenspace any linearly independent basis
may be chosen. Both choices can change $\kappa(V)$. A convention-independent
quantity is obtained by minimizing over all such choices, yielding the \emph{optimal}
condition number
\begin{equation}
    \kappa^{*} \;=\; \inf_{V} \kappa(V),
\end{equation}
where the infimum is taken over all eigenvector matrices $V$ that
diagonalize $A$.
Evaluating this infimum generally requires an optimization, even when all
eigenvalues are distinct~\cite{doi:10.1137/S0363012992238680}. In practice we
therefore normalize each right eigenvector to unit $2$-norm, as is standard. For
a distinct spectrum this determines $\kappa(V)$ uniquely, and the resulting value
exceeds $\kappa^{*}$ by at most a factor of
$\sqrt{N}$~\cite{van_der_sluis_condition_1969}. However, for a degenerate
spectrum this convention is not sufficient, since it does not select a basis
within a degenerate eigenspace. Different bases of that eigenspace can give different
values of $\kappa(V)$, which can be made arbitrarily large. Such a value is still
an upper bound on $\kappa^{*}$, but a large one does not by itself establish poor
intrinsic conditioning. The ambiguity originates in the degeneracy rather than in
the non-normality of $A$; see Appendix~\ref{app:condition-number}.

To resolve this ambiguity, for each distinct eigenvalue $\lambda_i$
($i=1,\ldots,r$) we further choose an orthonormal basis $Q_i$ of the exact
eigenspace $\ker(A-\lambda_i I)$ and collect these blocks in
\begin{equation}
    V_{\mathrm{os}}=\bigl[\,Q_1\ \cdots\ Q_r\,\bigr] .
    \label{eq:Vos}
\end{equation}
Although $V_{\mathrm{os}}$ itself is still not unique, this choice further
reduces the freedom within each eigenspace: arbitrary invertible changes of basis
are restricted to unitary ones, which preserve all singular
values~\cite{golub_matrix_2013}. Consequently, $\kappa(V_{\mathrm{os}})$ is
independent of the particular orthonormal basis chosen within each exact
eigenspace.

To relate the convention-defined $\kappa(V_{\mathrm{os}})$ to the intrinsic
optimum $\kappa^{*}$, let $P_i$ denote the spectral projector onto the
$m_i$-dimensional eigenspace $\ker(A-\lambda_i I)$; it is fixed by $A$ alone.
The corresponding projector norms satisfy the
bounds~\cite{doi:10.1137/0720040,10.1145/152613.152617}
\begin{equation}
   \max_i \|P_i\|_2 \;\leq\;
   \kappa^{*} \;\leq\;
   \kappa(V_{\mathrm{os}})
   \;\leq\; r \cdot \max_i \|P_i\|_2 .
   \label{eq:conditioning-bounds}
\end{equation}
Thus, $\max_i\|P_i\|_2$ and $\kappa(V_{\mathrm{os}})$ provide,
respectively, an intrinsic lower bound and a convention-defined upper bound on
$\kappa^{*}$.

When $\lambda_i$ is simple ($m_i=1$), the norm of its spectral
projector reduces to the eigenvalue condition number. Let $x$ and $y$ denote
the associated right and left eigenvectors of $A$, satisfying
\begin{equation}
   A x = \lambda_i x, \qquad A^\dagger y = \bar\lambda_i y.
\end{equation}
The eigenvalue condition number is then defined by
\begin{equation}
   \kappa(\lambda_i; A) \;=\; \frac{\|x\|_2\, \|y\|_2}{|y^\dagger x|}.
\end{equation}
A closely related quantity in the physics literature is the
Petermann factor~\cite{Petermann_factor}, which quantifies
left-right eigenvector non-orthogonality; for a simple eigenvalue,
it is the square of $\kappa(\lambda_i; A)$.

In addition, if $A$ commutes with a Hermitian symmetry operator $G$, the
orthogonal eigenspaces of $G$ reduce $A$, giving
$A=\bigoplus_\alpha A_\alpha$. If
$A_\alpha=V_\alpha\Lambda_\alpha V_\alpha^{-1}$, then
$V=\bigoplus_\alpha V_\alpha$ is an associated full-space eigenvector matrix and
\begin{equation}
    \kappa(V)
    =
    \frac{\max_\alpha\sigma_{\max}(V_\alpha)}
         {\min_\alpha\sigma_{\min}(V_\alpha)}
    \;\geq\;
    \max_\alpha\kappa(V_\alpha) .
    \label{eq:direct-sum-cond}
\end{equation}

Finally, any computed $\kappa(V)$ is an upper bound on $\kappa^{*}$ and may
be used in its place in the complexity bounds cited
above~\cite{Low_2026}; the resulting resource estimate then depends on the value
used.

\section{Exactly Solvable Model}
\label{sec:SHO-delta}

Before turning to electronic Hamiltonians, we analyze the
one-dimensional harmonic oscillator with a contact
interaction~\cite{avakian_spectroscopy_1987, patil_harmonic_2006,viana-gomes_solution_2011}.
Three properties make this model a controlled setting for the TC analysis:
(i) it is exactly solvable and, for any finite contact strength, has a simple
spectrum, so the target problem has no degenerate-eigenspace ambiguity; (ii) the
Hamiltonian preserves parity, so the even- and odd-parity sectors are decoupled;
and (iii) the contact interaction leaves the odd sector unchanged but generates
a cusp in the even sector. Since the parity sectors are decoupled and the cusp
is confined to the even sector, the exactly known odd harmonic-oscillator sector
can be set aside before constructing the TC representation. The analysis can
then focus on the interacting even sector, where the simple spectrum permits
direct examination of how cusp removal affects finite-basis convergence and
eigenvector conditioning.

In dimensionless harmonic-oscillator units, the model is described by the
Hamiltonian
\begin{equation}
\label{eq:sho-hamiltonian}
\hat H = -\frac{1}{2}\frac{d^2}{dx^2} + \frac{1}{2}x^2 + g\,\delta(x),
\end{equation}
where $g$ denotes the strength of the delta potential.  The exact wavefunction
$\Psi(x)$ satisfies the jump condition at the origin: 
\begin{equation}
    \left.\frac{d\Psi}{dx}\right|_{0^+}-\left.\frac{d\Psi}{dx}\right|_{0^-}
=2g\,\Psi(0).
\end{equation}
Since the potential is even in $x$, eigenfunctions have definite parity. In the
even-parity sector, this reduces to
\begin{equation}
\left.\frac{d\Psi}{d|x|}\right|_{|x|=0^+}=g\,\Psi(0).
\end{equation}

Because the odd-parity eigenfunctions vanish at the origin, they are not
affected by the delta potential and retain the simple-harmonic-oscillator
energies
\begin{equation}
E_n^{\mathrm{odd}}=2n+\frac{3}{2}, \qquad n=0,1,2,\dots .
\end{equation}
The physical Hamiltonian therefore decomposes as $\hat H=\hat
H_{\mathrm{even}}\oplus\hat H_{\mathrm{SHO}}|_{\mathcal H_{\mathrm{odd}}}$,
where $\hat H_{\mathrm{SHO}}=-\tfrac12 d^2/dx^2+\tfrac12x^2$. The odd block is
exactly known and decoupled from the even sector: it is Hermitian and diagonal
in the odd SHO basis. We remove this block before introducing the TC transformation and
restrict the subsequent finite-basis analysis to the even sector.

The even-parity states are shifted by the delta potential, and their energies
$E_n^{\mathrm{even}}$ are determined by the transcendental equation
\begin{equation}
\label{eq:sho-even-spectrum}
\frac{\Gamma\!\left(\frac{3}{4}-\frac{E_n^{\mathrm{even}}}{2}\right)}
     {\Gamma\!\left(\frac{1}{4}-\frac{E_n^{\mathrm{even}}}{2}\right)}
= -\frac{g}{2},
\qquad n=0,1,2,\dots ,
\end{equation}
where $\Gamma(\cdot)$ denotes the Gamma function.

Variants of this model have been studied both without and with the TC
transformation. Grining \emph{et al.}~\cite{grining_many_2015} studied the
corresponding many-fermion Hamiltonian with the same harmonic confinement and
contact interaction, without TC, using a one-particle SHO basis. Jeszenszki
\emph{et al.}~\cite{jeszenszki_accelerating_2018} instead applied TC to a
homogeneous gas with no external potential ($V=0$), represented in a plane-wave
basis. Both many-particle settings can contain degenerate eigenspaces, adding
basis freedom within each degenerate eigenspace. By contrast, the simple
spectrum of our model at finite $g$ provides a controlled setting for analyzing
both basis-set convergence and eigenvector conditioning.

\subsection{Projected Transcorrelated Hamiltonian}
\label{sec:projected-sho-tc}
\subsubsection{Standard method}
To analyze basis truncation in a controlled way, we project the Hamiltonian in
Eq.~\eqref{eq:sho-hamiltonian} onto a finite-dimensional subspace. For this
model, the natural choice of basis is the eigenbasis of the one-dimensional
simple harmonic oscillator (SHO):
\begin{equation}
    \phi_p(x)=\frac{1}{\sqrt{2^p p!\sqrt{\pi}}}H_p(x)e^{-x^2/2},
    \qquad p=0,1,2,\ldots,
\end{equation}
where $H_p(x)$ is the $p$th Hermite polynomial. Since $\phi_p(0)=0$ for odd
$p$, the contact interaction acts only in the even-parity block. We therefore
retain the $N_{\mathrm{even}}$-dimensional space
$\mathcal B_{N_{\mathrm{even}}}^{(+)}
=\operatorname{span}\{\phi_{2n}:n=0,\ldots,N_{\mathrm{even}}-1\}$ and, throughout
this subsection, use $m,n=0,\ldots,N_{\mathrm{even}}-1$ to index this retained
basis.

The bare Hamiltonian then has matrix elements
\begin{equation}
    H_{mn}^{(\mathrm{even})}
    =\matrixel{\phi_{2m}}{\hat H}{\phi_{2n}}
    =\left(2n+\frac12\right)\delta_{mn}
    +g\,\phi_{2m}(0)\phi_{2n}(0).
\end{equation}

\subsubsection{Transcorrelated method}
On this even-parity space, the correlator
$\hat\tau_{\boldsymbol\theta}$ for the 1D harmonic oscillator with a contact
interaction is the multiplication operator generated by a real-valued
even scalar function $u(x)$, so the transformation preserves parity. The
correlation function $u(x)$ is chosen so that
the delta-function term in Eq.~\eqref{eq:sho-hamiltonian} is
canceled in the transformed Hamiltonian. 
This requirement does not specify $u(x)$ uniquely. Different admissible choices
can produce different basis-set convergence behavior
(Sec.~\ref{sec:finite-dimension-tc}) and eigenvector conditioning
(Sec.~\ref{sec:cond}) in the projected matrices
$H_{\mathrm{TC}}^{\mathcal B}(\boldsymbol\theta)$.
To assess these dependencies, we compare the exponential form of
Ref.~\cite{baiardi_explicitly_2022} and the $\operatorname{erfc}$-based form of
Ref.~\cite{giner2021new}:
\begin{align}
    u_1(x) &= g |x| e^{-\mu |x|}, \label{eq:jastrow_exp-1} \\
    u_2(x) &= g \left[ |x| \operatorname{erfc}(\mu |x|)
    - \frac{1}{\sqrt{\pi}\mu} e^{-(\mu x)^2} \right].
    \label{eq:jastrow_erf-1}
\end{align}
We refer to the TC Hamiltonians generated by $u_1$
and $u_2$ as TC-exp and TC-erfc, respectively. The two correlation functions
satisfy $u_j'(0^+)=g$, which gives the delta-function cancellation used below.
At fixed $g$, each depends only on the positive range parameter $\mu$. Their
explicit first and second derivatives reduce the required SHO matrix elements
to the formulas derived in Appendix~\ref{app:sho-integrals}.

Substitution into Eq.~\eqref{eq:Htc-commutator} together with
Eq.~\eqref{eq:sho-hamiltonian} yields
\begin{equation}
\hat H_{\mathrm{TC}}
=
-\frac12\frac{d^2}{dx^2}
+\frac12x^2
+g\delta(x)
-\frac12 u''(x)
-u'(x)\frac{d}{dx}
-\frac12 \bigl(u'(x)\bigr)^2.
\label{eq:sho_htc_explicit}
\end{equation}

\begin{samepage}
We can further rewrite Eq.~\eqref{eq:sho_htc_explicit} in the same structural
form introduced in Eq.~\eqref{eq:TC-first-order-form}.
We define
\begin{equation}
\begin{aligned}
F(x)&\equiv u'(x),\\
V_{\mathrm{TC,eff}}(x)
&\equiv g\delta(x)-\frac12u''(x)-\frac12F(x)^2\\
&=-\frac12u''_{\mathrm{reg}}(x)-\frac12F(x)^2,\\
V_{\mathrm{eff}}(x)
&\equiv \frac12x^2+V_{\mathrm{TC,eff}}(x).
\end{aligned}
\label{eq:sho_Veff_total_def}
\end{equation}

Here $u''=u''_{\mathrm{reg}}+2g\delta$ in the distributional sense, so
the contact term cancels the distributional part of $-\tfrac12u''$. Because
we use $\hat H_{\mathrm{SHO}}$ as the reference operator in this section,
$V_{\mathrm{TC,eff}}$ denotes the regular scalar correction that remains after
this exact cancellation, whereas $V_{\mathrm{eff}}$ is the complete
transformed scalar potential.
Then Eq.~\eqref{eq:sho_htc_explicit} becomes
\begin{equation}
    \label{eq:sho_htc_compact}
    \hat H_{\mathrm{TC}}
    =\hat H_{\mathrm{SHO}}
      -F(x)\frac{d}{dx}
      +V_{\mathrm{TC,eff}}(x).
\end{equation}

Here $\hat H_{\mathrm{SHO}}=\hat T+\tfrac12x^2$.
\end{samepage}

Projecting Eq.~\eqref{eq:sho_htc_compact} onto the retained even SHO basis gives
\begin{equation}
\begin{aligned}
    (H_{\mathrm{TC}})_{mn}
    &=
    \left(2n+\frac12\right)\delta_{mn}
    - \matrixel{\phi_{2m}}{F(x)\frac{d}{dx}}{\phi_{2n}}
    + \matrixel{\phi_{2m}}{V_{\mathrm{TC,eff}}(x)}{\phi_{2n}}.
\end{aligned}
    \label{eq:sho_tc_matrix}
\end{equation}
The matrix elements of the regular scalar correction are therefore
\begin{equation}
\begin{aligned}
    V_{\mathrm{TC,eff},mn}
    &\equiv
    \matrixel{\phi_{2m}}{V_{\mathrm{TC,eff}}(x)}{\phi_{2n}}
    =\int_{-\infty}^{\infty}
    \phi_{2m}(x)\phi_{2n}(x)\,V_{\mathrm{TC,eff}}(x)\,dx.
\end{aligned}
\end{equation}
For the first-order differential operator, it is convenient to use the standard
derivative relation for harmonic-oscillator basis functions,
\begin{equation}
    \frac{d}{dx}\phi_p(x)
    =
    \sqrt{\frac{p}{2}}\,\phi_{p-1}(x)
    -
    \sqrt{\frac{p+1}{2}}\,\phi_{p+1}(x).
\end{equation}
This reduces the corresponding matrix elements to linear combinations of
multiplicative matrix elements of $F(x)$:
\begin{equation}
\begin{split}
D_{mn} \equiv & -\matrixel{\phi_{2m}}{F(x)\frac{d}{dx}}{\phi_{2n}} \\
= & -\sqrt{n}\,\matrixel{\phi_{2m}}{F(x)}{\phi_{2n-1}} \\
& + \sqrt{\frac{2n+1}{2}}\,\matrixel{\phi_{2m}}{F(x)}{\phi_{2n+1}}.
\label{eq:sho_Dmn}
\end{split}
\end{equation}
Hence, the projected transcorrelated Hamiltonian takes the form
\begin{equation}
\begin{aligned}
    (H_{\mathrm{TC}})_{mn}
    &=
    \left(2n+\frac12\right)\delta_{mn} + D_{mn}
    +V_{\mathrm{TC,eff},mn}.
\end{aligned}
\end{equation}

\subsubsection{Structure of the Correlation Functions}
The two correlation functions $u_1$ and $u_2$ defined in
Eqs.~\eqref{eq:jastrow_exp-1} and~\eqref{eq:jastrow_erf-1} both satisfy the
delta-function cancellation condition. However, this does not determine the full structure of the
transcorrelated Hamiltonian.  To compare the two choices at the operator level,
we focus on the position-dependent coefficient $F(x) = u'(x)$ in
Eq.~\eqref{eq:sho_htc_compact}. The first-order  differential operator $-F(x)\,d/dx$ is the source of
non-Hermiticity in the transformed Hamiltonian, so $F(x)$
captures the spatial range over which non-Hermiticity is
introduced.

\begin{figure}[ht!]
    \centering
    \includegraphics[width=\linewidth]{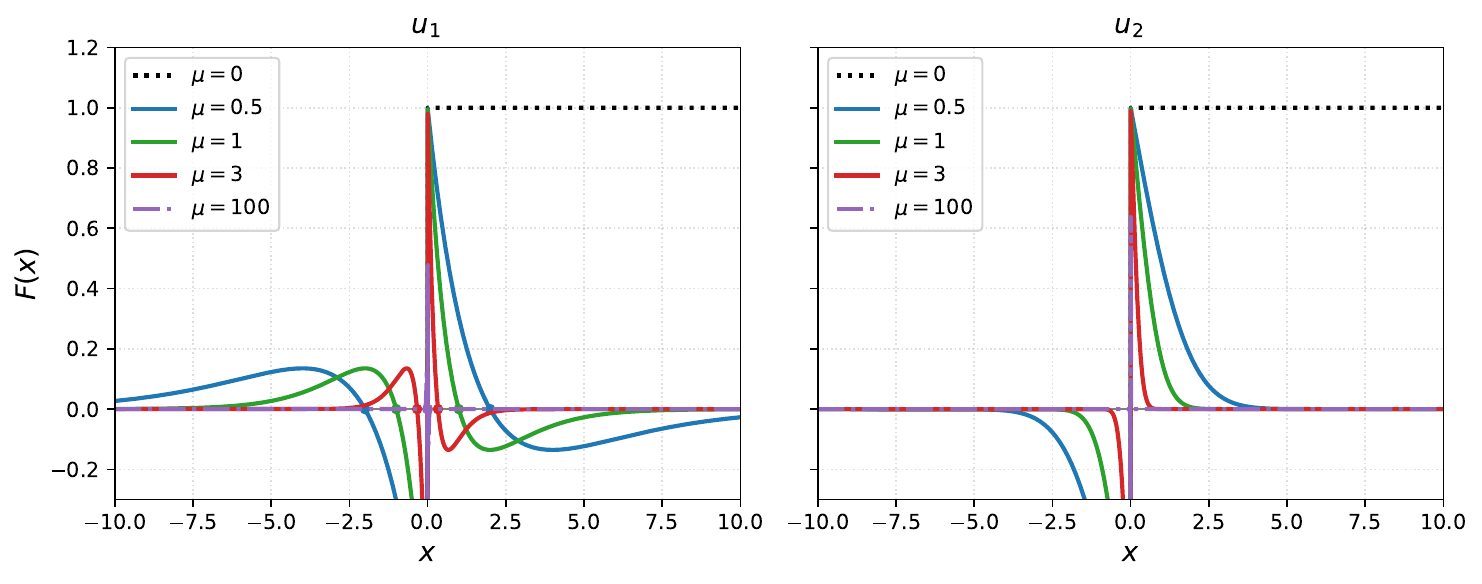}
    \caption{Comparison of the position-dependent coefficient $F(x)=u'(x)$ for the
    two correlation functions (left: $u_1$, right: $u_2$) at $g=1$. The
    parameter $\mu$ controls the effective real-space range over which $F(x)$
    has non-negligible magnitude: smaller $\mu$ produces a broader
    non-Hermitian contribution, whereas larger $\mu$ localizes the contribution
    near the origin. The plotted $\mu=0$ curve denotes the $\mu\to0^+$ limit of
    $F(x)$.}

    \label{fig:F_comparison}
\end{figure}

Figure~\ref{fig:F_comparison} compares these ranges directly. The following
results examine how the resulting correlator-dependent operator structures are
reflected in the finite-dimensional Hamiltonians.

\subsection{Results for the 1D harmonic oscillator with a contact interaction}
The real-space comparison above isolates the operator-level distinction between
the two correlators. In a finite calculation, however, this distinction is
not accessed directly from $F(x)$ itself, but through the projected matrix
elements of the full transcorrelated Hamiltonian. Therefore, we now turn to the
finite even-parity SHO representation and examine how the two correlator choices
affect the resulting projected spectra and eigenvector conditioning.

For the numerical studies in this section, we fix $g=1$, corresponding to
a repulsive contact interaction that produces a one-dimensional analogue of the
electron--electron cusp generated by Coulomb repulsion. For every reported
combination of basis size and Jastrow parameter $\mu$, the projected Hamiltonian
matrix has real entries, and all computed eigenvalues are real and distinct. In this setting, we
define the ground-state energy as the lowest eigenvalue and the ground state as
its associated right eigenvector. The same model can also be used to study
excited states; the present analysis focuses on the ground state.

Figure~\ref{fig:convergence_tc} analyzes the convergence behavior of the
absolute energy error, $\Delta E = |E_0^{(N)}-E_0|$, for the lowest even-parity
state at $g=1$. Here, $N\equiv N_{\mathrm{even}}$, and $E_0^{(N)}$ denotes the lowest eigenvalue obtained by
projecting the Hamiltonian onto a truncated basis of $N_{\text{even}}$
even-parity harmonic-oscillator functions, and $E_0$ represents the exact
ground-state energy given by the $n=0$ solution of
Eq.~\eqref{eq:sho-even-spectrum}. The approximately linear
large-$N_{\text{even}}$ behavior on the log--log scale is consistent with a
power-law decay of the error $\Delta E \sim N_{\text{even}}^s$,
i.e., algebraic convergence with exponent
$s$, which is negative for a decaying error (the insets report, e.g.,
$s\approx-0.5$ bare and $s\approx-1.5$ for TC). The fitted slopes $s$,
extracted from the large-$N_{\text{even}}$ regime ($N_{\text{even}} > 450$),
therefore provide numerical estimates of the convergence rates for the bare
Hamiltonian and both projected TC Hamiltonians. The noticeable bend in the
TC-erfc curve is consistent with the non-Hermitian character of the projected
TC Hamiltonian, for which the variational monotonicity of Hermitian Ritz values
is not guaranteed, so $E_0^{(N)}$ may cross
$E_0$ and approach it from below, and $\Delta E$ need not decrease
monotonically.

\begin{figure}[ht!]
    \centering
    \includegraphics[width=\linewidth]{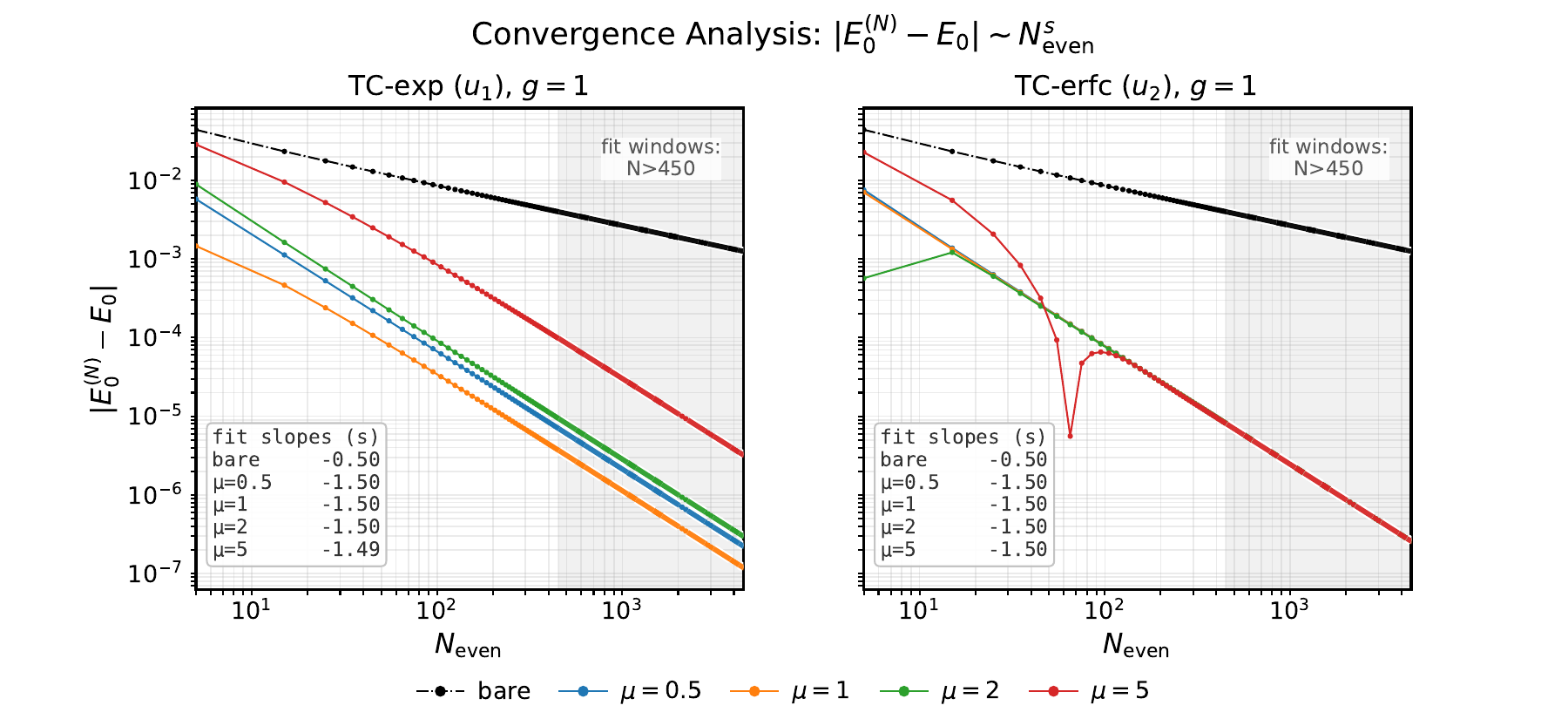}
    \caption{Convergence analysis of the absolute error in the lowest
    even-parity energy, $|E_0^{(N)} - E_0|$, for the one-dimensional harmonic
    oscillator with a contact interaction at $g=1$. The error is plotted as a
    function of the number of even-parity basis functions $N_{\text{even}}$ on a
    log-log scale. The left and right panels show the TC-exp and TC-erfc results generated by the
    correlation functions $u_1$ and $u_2$, respectively. The bare projected Hamiltonian is
    shown as a black dash-dotted line, while the transcorrelated results for
    various $\mu$ values are plotted as solid colored lines. The gray shaded
    regions denote the fit window ($N_{\text{even}} > 450$) used to extract the
    power-law exponent $s$ reported in the insets.}
    \label{fig:convergence_tc}
\end{figure}

The log-log fits in Fig.~\ref{fig:convergence_tc} estimate the convergence
rate, but they do not by themselves separate the rate from the amplitude (prefactor) of the asymptotic tail. For example, the large-$N_{\mathrm{even}}$ TC-erfc curves
nearly collapse as $\mu$ is varied, while the TC-exp curves retain a
 vertical separation dependent on $\mu$. This indicates that the two correlator
forms differ in the amplitude of the asymptotic error rather than in the
leading power itself. To further expose this second piece of information, we plot the
rescaled error 
\[
    A_j(\mu;N):=N_{\mathrm{even}}^{3/2}|E_0^{(N)}-E_0|.
\]
When the asymptotic $N_{\mathrm{even}}^{-3/2}$ regime is reached, this quantity
approaches the corresponding leading prefactor. The analysis in Fig.~\ref{fig:prefactor_tc} connects the common $N_{\mathrm{even}}^{-3/2}$ behavior to cusp cancellation and relates the correlator dependence of the leading explicit-tail
amplitude to the coefficient in the local expansion of the exact
right transcorrelated eigenfunction.

\begin{figure}[ht!]
    \centering
    \includegraphics[width=\linewidth]{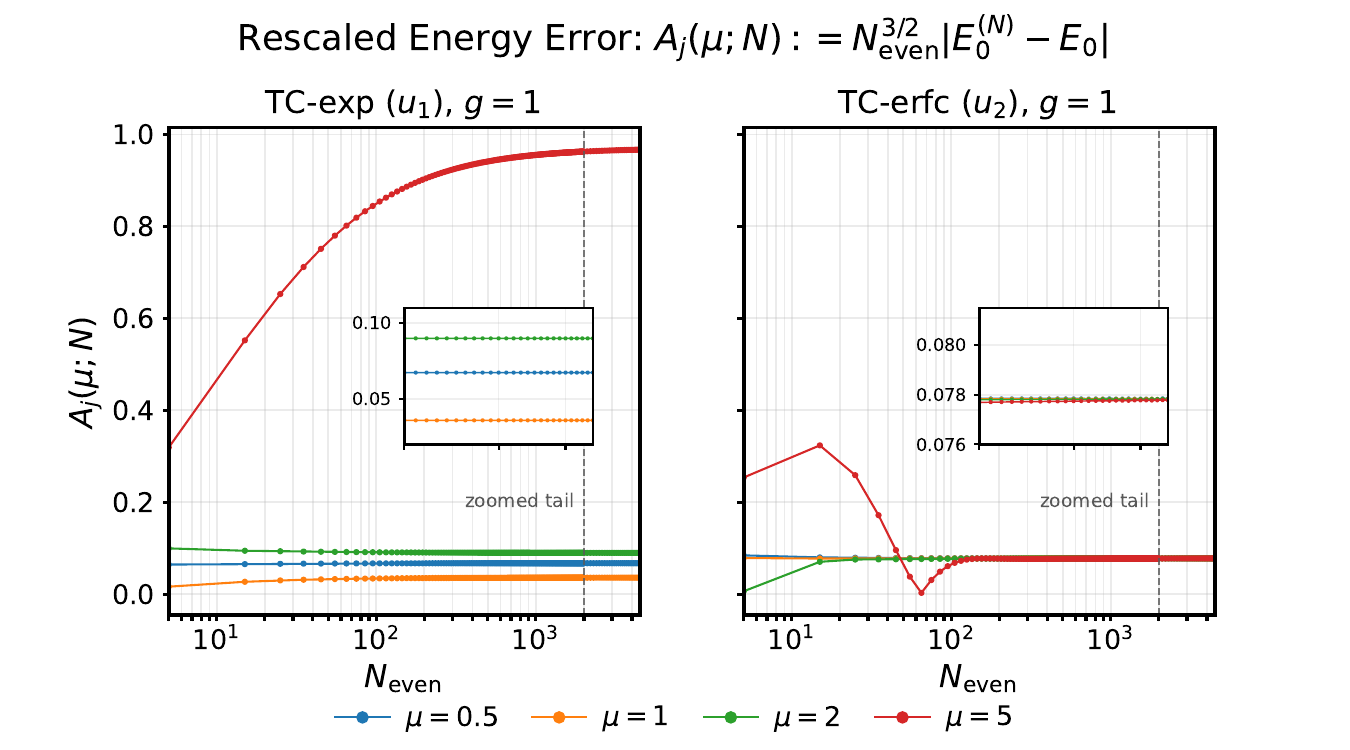}
    \caption{Rescaled energy error
    $A_j(\mu;N)=N_{\mathrm{even}}^{3/2}|E_0^{(N)}-E_0|$ for TC-exp and
    TC-erfc at $g=1$. The TC-erfc curves tend to a common asymptotic plateau,
    consistent with a $\mu$-independent leading error amplitude. The TC-exp
    plateaus depend strongly on $\mu$.}
    \label{fig:prefactor_tc}
\end{figure}

Indeed, both the generic convergence rate and the 
leading asymptotic prefactor can be derived analytically as follow:
\begin{mainresult}
Under the reduced-inverse stability assumption stated in
Appendix~\ref{app:sho-delta-details}, for fixed $\mu>0$ the projected TC
ground-state energy satisfies
\[
 E_0^{(N)}-E_0=-\frac{g\Psi_0(0)^2\beta_j(\mu)}{\pi}N_{\mathrm{even}}^{-3/2}
 +O(N_{\mathrm{even}}^{-2}),
\]
where $\Psi_0$ is the exact ground state and
\[
 \beta_{\exp}(\mu)=\frac{g(2E_0+g^2)}{3}-\frac{g\mu^2}{2},
 \qquad
 \beta_{\mathrm{erfc}}=\frac{g(2E_0+g^2)}{3}.
\]
Hence the asymptotic absolute-error prefactor is
$A_j(\mu)=|g|\Psi_0(0)^2|\beta_j(\mu)|/\pi$.
The full derivation is given in Appendix~\ref{app:sho-delta-details}.
\end{mainresult}

The results above and the conditional asymptotic analysis characterize the favorable side of the transcorrelated
projection: once the cusp has been incorporated into the operator, the
energy-error convergence improves from $N_{\mathrm{even}}^{-1/2}$ to $N_{\mathrm{even}}^{-3/2}$, and the
remaining correlator dependence is pushed into the asymptotic prefactor. This
improvement, however, is obtained through a non-unitary similarity
transformation. The resulting Hamiltonian is therefore 
non-Hermitian, and energy convergence alone does not determine whether the
projected problem is numerically tractable. The corresponding cost is measured
by the non-orthogonality of the eigenvectors. We therefore turn to the global
condition number $\kappa_2(V)$ of the diagonalizing matrix.

\begin{figure}[ht!]
    \centering
    \includegraphics[width=\linewidth]{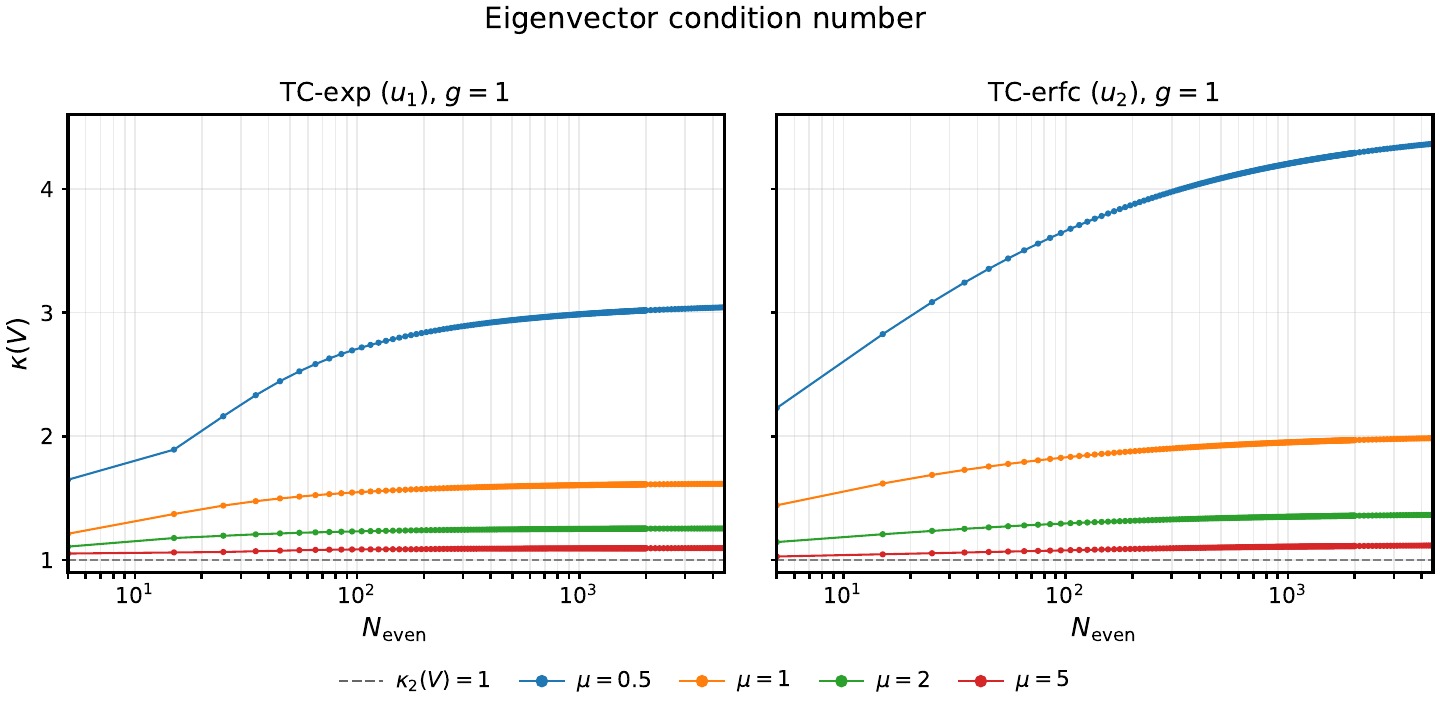}
    \caption{Global eigenvector condition number $\kappa_2(V)$ for the projected
    transcorrelated Hamiltonian at $g=1$.}
    \label{fig:condition_tc}
\end{figure}

For the nondegenerate projected spectrum considered here, unit-$2$-norm
normalization removes the eigenvector-scaling ambiguity, making $\kappa_2(V)$
unique. Under this convention, Figure~\ref{fig:condition_tc} shows that the
projected transcorrelated matrices remain moderately conditioned over the full
basis range studied here. The data suggest slower-than-linear growth 
of $\kappa_2(V)$ with $N_{\mathrm{even}}$ over this range, although they do not
establish its asymptotic scaling. Furthermore, the conditioning differs between the two correlator forms. For each correlator form, the decrease in $\kappa_2(V)$ with increasing $\mu$ mirrors the real-space localization of $F(x)$ shown in
Fig.~\ref{fig:F_comparison}: as $F(x)$ becomes more localized, the projected
matrices become better conditioned. At fixed $\mu$, the TC-erfc matrices
generally have larger condition numbers than the TC-exp matrices.

Taken together, these results show that the most favorable convergence prefactor and the smallest eigenvector condition number do not necessarily occur for the same correlator form at the same value of $\mu$. The solvable model already exhibits the central trade-off that will reappear in the molecular setting: faster basis convergence can come with a conditioning cost, and that cost depends sensitively on the correlator form and its range parameter.

\subsection{Query Complexity for Solvable Model}
Having identified this trade-off, we now ask when the faster basis convergence
translates into an improved asymptotic query bound for estimating the exact
ground-state energy $E_0$ to accuracy $\varepsilon$ using QEVE.  Let $E_{0,\mathcal R}^{(N)}$ denote the eigenvalue of the projected
Hamiltonian $H_{\mathcal R}^{(N)}$, where
$\mathcal R\in\{\mathrm{bare},\mathrm{TC}\}$, and let $\widehat E_0$
denote the value estimated by the quantum algorithm.
The total error satisfies
\begin{equation}
    |\widehat E_0-E_0|
    \leq
    \underbrace{|E_{0,\mathcal R}^{(N)}-E_0|}_{\delta_{\mathcal R}(N)}
    +
    \underbrace{|\widehat E_0-E_{0,\mathcal R}^{(N)}|}_{\leq\varepsilon_{\mathrm{alg}}}.
    \label{eq:sho-total-error}
\end{equation}

For a fixed $\theta\in(0,1)$, we choose $N$ such that
$\delta_{\mathcal R}(N)\leq\theta\varepsilon$ and set
$\varepsilon_{\mathrm{alg}}=(1-\theta)\varepsilon$,
ensuring \mbox{$|\widehat E_0-E_0|\leq\varepsilon$}.
Assume that the projected Hamiltonian matrices used in the estimation
are diagonalizable with real spectra.  Under the same assumptions and the oracle-query model of Theorem~3 of Ref.~\cite{Low_2026},   
QEVE requires
$\widetilde{\mathcal O}
(\alpha_{\mathcal R,N}\kappa_{\mathcal R,N}/
\varepsilon_{\mathrm{alg}})$ oracle queries at fixed success probability.
Here $\kappa_{\mathcal R,N}$ is an available upper bound on the
optimal eigenvector condition number.

For the correlators considered here, at fixed finite $g$ and
$0<\mu<\infty$, the SHO contribution has norm $\mathcal O(N)$, while the projected bare or TC correction has norm
$\mathcal O(\sqrt N)$. Hence,
$\|H_{\mathcal R,N}\|=\mathcal O(N)$ for both representations. Within the above oracle query model, we can take $\alpha_{\mathcal R,N}=\mathcal O(N)$.
For either representation, suppose that
$\delta_{\mathcal R}(N)=\mathcal O(N^{-p})$ and
$\kappa_{\mathcal R,N}=\mathcal O(N^a)$, with $p>0$ and $a\geq0$.
A sufficient basis size is then $N=\mathcal O(\varepsilon^{-1/p})$,
which gives
\begin{equation}
    Q_{\mathcal R}(\varepsilon)
    =\widetilde{\mathcal O}
    \left(C(\theta) \varepsilon^{-\left(1+\frac{1+a}{p}\right)}\right).
    \label{eq:sho-general-query-scaling}
\end{equation}
Under the basis-convergence assumptions above, the bare and TC
representations have $p=1/2$ and $p=3/2$, respectively.
With $\kappa_{\mathrm{bare},N}=1$ and fixed $\theta$, we absorb $C(\theta)$ into the implicit
constants and obtain
\begin{equation}
    Q_{\mathrm{bare}}(\varepsilon)
    =\widetilde{\mathcal O}(\varepsilon^{-3}),
    \qquad
    Q_{\mathrm{TC}}(\varepsilon)
    =\widetilde{\mathcal O}(\varepsilon^{-(5+2a)/3}).
    \label{eq:sho-conditional-query-scaling}
\end{equation}
For TC, bounded conditioning ($a=0$) and linear growth in $N$ ($a=1$)
give $\widetilde{\mathcal O}(\varepsilon^{-5/3})$ and
$\widetilde{\mathcal O}(\varepsilon^{-7/3})$, respectively.
When $a<2$, TC with QEVE yields a polynomial improvement in the query-complexity upper bound over QPE applied to the bare Hamiltonian.

Overall, the solvable model provides a clean and controlled setting for
understanding how the accuracy--cost trade-off of the transcorrelated method
arises, allowing us to address the three questions posed in
Sec.~\ref{sec:intro}.

\section{Application to Electronic Hamiltonians}
\label{sec:molecular}
We now turn to second-quantized electronic Hamiltonians to examine how the
mechanisms identified in the solvable model appear in electronic structure
calculations and what implications they have for quantum algorithms.
We begin with the two-electron He atom, for which the
three-body TC term is absent, and first examine basis-set convergence and basis
reduction. We then study the accuracy--cost--conditioning trade-off
with the basis fixed at cc-pVDZ before turning to larger basis sets. Finally, we
move beyond two electrons to Be and LiH, where the full TC
Hamiltonian contains three-body contributions and can be compared with the xTC
approximation.

\subsection{Computational Details}
All molecular calculations were performed with
PySCF~\cite{sun_recent_2020}. The molecular pair correlators were defined by the analytic
forms in Eqs.~\eqref{eq:jastrow_exp-1} and \eqref{eq:jastrow_erf-1}, with $|x|$
replaced by the interelectronic distance $r_{ij}$ and $g=1/2$ to impose the
electron--electron cusp~\cite{baiardi_explicitly_2022,giner2021new}.
The corresponding molecular Hamiltonians are denoted TC-exp and TC-erfc,
respectively. Following
Ref.~\cite{cohen_similarity_2019}, the additional TC contributions were
evaluated directly in the canonical RHF orbital basis by numerical real-space
quadrature on PySCF atom-centered DFT grids at level~5 with default pruning. The
xTC calculations for Be and LiH followed Ref.~\cite{christlmaier_xtc_2023}.
The fermionic Hamiltonians were mapped to qubits using the Jordan--Wigner
transformation. 
For $\kappa(V_{\mathrm{os}})$, 
eigenvalues were initially grouped into
numerical clusters using a relative tolerance of $10^{-8}$, and orthonormal cluster bases were constructed by
singular-value decomposition with a relative rank threshold of $10^{-12}$.
For real-spectrum calculations, the Hamiltonian restricted to each cluster
subspace was rediagonalized. The resulting eigenvectors were orthonormalized
only within subclusters defined by a stricter relative degeneracy tolerance.
The final numerical eigenvector matrix was checked through the residual of the
eigenvalue equation, the off-diagonal part of the transformed Hamiltonian, and
agreement with the original eigenvalues.

\subsection{Basis-set convergence and basis reduction}
\label{sec:basis-conv}
We first isolate the accuracy gain provided by the transcorrelated
transformation. For the helium atom, we compare bare- and TC-FCI ground-state
energies across the cc-pVXZ basis set sequence.

\begin{figure}[!htbp]
    \centering
    \includegraphics[width=\linewidth]{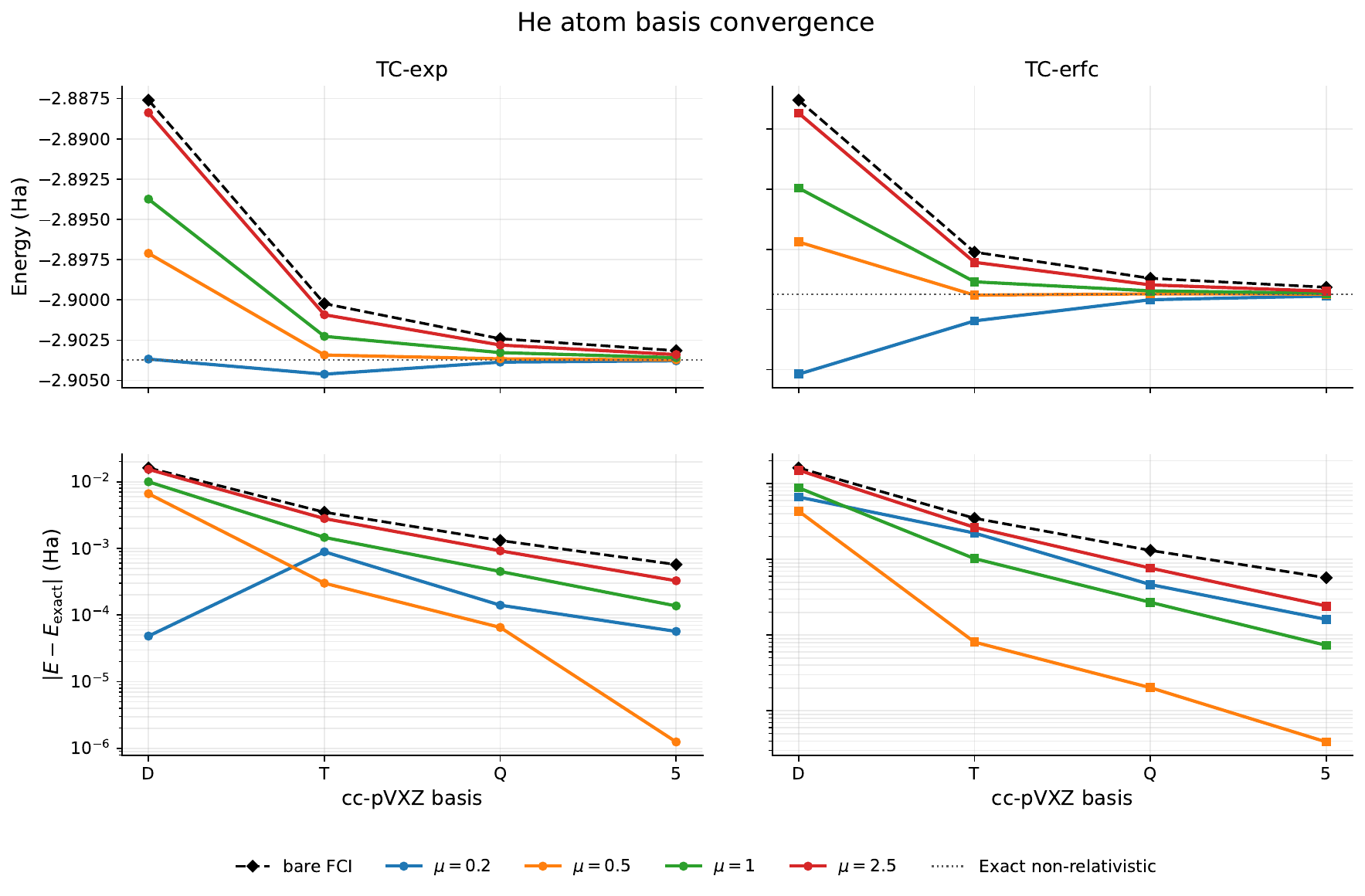}
        \caption{Basis-set convergence of the helium atom ground-state energy using the TC method.  The top panels show the total energies and the bottom panels show the absolute errors relative to the exact non-relativistic CBS energy as a function of the cardinal number $X$ of the cc-pV$X$Z basis sets. Results are shown for TC-exp ($u_1$; left) and TC-erfc ($u_2$; right) across various values of $\mu$. The bare FCI results (dashed black lines with diamonds) and the exact limit (dotted horizontal lines, obtained from \cite{PhysRevA.44.7071}) are included for reference.}
    \label{fig:he-u1u2-convergence-plot}
\end{figure}
Figure~\ref{fig:he-u1u2-convergence-plot} reports the all-electron TC
ground-state energies of helium and their absolute errors relative to the exact
non-relativistic reference value of Ref.~\cite{PhysRevA.44.7071} for the
cc-pV$X$Z ($X\in\{\mathrm{D},\mathrm{T},\mathrm{Q},5\}$) basis sets. The TC Hamiltonian
is non-Hermitian, so its finite-basis eigenvalues need not converge
monotonically. For example, the $\mu=0.2$ curve crosses below the CBS energy
before approaching it from below. We therefore use the absolute error to assess
convergence. At this level, the TC transformation has two complementary
effects: it reduces the absolute error at fixed cardinal number $X$, and that
error decreases more rapidly across the cc-pV$X$Z sequence.
However, because the
underlying one-particle spaces are not nested (that is, the cc-pVXZ basis sets are not constructed by successively adding functions to a fixed smaller basis), this sequence provides an
empirical basis-set comparison rather than the controlled truncation analysis
available for the solvable model. The remaining question is whether these gains
persist when cost-relevant diagnostics of the second-quantized, non-normal TC
Hamiltonian are taken into account.

\FloatBarrier

\subsection{Accuracy--cost diagnostics}
The preceding results quantify the accuracy gain associated with basis
reduction. Such a reduction is attractive for quantum algorithms: in standard second-quantized encodings, the number of spin-orbitals directly determines the number of logical qubits required to encode the quantum state of the system. 
Although this state-register size alone does not determine the total logical-qubit requirement (ancilla registers required for algorithm implementation can contribute substantially \cite{leeEven2021,lowDenser2026}), reducing the orbital space can still yield a relevant resource advantage.
This reduction can also lead to lower Hamiltonian-representation costs, although these remain algorithm- and implementation-dependent.
However, the
finite-basis TC Hamiltonian is non-normal. This non-normality may introduce
compensating costs. 
We therefore assess this gain together with
cost-relevant properties of the transformed Hamiltonian.
\subsubsection{Fixed-basis cost-relevant diagnostics}
To analyze this trade-off, we scan the correlator parameter $\mu$ in the
cc-pVDZ basis and track the absolute energy error relative to the CBS limit
together with three cost-relevant diagnostics. This scan therefore
characterizes the fixed-basis accuracy--cost trade-off relative to bare
cc-pVDZ. The three diagnostics are as follows: (i) For a Pauli decomposition $\hat H=\sum_i c_i\hat
P_i$, the Pauli one-norm of the qubit Hamiltonian is $\lambda(H)=\sum_i |c_i|$, where the sum runs
over the non-identity Pauli terms. 
(The constant identity coefficient $c_0$ is
excluded, so that $\lambda$ measures only the non-constant part of $\hat H$).
The one-norm sets the normalization scale in Pauli-based linear combination of unitaries block encoding of the Hamiltonian. 
For algorithms using such an encoding (for example, Hamiltonian simulation or quantum phase estimation), the leading query cost typically scales linearly with
the normalization factor $\lambda$~\cite{PhysRevResearch.3.033127}. (ii) The overlap between the
Hartree--Fock determinant and the right TC ground state provides an
initial-state diagnostic. Throughout this work, the reference state is the
canonical Hartree--Fock determinant of the bare Hamiltonian as produced by
PySCF~\cite{sun_recent_2020} with default settings; it is not re-optimized for
$\hat H_{\mathrm{TC}}$. This choice provides a uniform baseline across $\mu$.
A systematic study of re-optimized references lies beyond the present scope.
(iii) Eigenvector conditioning quantifies the non-normal amplification relevant
to non-Hermitian eigenvalue algorithms. Figure~\ref{fig:new-tc_summary-1}
reports the energy error, Hartree--Fock overlap, and one-norm diagnostics; the
conditioning diagnostic is analyzed separately below.

\begin{figure}[!htbp]
    \centering
    \includegraphics[width=\linewidth]{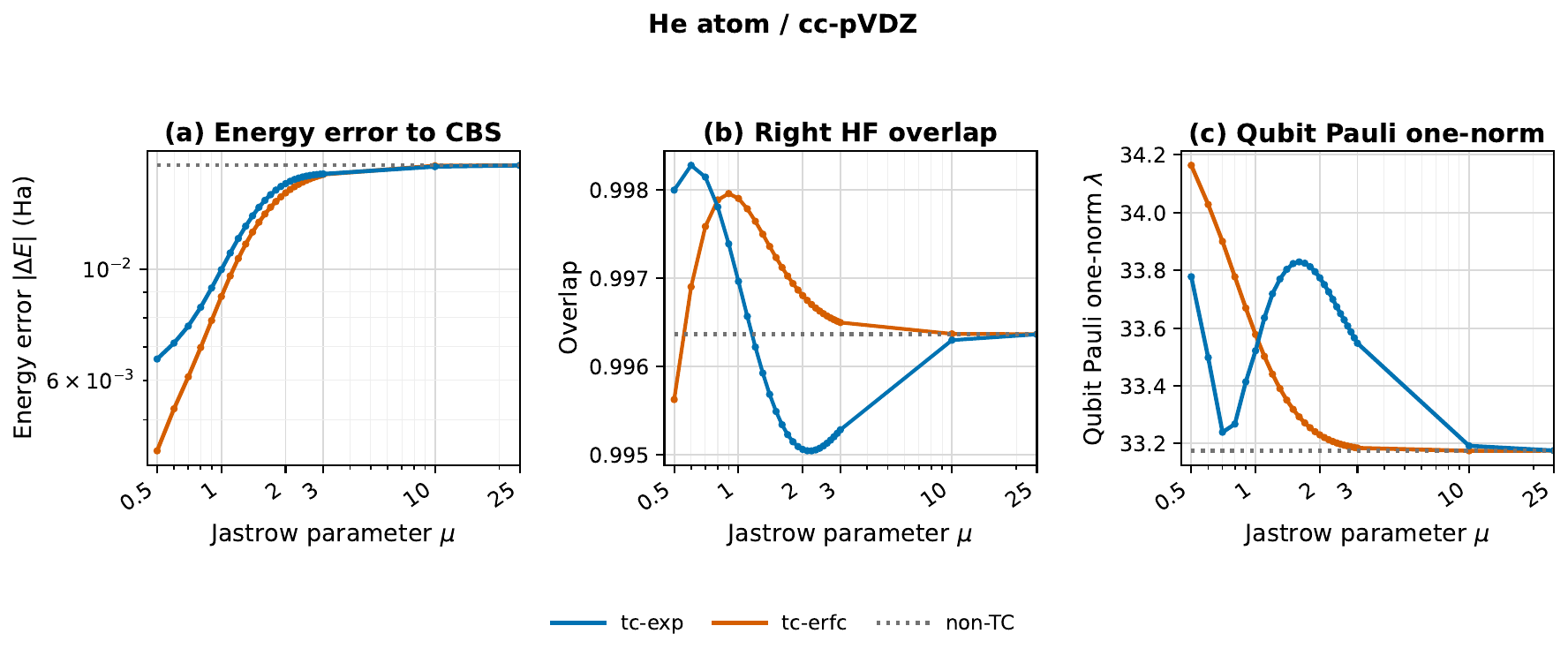}
        \caption{He/cc-pVDZ diagnostics versus the correlator parameter~$\mu$. Blue and orange curves correspond to TC-exp and TC-erfc, respectively. (a) Absolute energy error with respect to the CBS limit. The dotted line indicates the non-TC cc-pVDZ reference. (b) Initial overlap of the right ground state of $\hat H_{\mathrm{TC}}$ with the Hartree--Fock reference state. (c) Qubit Pauli $1$-norm $\lambda$ as a norm-based cost metric.
        The dotted line represents the corresponding non-TC $1$-norm baseline. }
    \label{fig:new-tc_summary-1}
\end{figure}

Figure~\ref{fig:new-tc_summary-1} reports the energy error together with
two of the three cost-relevant diagnostics introduced above. These quantities
display distinct dependences on the correlator parameter $\mu$ and its analytic
form. Panel (a) shows that the absolute energy error decreases as $\mu$ is
reduced for both correlator forms. TC-erfc shows a more nearly monotonic
approach to the CBS limit as $\mu$ decreases. Panel (b) reports the
overlap between the non-TC Hartree--Fock determinant and the right ground state
of $\hat H_{\mathrm{TC}}$.
The overlap remains high across the scan. For TC-erfc, it increases
monotonically as $\mu$ is lowered from $3.0$ to $0.9$, reaches its maximum at
$\mu=0.9$, and then decreases as $\mu$ is lowered further to $0.5$.
Panel (c) shows that the qubit one-norm $\lambda(\hat H)$
matches the non-TC baseline at large $\mu$ for both correlator forms.
Decreasing $\mu$ extends the spatial range of the correlator and generally
increases the qubit one-norm.
The
detailed $\mu$-dependence varies between the two forms: $\lambda(\hat H)$
increases monotonically as $\mu$ decreases for TC-erfc, whereas it is
non-monotonic for TC-exp.
Overall, the energy error and the two cost-relevant
diagnostics respond differently to the correlator form and parameter.
Their combined comparison is given below after the conditioning measure
has been fixed.

We next examine the non-normal conditioning of the TC qubit Hamiltonian. As
discussed in Sec.~\ref{sec:cond}, unlike the nondegenerate solvable model,
unit-$2$-norm normalization alone does not fix the eigenvector condition number
when degenerate eigenspaces are present.
The basis within each
degenerate eigenspace must also be fixed. Here and below, $V_R$ denotes the
right-eigenvector matrix of the TC qubit Hamiltonian
returned by the SciPy non-Hermitian eigensolver, with each column normalized to
unit $2$-norm. We therefore
orthonormalize each degenerate eigenspace, collect the resulting vectors in
$V_{\mathrm{os}}$ as in Eq.~\eqref{eq:Vos}, and compute
$\kappa(V_{\mathrm{os}})$.

\begin{figure}[!htbp]
    \centering
    \includegraphics[width=\linewidth]{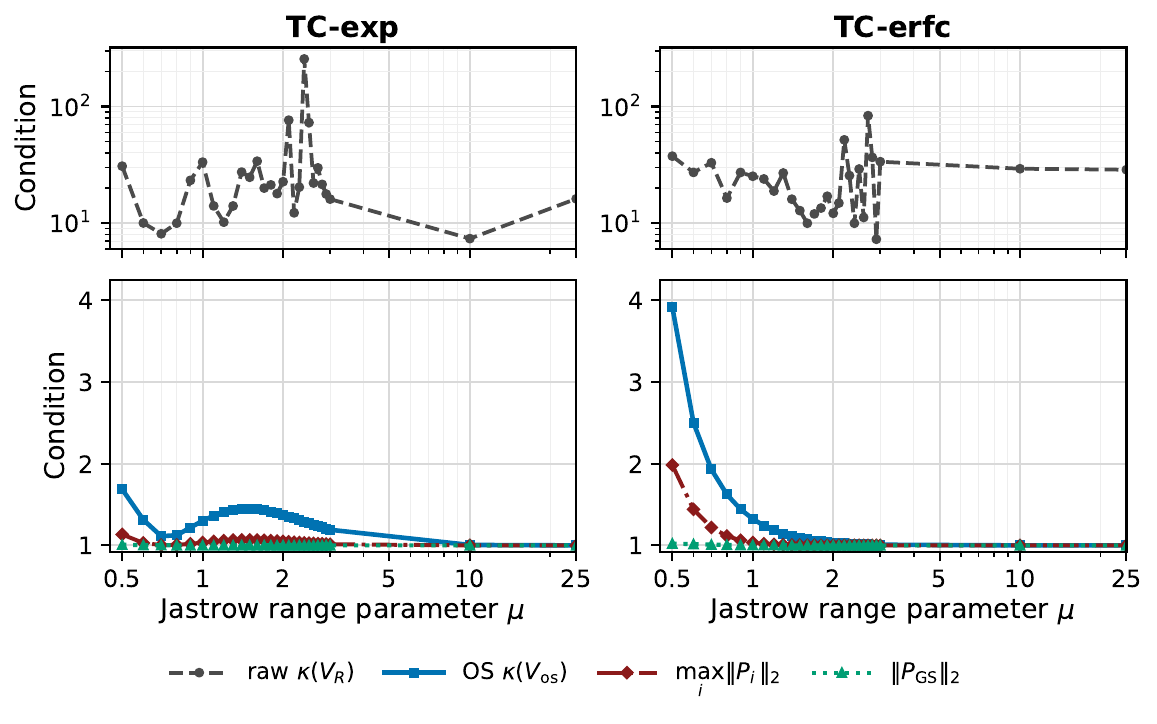}
    \caption{Condition numbers and spectral-projector norms versus the
    correlator parameter $\mu$ for TC-exp (left) and TC-erfc (right).
    Top (log scale): the condition number $\kappa(V_R)$ of the
    right-eigenvector matrix of the TC qubit Hamiltonian returned by the SciPy
    non-Hermitian eigensolver, with each column normalized to unit $2$-norm.
    Bottom (linear scale): the condition number $\kappa(V_{\mathrm{os}})$
    of the right-eigenvector matrix after orthonormalizing each
    degenerate subspace (blue, squares), the largest spectral projector
    norm $\max_i\|P_i\|_2$ (dark red, diamonds), and the ground-state
    projector norm $\|P_{\mathrm{GS}}\|_2$ (green, triangles).}
    \label{fig:new-tc_summary-2}
\end{figure}

Figure~\ref{fig:new-tc_summary-2} compares four conditioning quantities as
functions of $\mu$ for TC-exp and TC-erfc: the raw condition number
$\kappa(V_R)$, the eigenspace-orthonormalized condition number
$\kappa(V_{\mathrm{os}})$, the largest spectral-projector norm
$\max_i\|P_i\|_2$, and the ground-state projector norm
$\|P_{\mathrm{GS}}\|_2$. The upper panels show that $\kappa(V_R)$ has no clear
trend. After the basis within each degenerate eigenspace is orthonormalized,
$\kappa(V_{\mathrm{os}})$ in the lower panels is consistently smaller than
$\kappa(V_R)$ and shows a clearer dependence on $\mu$. By
Eq.~\eqref{eq:conditioning-bounds}, $\max_i\|P_i\|_2$ and
$\kappa(V_{\mathrm{os}})$ provide, respectively, an intrinsic lower bound and a
computable, convention-defined upper bound on the optimal condition number
$\kappa^*$, which generally requires a separate optimization to evaluate. The
two plotted quantities therefore bracket $\kappa^*$. The ground-state projector
norm instead characterizes the target state. Because the He ground-state
eigenvalue is simple in the finite-basis Hamiltonian, $\|P_{\mathrm{GS}}\|_2$
coincides with the simple-eigenvalue condition number
$\kappa(\lambda_0;\hat H_{\mathrm{TC}})$ of Sec.~\ref{sec:cond}. Together, these
quantities distinguish solver-basis sensitivity, intrinsic whole-spectrum
conditioning, and ground-state conditioning as the correlator is varied. For
TC-erfc, decreasing $\mu$ produces a monotonic increase in both
$\kappa(V_{\mathrm{os}})$ and $\max_i\|P_i\|_2$, whereas TC-exp shows no such
monotonic trend. To assess whether this eigenbasis sensitivity is specific to
the He calculations, we also re-analyzed the full-TC and xTC data of
Ref.~\cite{uvarov_accuracy_2025} using the same eigenspace-orthonormalized
convention; details are reported in Appendix~\ref{app:uvarov-reanalysis}.

These observations indicate that raw eigenvector condition numbers could be
inflated by the basis chosen within degenerate eigenspaces. This point is
directly relevant to resource estimates because the
query-complexity analyses of
Refs.~\cite{shao_computing_2022,Low_2026,zhang2025heisenberglimit} use an
\textit{optimal} condition number, or a known upper bound on it, rather than the
condition number of an arbitrary eigenvector basis returned by a numerical
solver. Concretely, Ref.~\cite{Low_2026} defines the Jordan condition number of
the input matrix as $\inf_V\|V\|\,\|V^{-1}\|$ over all similarity
transformations that bring it to Jordan form. The query complexity of its
Chebyshev-based eigenvalue algorithms is expressed through a known upper bound
$\kappa_V$ on this infimum. A more careful, basis-aware analysis is therefore
required before translating the conditioning of the TC Hamiltonian into
resource estimates.

\begin{figure}[!htbp]
    \centering
    \includegraphics[width=\linewidth]{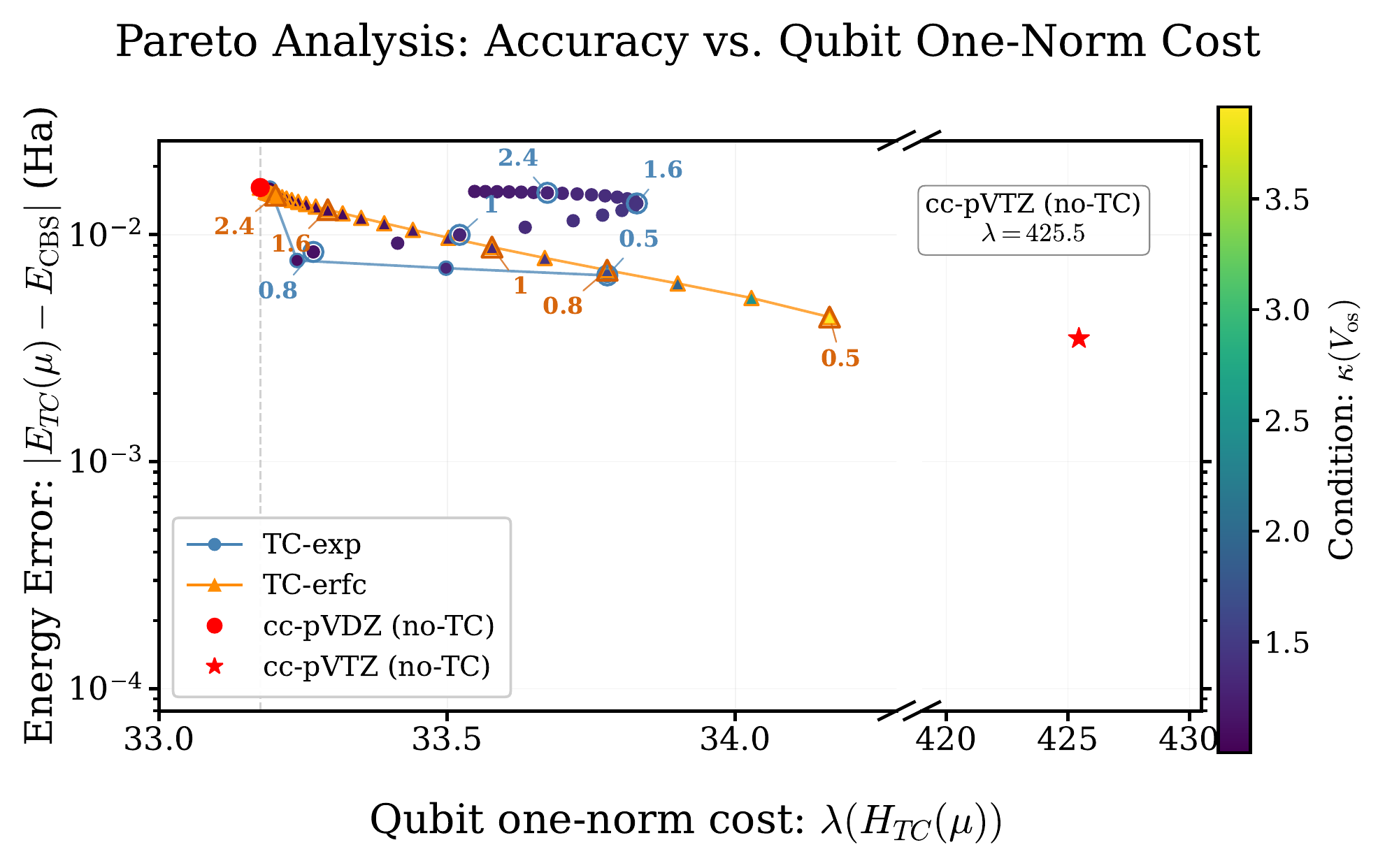}
    \caption{Pareto-front comparison of accuracy and cost-relevant diagnostics
    for He/cc-pVDZ. The vertical axis shows the absolute energy error
    $|E_{\mathrm{TC}}(\mu)-E_{\mathrm{CBS}}|$, and the horizontal axis shows the
    Hamiltonian Pauli $1$-norm $\lambda(\hat H_{\mathrm{TC}}(\mu))$. The color
    scale encodes $\kappa(V_{\mathrm{os}})$, and selected $\mu$ values are
    annotated. Bare cc-pVDZ (red circle) and cc-pVTZ (red star) values are
    included for comparison. The broken horizontal axis accommodates the bare
    cc-pVTZ value, $\lambda\approx425.5$.}
    \label{fig:new-tc_summary-3}
\end{figure}

Figure~\ref{fig:new-tc_summary-3} uses the basis-fixed conditioning measure
to combine the cc-pVDZ energy error with the cost-relevant quantities entering
the complexity bounds in a single accuracy--cost--conditioning comparison.
Within the parameter range included in this cost scan, the TC points improve
upon the bare cc-pVDZ energy error but do not reach the bare cc-pVTZ error.
Across the scan, applying the TC transformation within the same cc-pVDZ
space (5 spatial orbitals) increases the Pauli one-norm by at most $3.0\%$
relative to the bare cc-pVDZ Hamiltonian. By contrast, enlarging the bare
orbital space to cc-pVTZ (14 spatial orbitals) gives a one-norm that is $12.8$
times the bare cc-pVDZ value. The
fixed-basis accuracy gain is therefore accompanied by only a small change in
Hamiltonian one-norm over this parameter range.

\FloatBarrier

\subsubsection{Conditioning across orbital-basis sizes and symmetry sectors}
Having established the fixed-basis accuracy--cost comparison, we next ask
how the conditioning behaves as the one-particle basis is enlarged and what
role the physical particle-number sector plays in the condition number of the
full qubit matrix. After projection onto a finite one-particle space, the
TC-modified one- and two-electron integrals define a
particle-number-conserving qubit Hamiltonian on the full Fock space. Although
the physical He problem is confined to the $N=2$ sector, the full qubit
matrix therefore also contains other particle-number sectors that can affect
its condition number. Recalling Eq.~\eqref{eq:direct-sum-cond}, the condition
number of the assembled eigenvector matrix is bounded below by the largest
sectorwise condition number. An analysis at fixed particle number therefore
provides a proxy for assessing how conditioning changes as the
one-particle basis is enlarged. 

\begin{figure}[!htbp]
    \centering
    \includegraphics[width=\linewidth]{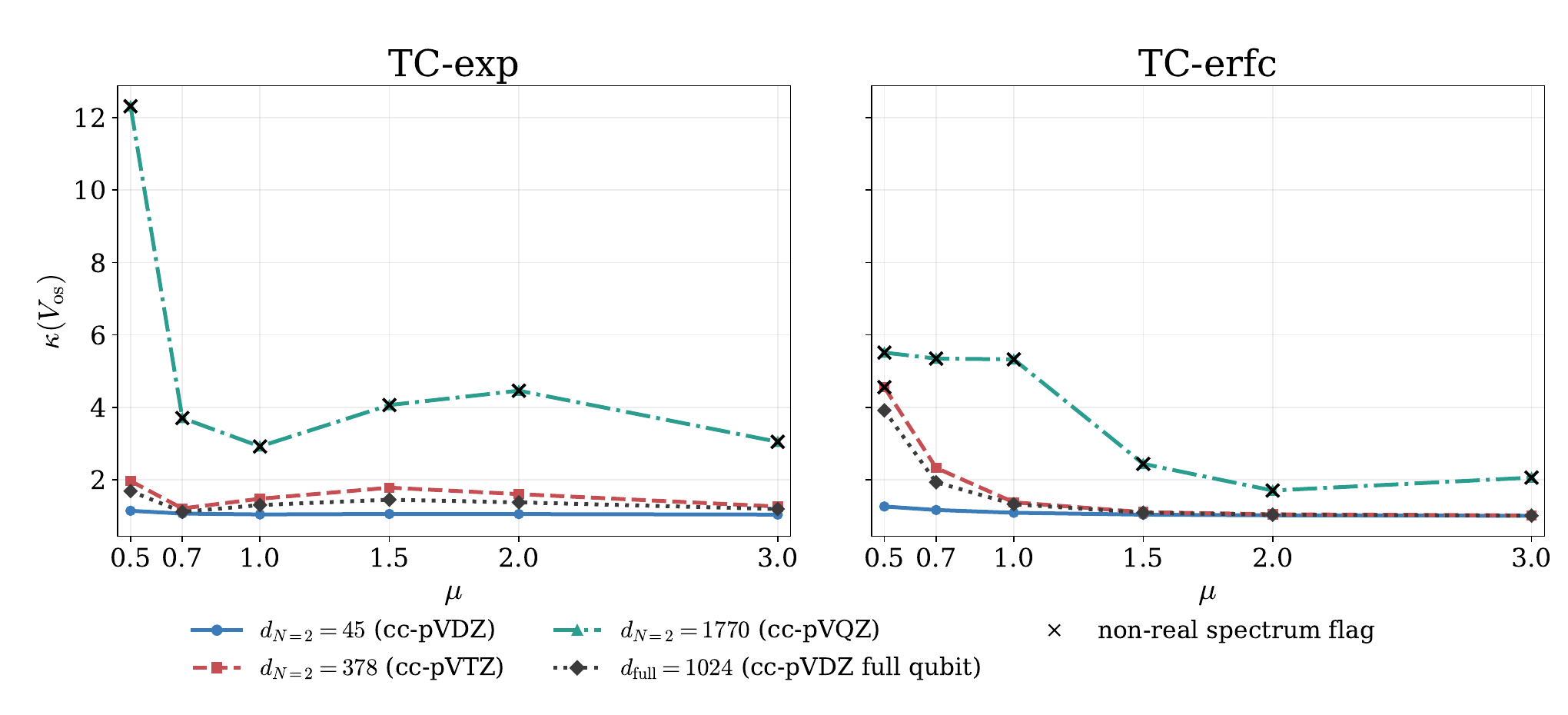}
    \caption{Eigenvector condition number
    $\kappa_{\mathrm{os}}^{(N=2)}$ for the He transcorrelated Hamiltonian as a
    function of the correlator parameter $\mu$, restricted to the physical $N=2$
    particle-number sector at three basis-set sizes: cc-pVDZ (blue circles,
    $d_{N=2}=45$), cc-pVTZ (red squares, $d_{N=2}=378$), and cc-pVQZ (green
    triangles, $d_{N=2}=1770$). The cc-pVDZ
    full-qubit-space value $\kappa_{\mathrm{os}}^{\mathrm{full}}$
    ($d_{\mathrm{full}}=1024$, black diamonds) from
    Fig.~\ref{fig:new-tc_summary-2} is included for comparison. Left: TC-exp;
    right: TC-erfc. The $\times$ marker denotes a non-real spectrum flag,
    discussed in the main text.
}
    \label{fig:cond_sector}
\end{figure}

Figure~\ref{fig:cond_sector} compares the basis-fixed eigenvector
condition number in the physical $N=2$ sector across cc-pVDZ, cc-pVTZ, and
cc-pVQZ, together with the corresponding full-qubit-space value at cc-pVDZ.
The full-qubit-space comparison is limited to cc-pVDZ because exact
diagonalization is not feasible for the larger basis sets. At
cc-pVDZ, $\kappa_{\mathrm{os}}^{\mathrm{full}}$ is consistently larger than
$\kappa_{\mathrm{os}}^{(N=2)}$, showing that the assembled full qubit space is
more poorly conditioned than the physical sector alone. For $\mu>1$, the increase in
$\kappa_{\mathrm{os}}^{(N=2)}$ from cc-pVDZ to cc-pVTZ remains modest: across
the two correlator forms, the cc-pVTZ/cc-pVDZ ratio ranges from $1.01$ to
$1.69$. At $\mu\leq1$, the basis dependence is less uniform.  The increase in
$\kappa_{\mathrm{os}}^{(N=2)}$ becomes more
noticeable from cc-pVTZ to cc-pVQZ, where every plotted $\mu$ value carries the
non-real-spectrum flag. The ground-state eigenvalue remains real in all of
these cases. Because these correlation-consistent orbital spaces are not
nested, unlike the systematically enlarged SHO basis, the three basis sets do
not form a controlled scaling sequence. We therefore regard the observed
increase as a diagnostic trend rather than a scaling law; a systematic study
using nested orbital spaces is left for future work.

To further examine the spectral structure associated with the
non-real-spectrum flags observed above, we focus on the He/TC cc-pVQZ
$N=2$-sector Hamiltonian at $\mu=3.0$.
Panels~(a) and (b) of Fig.~\ref{fig:cond_sector-2} plot its full
eigenvalue spectrum in the complex plane. Within this finite basis, the
bound-state window contains only the real ground state (red star)\footnote{The cc-pVQZ basis used here does
not include diffuse augmentation and cannot represent the higher (Rydberg)
bound states of the He atom; the bound-state window $\mathrm{Re}\,E < -2$~Ha of
the projected $N=2$-sector Hamiltonian therefore contains only the ground
state.}.
The value $-2$~Ha is the first ionization threshold of the underlying
Hermitian He Hamiltonian, corresponding to $\mathrm{He}^{+}+e^{-}$; by the
Hunziker--van Winter--Zhislin (HVZ) theorem, it is the lower edge of the
essential spectrum~\cite{cycon_schrodinger_1987}. Some finite-basis
eigenstates above this threshold may be pseudocontinuum states associated
with the ionization continuum. In the finite-basis TC spectra shown here, all
non-real eigenpairs ($|\mathrm{Im}\,E|>10^{-6}$~Ha, marked by $\times$) lie
above the first ionization threshold, whereas the ground-state eigenvalue
remains real. 
Panels~(c) and (d) track how the conditioning of the retained right-eigenvector matrix and the number 
of non-real eigenpairs change as the energy cutoff is increased.

\begin{figure}[!htbp]
    \centering
    \includegraphics[width=\linewidth]{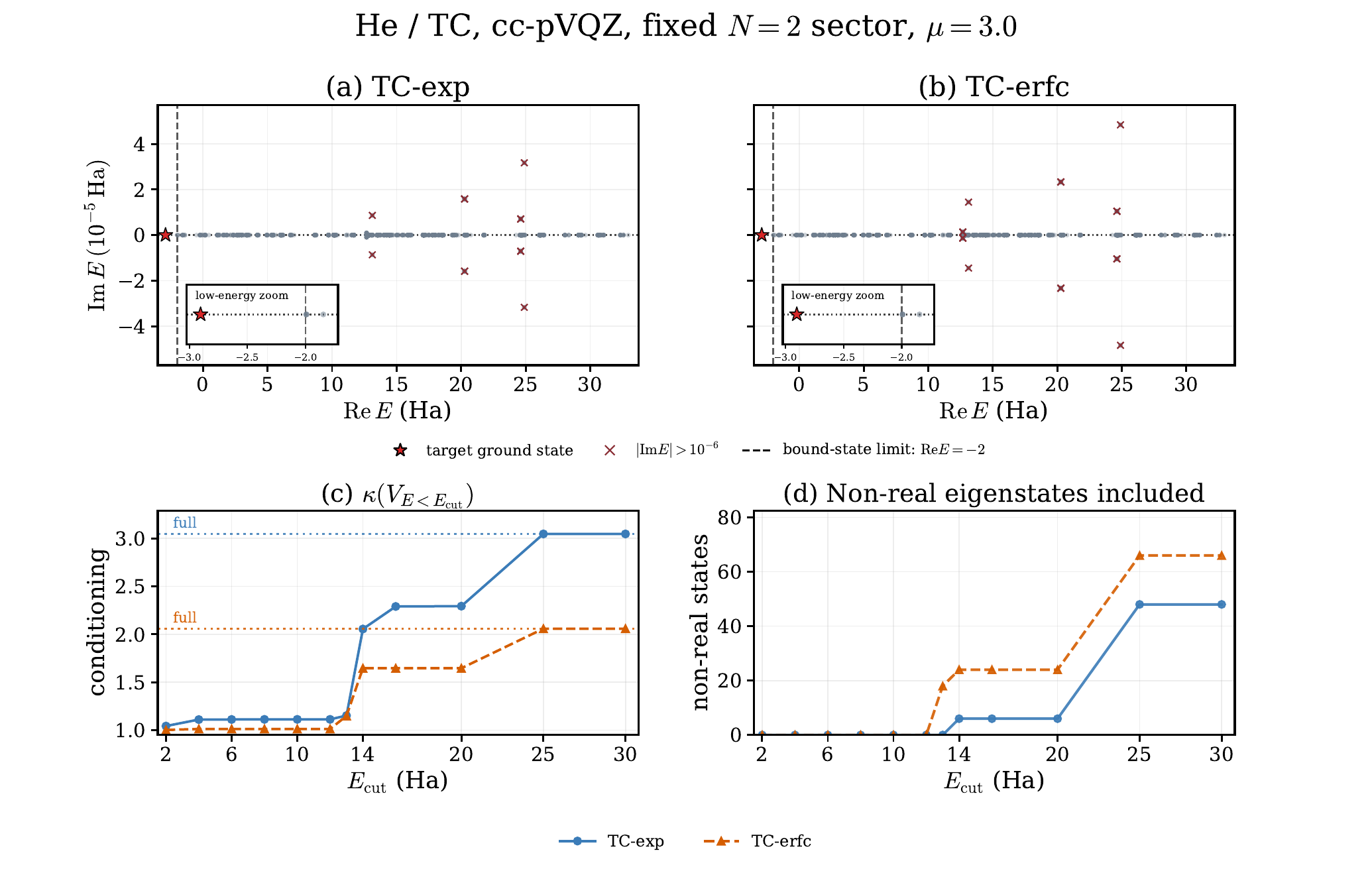}
    \caption{Spectrum and energy-resolved conditioning of the He TC Hamiltonian
in the cc-pVQZ basis set, fixed $N=2$ sector and $\mu = 3.0$. (a)--(b)
Complex-plane spectrum for TC-exp and TC-erfc; the ground state is marked by a
red star, the vertical dashed line indicates the bound-state threshold
$\mathrm{Re}\,E = -2$~Ha, and $\times$ denotes eigenpairs with $|\mathrm{Im}\,E|
> 10^{-6}$~Ha. The insets show a low-energy zoom around the ground state. (c)
Conditioning number $\kappa(V_{E<E_{\mathrm{cut}}})$ of the cumulatively
retained right eigenvector matrix, with $V_{E<E_{\mathrm{cut}}}$ defined via
Eq.~\eqref{eq:V-ecut}, as a function of the energy cutoff $E_{\mathrm{cut}}$;
horizontal dotted lines mark the full-spectrum value. (d) Number of included
eigenpairs with non-real eigenvalues up to $E_{\mathrm{cut}}$. }
    \label{fig:cond_sector-2}
\end{figure}

Panel~(c) reports the condition number
$\kappa(V_{E<E_{\mathrm{cut}}})$ as right invariant subspaces are cumulatively
included in order of increasing energy\footnote{In the numerical
implementation, the cutoff is applied to eigenvalue clusters rather than to
individual eigenvectors. A cluster $C$ is retained when
$\min_{\lambda\in C}\mathrm{Re}\,\lambda \le E_{\mathrm{cut}}$, and the entire
right invariant subspace associated with that cluster is included. Within each
retained cluster, we first construct an SVD-orthonormal basis $Q_C$ for the
right eigenspace and then form $V_{E<E_{\mathrm{cut}}}$ by concatenating these
$Q_C$ blocks. Thus the notation $E<E_{\mathrm{cut}}$ is a shorthand for this
cluster-preserving selection rule.} 
Specifically, we write
\begin{equation}
    V_{E<E_{\mathrm{cut}}}
    =
    \bigl[\, Q_C : \min_{\lambda\in C}\mathrm{Re}\,\lambda \le E_{\mathrm{cut}} \,\bigr]
    \in \mathbb{C}^{D\times m(E_{\mathrm{cut}})} ,
    \label{eq:V-ecut}
\end{equation}
where $D$ is the dimension of the fixed-$N=2$ sector and $m(E_{\mathrm{cut}})\le
D$ is the total dimension of the retained invariant subspaces. The matrix
$V_{E<E_{\mathrm{cut}}}$ is therefore generally rectangular. We define its
conditioning by the singular-value ratio $\kappa(V_{E<E_{\mathrm{cut}}}) =
\sigma_{\max}(V_{E<E_{\mathrm{cut}}})/\sigma_{\min}(V_{E<E_{\mathrm{cut}}})$,
which reduces to the standard $\|V\|_2\,\|V^{-1}\|_2$ when the full square
eigenvector matrix is retained. This ratio measures the linear independence of
the retained right invariant subspaces. For a window containing only the
normalized ground-state vector, it is identical one.

Panel~(d) counts the non-real eigenpairs retained at each $E_{\mathrm{cut}}$ and
therefore identifies when the branches seen in panels~(a) and (b) enter
$V_{E<E_{\mathrm{cut}}}$. Panels~(c) and (d) together show that the main
increases in the cumulative condition number occur over the same cutoff ranges
in which additional non-real eigenpairs are included. For $E_{\mathrm{cut}}\leq
12$~Ha, $\kappa(V_{E<E_{\mathrm{cut}}})$ remains close to unity and no retained
eigenpair exceeds the threshold $|\mathrm{Im}\,E|>10^{-6}$~Ha. The first
substantial increase in conditioning occurs over $E_{\mathrm{cut}}=13$--$14$~Ha,
where the first eigenpairs flagged as non-real also enter the retained space.
Both quantities exhibit a second increase near $E_{\mathrm{cut}}=25$~Ha. Because
varying $E_{\mathrm{cut}}$ changes only the retained set of eigenpairs, rather
than the Hamiltonian itself, this correspondence does not establish a causal
relation between the two phenomena.

\subsection{Beyond two electrons and the xTC approximation}
We now extend the helium analysis in two directions. First, we consider
systems with more than two electrons, for which the full TC Hamiltonian
contains the three-body operator $\hat L_{\boldsymbol\theta}$ of
Sec.~\ref{sec:finite-dimension-tc}. To test whether the same
accuracy--cost--conditioning trade-off persists in this setting, we consider
the beryllium atom, the simplest closed-shell atom with more than two
electrons. We use the minimal STO-6G basis and compare the results with the
bare STO-6G and cc-pVDZ references. Second, we replace the full TC
Hamiltonian with the xTC approximation and test it for both Be/STO-6G and
LiH/STO-6G.

\begin{figure}[!htbp]
    \centering
    \includegraphics[width=\linewidth]{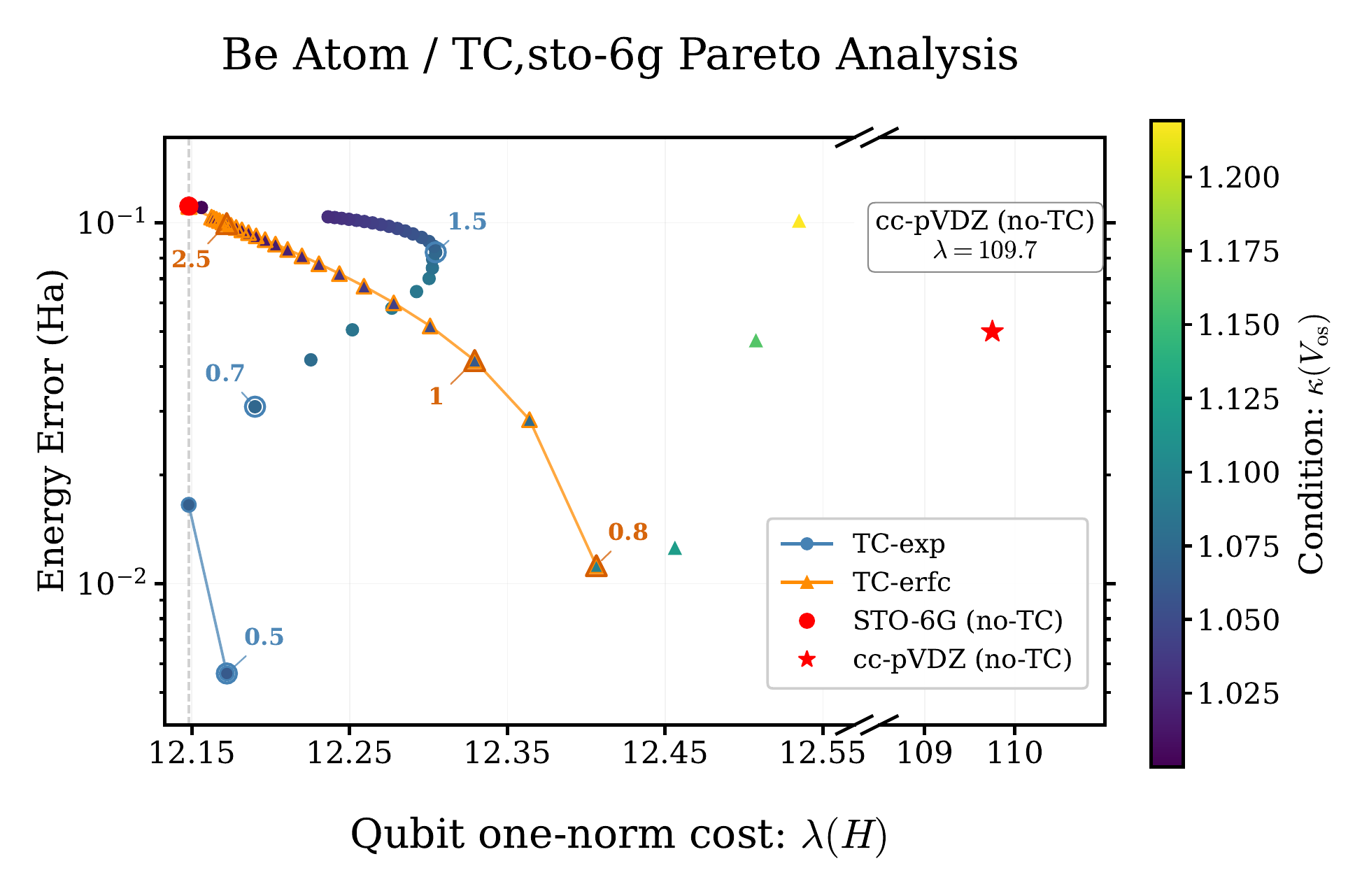}
    \caption{Pareto analysis for Be/STO-6G. Red markers denote
    the bare STO-6G and cc-pVDZ references, while the colored TC points
    represent the $\mu$ scan; selected
    values of $\mu$ are annotated. The color scale
    indicates $\kappa(V_{\mathrm{os}})$.
}
    \label{fig:Be-atom-tc}
\end{figure}

Figure~\ref{fig:Be-atom-tc} shows the analogous accuracy--cost Pareto analysis for
Be, with the bare STO-6G and cc-pVDZ references serving as quantitative
benchmarks. The bare STO-6G FCI point has an error of
$1.11\times 10^{-1}$~Ha relative to the non-relativistic CBS reference \cite{PhysRevA.44.7071}, while the
bare cc-pVDZ point has an error of $5.00\times10^{-2}$~Ha at
$\lambda(\hat H)\simeq 109.7$. Several TC/STO-6G parameter choices already
attain errors no larger than the cc-pVDZ value at a one-norm close to the STO-6G scale: for
TC-exp, $\mu=0.5$--$0.8$ gives errors from $5.6\times10^{-3}$ to
$4.2\times10^{-2}$~Ha with $\lambda(\hat H)\simeq 12.15$--$12.23$, and for
TC-erfc, $\mu=0.6$--$1.0$ gives errors from $1.1\times10^{-2}$ to
$4.7\times10^{-2}$~Ha with $\lambda(\hat H)\simeq 12.33$--$12.51$. Thus, the
TC/STO-6G Hamiltonians can match the cc-pVDZ energy accuracy with a Pauli one-norm that is approximately one order of magnitude smaller.
Across the parameter range
studied, $\kappa(V_{\mathrm{os}})$ remains in the range $1.0$--$1.2$, and the
points that attain cc-pVDZ accuracy lie in the narrower range
$\kappa(V_{\mathrm{os}})\simeq 1.06$--$1.16$. Together with the He results, this supports the favorable fixed-basis accuracy--one-norm--conditioning trade-off can also occur beyond two-electron systems. 
The correlator form continues to influence the location of each $\mu$ scan on the Pareto front.

We next compare the full TC Hamiltonian with the xTC approximation
\cite{christlmaier_xtc_2023}. We perform this comparison for Be/STO-6G
and its molecular counterpart, LiH/STO-6G, at the near-equilibrium bond length
$R=3.0$ bohr. Together, these cases extend the comparison from a closed-shell
atom to a minimal two-center molecular system.
In the full TC Hamiltonian, the three-body operator $\hat L_{\boldsymbol\theta}$
is retained explicitly, whereas xTC replaces it with its normal-ordered
effective one- and two-body contributions. For each system, we compare four quantities:
the ground-state energy difference
$|E_{\mathrm{xTC}}-E_{\mathrm{TC}}|$, the number of non-identity Pauli terms,
the Pauli one-norm, and $\kappa(V_{\mathrm{os}})$.

\begin{figure}[!htbp]
    \centering
    \includegraphics[width=\linewidth]{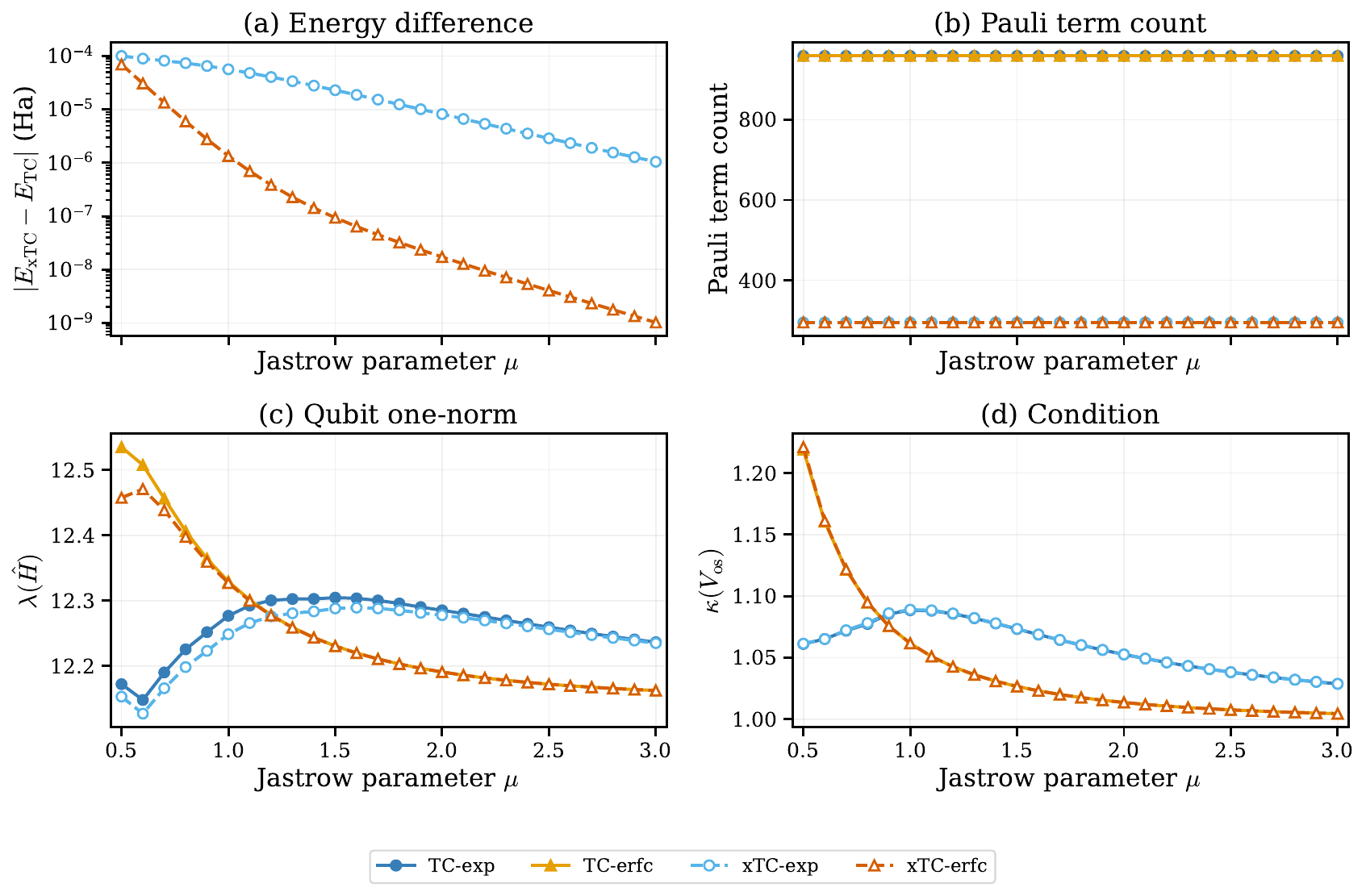}
    \caption{Comparison of full TC and xTC for the Be atom in
    the STO-6G basis. (a) Ground-state energy difference
    $|E_{\mathrm{xTC}}-E_{\mathrm{TC}}|$. (b) Number of non-identity Pauli terms
    in the qubit Hamiltonian. (c) Qubit Hamiltonian Pauli-LCU one-norm
    $\lambda(\hat H)$. (d) Condition number $\kappa(V_{\mathrm{os}})$.
    Filled markers with solid lines denote full TC,
    whereas open markers with dashed lines
    denote xTC; circles correspond to TC-exp and triangles to TC-erfc.
}
    \label{fig:Be-xtc}
\end{figure}

\begin{figure}[!htbp]
    \centering
    \includegraphics[width=\linewidth]{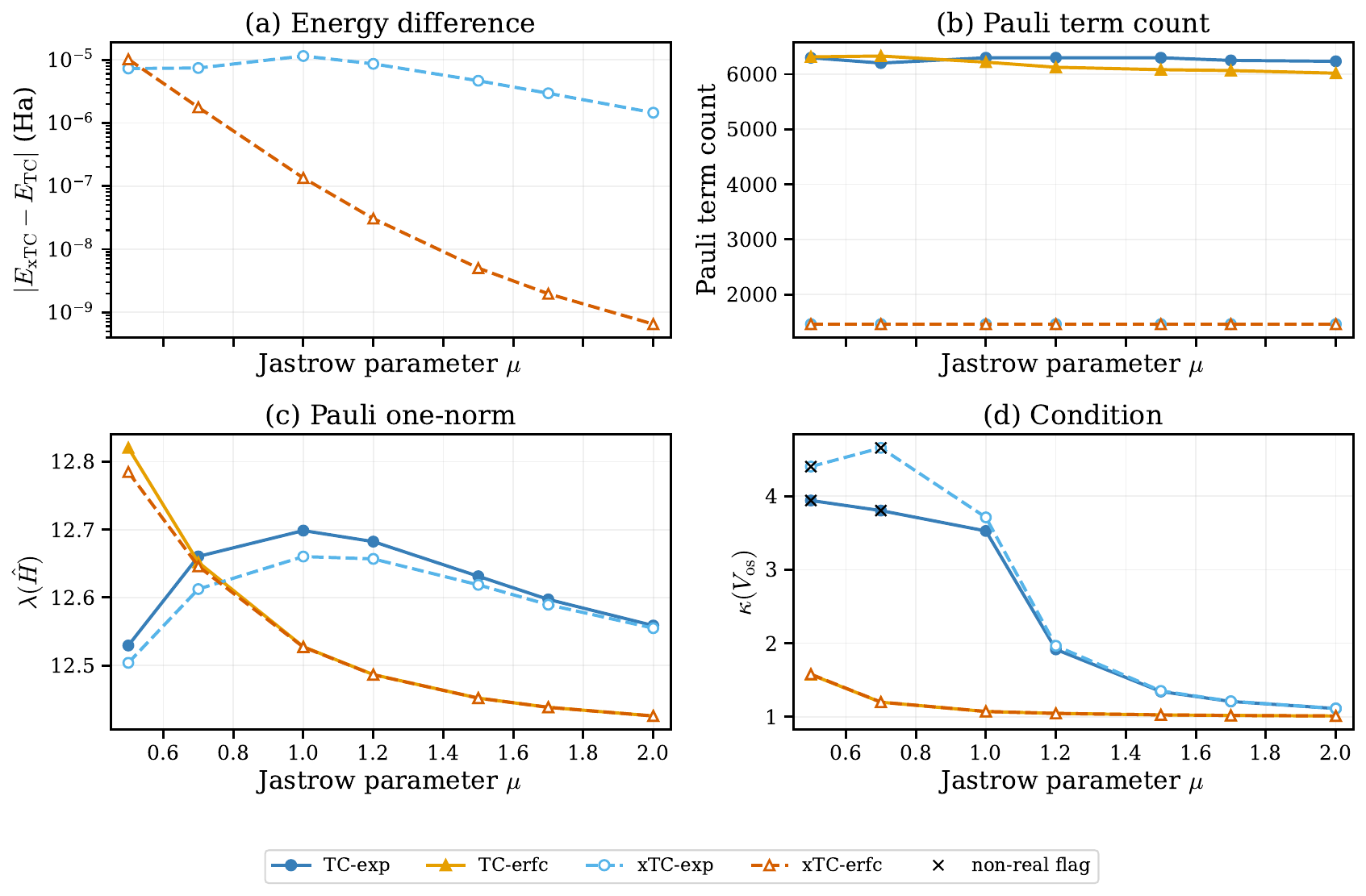}
    \caption{Comparison of full TC and xTC for the LiH molecule
    in the STO-6G basis at $R=3.0$ bohr. (a) Ground-state energy difference
    $|E_{\mathrm{xTC}}-E_{\mathrm{TC}}|$. (b) Number of non-identity Pauli terms
    in the qubit Hamiltonian. (c) Pauli one-norm of the qubit Hamiltonian
    $\lambda(\hat H)$. (d) Condition number $\kappa(V_{\mathrm{os}})$.
    Filled markers with solid lines denote full TC,
    whereas open markers with dashed lines
    denote xTC; circles correspond to TC-exp and triangles to TC-erfc.
}
    \label{fig:lih-xtc}
\end{figure}

The Be and LiH results in Figs.~\ref{fig:Be-xtc}
and~\ref{fig:lih-xtc} show close agreement between the xTC and full-TC
ground-state energies,
while the changes in conditioning are system-dependent.
Quantitatively, over the plotted parameter ranges, the maximum ground-state
energy deviation, $\max_\mu |E_{\mathrm{xTC}}-E_{\mathrm{TC}}|$, is
$1.01\times10^{-4}$~Ha for Be and $1.2\times10^{-5}$~Ha for LiH. These maxima
occur for TC-exp at $\mu=0.5$ for Be and $\mu=1.0$ for LiH.
The Pauli one-norm is also slightly lower for xTC throughout the
plotted $\mu$ ranges. For Be, the one-norm spans
$\lambda(\hat H)=12.15$--$12.53$ for full TC and $12.13$--$12.47$ for xTC.
Eliminating the residual explicit three-body operator in xTC reduces the
number of non-identity Pauli terms from
$959$ to $295$. For LiH, the corresponding one-norm ranges are
$12.43$--$12.82$ for full TC and $12.43$--$12.78$ for xTC. The Pauli term count
drops from $6018$--$6330$ to $1462$. For Be, the condition-number range is
$\kappa(V_{\mathrm{os}})=1.00$--$1.22$ for both full TC and xTC. For LiH, the
corresponding range changes from $1.01$--$3.94$ for full TC to $1.01$--$4.66$
for xTC, with an increase in its upper endpoint.

However, given the He analysis in
Fig.~\ref{fig:cond_sector}, this conditioning comparison should
be interpreted with care. The xTC Hamiltonian is obtained by normal ordering with respect to a
fixed-particle-number Hartree--Fock determinant and retaining the resulting
effective one- and two-body terms.
Because the condition numbers reported here are computed over the full
block structure, they can receive contributions from sectors outside the target
physical particle-number sector. The differences induced by xTC could therefore
arise partly from these auxiliary sectors. Without a target-sector calculation,
they cannot be attributed specifically to changes in the conditioning of the
target low-energy problem. Future work should therefore assess target-sector
conditioning for xTC and for other reduced Hamiltonians, such as those obtained
by active-space selection.

\section{Conclusion and Outlook}
\label{sec:conclusion}
In this paper, we examined the trade-off between accuracy and cost in transcorrelated Hamiltonians and its implications for quantum algorithms. We first used a solvable model to isolate the mechanisms behind this trade-off in a controlled setting. We then studied second-quantized electronic Hamiltonians to analyze how accuracy, the Pauli one-norm, and eigenvector conditioning behave in finite orbital spaces. The two settings exhibit a partial correspondence. Comparing them therefore helps distinguish intrinsic features from setting-specific effects and addresses the challenge emphasized in Ref.~\cite{uvarov_accuracy_2025}: understanding the behavior of the condition number relevant to non-Hermitian quantum algorithms. Furthermore, in the solvable model, we show that TC improves wavefunction regularity, which can yield a polynomial improvement in the QEVE query bound. For electronic calculations, we identify three setting-specific phenomena: (i) degeneracies require a consistent convention for the conditioning measure, such as \(\kappa(V_{\mathrm{os}})\); (ii) particle-number sectors can differ remarkably in conditioning; and (iii) the non-nested orbital spaces limit the direct asymptotic scaling interpretation. Despite these phenomena, numerical experiments on small systems (He, Be, and LiH) show that TC and xTC can improve finite-basis accuracy with a mild increase in the Pauli one-norm and \(\kappa(V_{\mathrm{os}})\), illustrating a trade-off similar to that of the solvable model.

These findings motivate several complementary directions for improving the TC representation. 
Extending to more flexible forms~\cite{Haupt2023-uo,Haupt2025,Simula2025,Haupt2026}, the systematically refinable
representations ~\cite{feniou_real-space_2025,bachmayr_hyperbolic_2012}. Restricting the encoded problem to physically relevant subspaces,
through spin-sector restriction~\cite{Dobrautz2019b,Dobrautz2021,Dobrautz2022b,Gandon2025}
or active-space selection~\cite{Roos1980,LiManni2020,LiManni2023},
may further reduce representation costs and clarify the conditioning relevant to the target problem.  

Finally, understanding how the relevant conditioning scales with basis size and physical system size is 
only an intermediate step toward assessing its algorithmic impact. 
Quantifying the impact of these findings on end-to-end quantum resource estimates remains an open problem. In particular, the distinction observed here between full-space and target-sector conditioning suggests that symmetry-aware analyses may yield tighter QEVE resource bounds than estimates based solely on global conditioning, but the quantitative impact remains to be established.

\section*{acknowledgments}

We acknowledge funding from the German Federal Ministry of Research, Technology and Space (BMFTR) under the research program Quantensysteme and funding measure Quantum Futur 3 for project No. 13N17229 as well as the Helmholtz Association via the Initiative and Networking Fund for projects No. VH-NG-21-08 under the Helmholtz Investigator program as well Projects KA-QUS-02 (qFLOW) and KA-QUS-03 (QT-Batt) under the Helmholtz Quantum Use case call. PKF acknowledges funding from the BMFTR (HYBRID and QPIC-1).

\section*{AI Disclosure}
Claude Opus 4.6 to 5.5, GPT-5.5, GPT-5.6 Sol, and GPT-6 Astra were used for technical discussion, literature searches, proof checking, writing support, and developing and checking analysis, verification, and plotting scripts. The authors conceived the research questions and theoretical direction, directed and verified all AI-assisted output, and take full responsibility for the content of this manuscript.

\appendix

\section{Derivation of the Baker--Campbell--Hausdorff Expansion of the TC Hamiltonian}
\label{app:BCH_derivation}

In this appendix, we provide the detailed derivation of the transcorrelated Hamiltonian using the Baker--Campbell--Hausdorff (BCH) expansion.
Throughout this appendix and Appendix~\ref{app:pairwise_derivation} the
correlator parameter $\boldsymbol\theta$ is suppressed, so that
$\hat\tau\equiv\hat\tau_{\boldsymbol\theta}$ and
$u_{ij}\equiv u_{\boldsymbol\theta}(\mathbf r_i,\mathbf r_j)$.
Starting from the definition $\hat{H}_{\mathrm{TC}} = e^{-\hat{\tau}} \hat{H} e^{\hat{\tau}}$, the expansion is given by:
\begin{equation}
    \hat{H}_{\mathrm{TC}} = \hat{H} + [\hat{H}, \hat{\tau}] + \frac{1}{2!} [[\hat{H}, \hat{\tau}], \hat{\tau}] + \dots
\end{equation}
Since the potential operator $\hat{V}$ and the correlator $\hat{\tau}$ are both multiplicative operators in the position representation, they commute ($[\hat{V}, \hat{\tau}] = 0$). 
Thus, the commutators reduce to terms involving only the kinetic energy operator $\hat{T} = -\frac{1}{2}\sum_i \nabla_i^2$:
\begin{equation}
    [\hat{H}, \hat{\tau}] = [\hat{T}, \hat{\tau}].
\end{equation}
Applying the first commutator to an arbitrary test function $\Psi$:
\begin{align}
    [\hat{T}, \hat{\tau}]\Psi &= \hat{T}(\hat{\tau}\Psi) - \hat{\tau}(\hat{T}\Psi) \nonumber \\
    &= -\frac{1}{2}\sum_i \left[ \nabla_i^2(\tau\Psi) - \tau(\nabla_i^2 \Psi) \right].
\end{align}
Using the Laplacian product rule $\nabla^2(fg) = (\nabla^2 f)g + 2(\nabla f)\cdot(\nabla g) + f(\nabla^2 g)$, we expand the first term:
\begin{equation}
    \nabla_i^2(\tau\Psi) = (\nabla_i^2\tau)\Psi + 2(\nabla_i\tau)\cdot(\nabla_i\Psi) + \tau(\nabla_i^2\Psi).
\end{equation}
Substituting this back, the $\tau(\nabla_i^2\Psi)$ terms cancel:
\begin{align}
    [\hat{T}, \hat{\tau}]\Psi &= -\frac{1}{2}\sum_i \left[ (\nabla_i^2\tau)\Psi + 2(\nabla_i\tau)\cdot(\nabla_i\Psi) \right] \nonumber \\
    &= \left( -\frac{1}{2}\sum_i (\nabla_i^2\tau) - \sum_i (\nabla_i\tau)\cdot\nabla_i \right) \Psi.
\end{align}
Since $\Psi$ is arbitrary, we obtain the operator identity:
\begin{equation}
    [\hat{T}, \hat{\tau}] = -\frac{1}{2}\sum_i (\nabla_i^2\tau) - \sum_i (\nabla_i\tau)\cdot\nabla_i.
    \label{eq:app_first_comm}
\end{equation}

The second commutator is defined as $[[\hat{T}, \hat{\tau}], \hat{\tau}]$. 
Using Eq.~\eqref{eq:app_first_comm}, and noting that the scalar term $-\frac{1}{2}\sum_i \nabla_i^2\tau$ commutes with $\hat{\tau}$, we only need to evaluate:
\begin{equation}
    [[\hat{T}, \hat{\tau}], \hat{\tau}] = \left[ -\sum_i (\nabla_i\tau)\cdot\nabla_i, \hat{\tau} \right].
\end{equation}
Applying this to a test function $\Psi$:
\begin{align}
    \left[ -\sum_i (\nabla_i\tau)\cdot\nabla_i, \hat{\tau} \right] \Psi
    &= -\sum_i (\nabla_i\tau)\cdot\nabla_i(\tau\Psi) \nonumber \\
    &\quad - \hat{\tau}\left( -\sum_i (\nabla_i\tau)\cdot\nabla_i \Psi \right).
\end{align}
Expanding the gradient $\nabla_i(\tau\Psi) = (\nabla_i\tau)\Psi + \tau(\nabla_i\Psi)$:
\begin{align}
    &= -\sum_i (\nabla_i\tau)\cdot\left[ (\nabla_i\tau)\Psi + \tau(\nabla_i\Psi) \right] + \sum_i \tau (\nabla_i\tau)\cdot(\nabla_i\Psi) \nonumber \\
    &= -\sum_i |\nabla_i\tau|^2 \Psi - \sum_i \tau(\nabla_i\tau)\cdot(\nabla_i\Psi) \nonumber \\
    &\quad + \sum_i \tau (\nabla_i\tau)\cdot(\nabla_i\Psi) \nonumber \\
    &= -\sum_i |\nabla_i\tau|^2 \Psi.
\end{align}
Thus, the second commutator is a purely multiplicative operator:
\begin{equation}
    [[\hat{T}, \hat{\tau}], \hat{\tau}] = -\sum_i |\nabla_i\tau|^2.
    \label{eq:app_second_comm}
\end{equation}
Since this result depends only on position coordinates, it commutes with $\hat{\tau}$, ensuring that all higher-order commutators vanish.

Substituting Eqs.~\eqref{eq:app_first_comm} and~\eqref{eq:app_second_comm}
into the truncated series $\hat H+[\hat T,\hat\tau]+\tfrac12[[\hat T,\hat\tau],\hat\tau]$
gives the closed form quoted as Eq.~\eqref{eq:tc-real-space} of the main text,
\begin{equation}
    \hat{H}_{\mathrm{TC}}
    = \hat{H}
    - \sum_i (\nabla_i\tau)\cdot\nabla_i
    - \frac{1}{2}\sum_i \left[ \nabla_i^2\tau + |\nabla_i\tau|^2 \right] ,
    \label{eq:app_htc_closed}
\end{equation}
where the second term is the non-Hermitian drift operator built from the
field $\mathbf F_{i,\boldsymbol\theta}$ of Eq.~\eqref{eq:first_order_field_def}
and the third is the TC-generated scalar correction
$V_{\mathrm{TC,eff},\boldsymbol\theta}$ of Eq.~\eqref{eq:Veff_def}.

\section{Derivation of Two- and Three-Body Operators}
\label{app:pairwise_derivation}
Here, we derive the explicit form of the TC Hamiltonian assuming a symmetric pairwise Jastrow correlator:
\begin{equation}
    \hat{\tau} = \sum_{i<j} u(\mathbf{r}_i, \mathbf{r}_j).
\end{equation}
The gradient of $\tau$ with respect to electron $i$ is given by:
\begin{equation}
    \nabla_i \tau = \sum_{j \neq i} \nabla_i u_{ij},
\end{equation}
where we use the shorthand $u_{ij} \equiv u(\mathbf{r}_i, \mathbf{r}_j)$. Substituting this into the general BCH expansion derived in Appendix~\ref{app:BCH_derivation}, we analyze the terms individually.

The non-self-adjoint drift term becomes:
\begin{equation}
    - \sum_i (\nabla_i\tau)\cdot\nabla_i = - \sum_i \sum_{j \neq i} (\nabla_i u_{ij}) \cdot \nabla_i.
\end{equation}
The Laplacian term acts linearly:
\begin{equation}
    -\frac{1}{2}\sum_i \nabla_i^2\tau = - \frac{1}{2} \sum_i \sum_{j \neq i} \nabla_i^2 u_{ij}.
\end{equation}
The squared gradient term requires expanding the square of the sum:
\begin{align}
    -\frac{1}{2}\sum_i |\nabla_i\tau|^2 &= - \frac{1}{2} \sum_i \left| \sum_{j \neq i} \nabla_i u_{ij} \right|^2 \nonumber \\
    &= - \frac{1}{2} \sum_i \sum_{j \neq i} |\nabla_i u_{ij}|^2 \nonumber \\
    &\quad - \frac{1}{2} \sum_i \sum_{j \neq i} \sum_{k \neq i, j}
        \nabla_i u_{ij} \cdot \nabla_i u_{ik} .
\end{align}
Here, the first term represents diagonal interactions ($j=k$, two-body), while the second term represents cross interactions ($j \neq k$, three-body).
Combining all terms and regrouping summations over unique pairs ($i<j$) and triplets ($i<j<k$) yields the operators $\hat{K}$ and $\hat{L}$ presented in the main text.
In the three-body term the ordered pairs $(j,k)$ and $(k,j)$ contribute
equally at fixed $i$, so each unordered triplet is counted twice; this cancels
the prefactor $-\tfrac12$ and $\hat L_{\boldsymbol\theta}$ enters with unit
coefficient, whereas the factors $\tfrac12$ are retained in
$\hat K_{\boldsymbol\theta}$. Explicitly,
\begin{equation}
    \hat H_{\mathrm{TC}} - \hat H
    = - \sum_{i<j} \hat K_{\boldsymbol\theta}(\mathbf r_i,\mathbf r_j)
      - \sum_{i<j<k} \hat L_{\boldsymbol\theta}(\mathbf r_i,\mathbf r_j,\mathbf r_k) .
    \label{eq:app_KL_assembled}
\end{equation}

\section{A concrete non-normal example}
\label{app:condition-number}
Consider a real, non-normal matrix 
\begin{equation}
M=\begin{pmatrix}
1 & 0 & 1\\
0 & 1 & 0\\
0 & 0 & 2
\end{pmatrix}.
\end{equation}
It is non-normal $(MM^\dagger \neq M^\dagger M )$ but diagonalizable.
Its eigenvalues are \(\{1,1,2\}\). The eigenspace associated with the
degenerate eigenvalue \(\lambda=1\) is
\begin{equation}
\mathcal E_1
=
\{(x,y,0)^T:x,y\in\mathbb R\}
=
\mathrm{span}\{e_1,e_2\}.
\end{equation}
A right eigenvector associated with \(\lambda=2\) is \((1,0,1)^T\), which
we normalize as
\begin{equation}
v_3=\frac{1}{\sqrt2}(1,0,1)^T.
\end{equation}

To illustrate the remaining basis freedom inside the degenerate
eigenspace, define
\begin{equation}
v_1=e_1,
\qquad
v_2(\theta)=\cos\theta\,e_1+\sin\theta\,e_2,
\end{equation}
with $\theta\in(0,\pi)$, so that $v_1$ and $v_2(\theta)$ are linearly
independent. Both vectors have unit Euclidean norm. The corresponding
eigenvector matrix is
\begin{equation}
V(\theta)
=
\bigl[v_1\ \ v_2(\theta)\ \ v_3\bigr]
=
\begin{pmatrix}
1 & \cos\theta & 1/\sqrt2\\
0 & \sin\theta & 0\\
0 & 0 & 1/\sqrt2
\end{pmatrix}.
\end{equation}
It satisfies
\begin{equation}
M V(\theta)
=
V(\theta)\,\mathrm{diag}(1,1,2).
\end{equation}
Thus, for every $\theta\in(0,\pi)$, $V(\theta)$ is a valid
diagonalizing eigenvector matrix with unit-norm columns.

The unit-norm convention does not, however, fix the value of the condition
number, since $\kappa_2(V(\theta))$ still varies strongly with $\theta$.

For this example, the optimal condition number $\kappa^*$
introduced in Sec.~\ref{sec:cond} can also be computed explicitly. 
Let $\widetilde V$ be any diagonalizing eigenvector matrix, with its columns
ordered so that the first two columns correspond to $\lambda=1$ and the
third column corresponds to $\lambda=2$. Then
\begin{equation}
M
=
\widetilde V\,\mathrm{diag}(1,1,2)\,\widetilde V^{-1}.
\end{equation}
We now shift and rescale $M$ so that its two eigenvalues are moved to $+1$ and
$-1$. This operation does not touch the eigenvectors, so the same $\widetilde V$
diagonalizes the result:
\begin{equation}
3I-2M
=
\widetilde V\,\mathrm{diag}(1,1,-1)\,\widetilde V^{-1}.
\end{equation}
An orthogonal eigenbasis would now give $\|3I-2M\|_2=1$; any excess over unity
therefore measures the non-orthogonality of the columns of $\widetilde V$, which is
what turns this norm into a lower bound on $\kappa_2(\widetilde V)$.
Using $\|\mathrm{diag}(1,1,-1)\|_2=1$, we obtain
\begin{equation}
\|3I-2M\|_2
\leq
\|\widetilde V\|_2\,\|\widetilde V^{-1}\|_2
=
\kappa_2(\widetilde V).
\end{equation}
This inequality implies that $\|3I-2M\|_2$ serves as a universal lower bound 
for $\kappa_2(\widetilde V)$, for every admissible eigenvector matrix $\widetilde V$.

A direct calculation gives
\begin{equation}
(3I-2M)^T(3I-2M)
=
\begin{pmatrix}
1&0&-2\\
0&1&0\\
-2&0&5
\end{pmatrix}.
\end{equation}
The eigenvalues of this matrix are
\begin{equation}
(1+\sqrt2)^2,
\qquad
1,
\qquad
(\sqrt2-1)^2.
\end{equation}
Hence
\begin{equation}
\|3I-2M\|_2=1+\sqrt2.
\end{equation}
Thus every diagonalizing eigenvector matrix satisfies
\begin{equation}
\kappa_2(\widetilde V)\geq 1+\sqrt2,
\end{equation}
and therefore
\begin{equation}
\kappa^*\geq 1+\sqrt2.
\end{equation}

In this particular example, this lower bound is attained by the gauge-fixing
prescription of Eq.~\eqref{eq:Vos}: orthonormalizing inside each eigenspace, with
$Q_1=[\,e_1\ e_2\,]$ and $Q_2=v_3$, gives $V_{\mathrm{os}}=V(\pi/2)$, that is,
\begin{equation}
V(\pi/2)
=
\begin{pmatrix}
1&0&1/\sqrt2\\
0&1&0\\
0&0&1/\sqrt2
\end{pmatrix}.
\end{equation}
The squared singular values of $V(\pi/2)$ are the eigenvalues of
\begin{equation}
V(\pi/2)^T V(\pi/2)
=
\begin{pmatrix}
1&0&1/\sqrt2\\
0&1&0\\
1/\sqrt2&0&1
\end{pmatrix}.
\end{equation}
These eigenvalues are
\begin{equation}
1+\frac{1}{\sqrt2},
\qquad
1,
\qquad
1-\frac{1}{\sqrt2}.
\end{equation}
Therefore,
\begin{equation}
\kappa_2(V(\pi/2))
=
\sqrt{
\frac{1+1/\sqrt2}{1-1/\sqrt2}
}
=
1+\sqrt2.
\end{equation}
Combining the lower bound with this attainable value yields
\begin{equation}
\kappa^*=1+\sqrt2.
\end{equation}
The orthonormalized construction of Eq.~\eqref{eq:Vos}, which is only
suboptimal in general, is thus optimal for this matrix:
$\kappa_2(V_{\mathrm{os}})=\kappa^*$.

Finally, this value is consistent with the spectral-projector bounds
discussed in Sec.~\ref{sec:cond}. The matrix $M$ has $r=2$ distinct eigenvalues,
$\lambda_1=1$ and $\lambda_2=2$; the corresponding spectral projectors are
\begin{equation}
P_2
=
\frac{M-I}{2-1}
=
M-I
=
\begin{pmatrix}
0&0&1\\
0&0&0\\
0&0&1
\end{pmatrix},
\qquad
P_1
=
I-P_2
=
\begin{pmatrix}
1&0&-1\\
0&1&0\\
0&0&0
\end{pmatrix}.
\end{equation}
These are spectral projectors, not orthogonal projectors. In particular,
$P_1$ projects onto
$\mathcal E_1=\mathrm{span}\{e_1,e_2\}$ along the eigenspace
$\mathrm{span}\{(1,0,1)^T\}$ associated with $\lambda=2$, which is not
orthogonal to $\mathcal E_1$. Directly,
\begin{equation}
P_1^T P_1
=
\begin{pmatrix}
1&0&-1\\
0&1&0\\
-1&0&1
\end{pmatrix},
\qquad
P_2^T P_2
=
\begin{pmatrix}
0&0&0\\
0&0&0\\
0&0&2
\end{pmatrix}.
\end{equation}
The largest eigenvalue of both $P_1^T P_1$ and $P_2^T P_2$ is $2$.
Hence
\begin{equation}
\|P_1\|_2=\|P_2\|_2=\sqrt2.
\end{equation}
Therefore,
\begin{equation}
\max_i\|P_i\|_2=\sqrt2.
\end{equation}
The condition numbers found above therefore satisfy
\begin{equation}
\sqrt2
=
\max_i\|P_i\|_2
<
\kappa_2(V_{\mathrm{os}})
=
\kappa^*
=
1+\sqrt2
<
2\sqrt2
=
r\,\max_i\|P_i\|_2,
\end{equation}
in agreement with the bounds in Sec.~\ref{sec:cond}.

\section{Analytical SHO + delta potential matrix elements}
\label{app:sho-integrals}

This appendix derives the matrix-element formulas used in Sec.~\ref{sec:projected-sho-tc}.  The derivation starts from full-line matrix elements.  The contact distribution is canceled before any parity argument is used.  Only after this cancellation do the remaining regular full-line integrals reduce to positive-half-axis scalar moments.  Throughout this appendix the correlator range parameter satisfies $\mu>0$.

\subsection{Full-line identities and contact cancellation}
\label{appsec:sho-full-line}

The normalized SHO basis functions are
\begin{equation}
    \phi_n(x)=\frac{1}{\sqrt{2^n n!\sqrt\pi}}H_n(x)e^{-x^2/2},
    \qquad n=0,1,2,\ldots,
    \label{eq:app-sho-basis}
\end{equation}
with parity
\begin{equation}
    \phi_n(-x)=(-1)^n\phi_n(x)
    \label{eq:app-sho-parity}
\end{equation}
and derivative identity
\begin{equation}
    \phi'_n(x)=\sqrt{\frac n2}\,\phi_{n-1}(x)
    -\sqrt{\frac{n+1}{2}}\,\phi_{n+1}(x).
    \label{eq:app-sho-derivative}
\end{equation}
Their values at the origin are
\begin{equation}
    \phi_{2m}(0)=(-1)^m\frac{\sqrt{(2m)!}}{2^m m!\pi^{1/4}},
    \qquad \phi_{2m+1}(0)=0 .
    \label{eq:app-sho-origin}
\end{equation}

Let $u(x)=v(|x|)$ be an even Jastrow exponent satisfying $v'(0^+)=g$.  Then
\begin{equation}
    u''(x)=u''_{\mathrm{reg}}(x)+2g\delta(x)
    \qquad \text{in }\mathcal D'(\mathbb R).
    \label{eq:app-u-second-dist}
\end{equation}
Consequently, the physical contact term $g\delta(x)$ of the Hamiltonian
and the distributional part of $-\tfrac12u''$ cancel identically: for arbitrary SHO functions,
\begin{equation}
\begin{aligned}
    &\int_{-\infty}^{\infty}\phi_m(x)
      \left[g\delta(x)-\frac12\,2g\delta(x)\right]
      \phi_n(x)\,dx \\
    &\qquad =g\phi_m(0)\phi_n(0)-g\phi_m(0)\phi_n(0)=0.
\end{aligned}
\label{eq:app-delta-cancel-matrix}
\end{equation}
Thus the even-sector matrix contains no separate contact term once the regular operator has been formed.  With $F(x)=u'(x)$ and
\begin{equation}
    V_{\mathrm{TC,eff}}(x)
    =-\frac12u''_{\mathrm{reg}}(x)-\frac12F(x)^2,
    \label{eq:app-Veff-def}
\end{equation}
the even-sector TC matrix is
\begin{equation}
    (H_{\mathrm{TC}})_{mn}
    =\left(2n+\frac12\right)\delta_{mn}+D_{mn}
    +V_{\mathrm{TC,eff},mn},
    \qquad m,n=0,1,\ldots,N_{\mathrm{even}}-1,
    \label{eq:app-even-matrix}
\end{equation}
where
\begin{align}
    D_{mn}
    &=-\int_{-\infty}^{\infty}\phi_{2m}(x)F(x)\phi'_{2n}(x)\,dx,
    \label{eq:app-D-full-line}\\
    V_{\mathrm{TC,eff},mn}
    &=\int_{-\infty}^{\infty}\phi_{2m}(x)
    V_{\mathrm{TC,eff}}(x)\phi_{2n}(x)\,dx.
    \label{eq:app-V-full-line}
\end{align}
Define the full-line field integral
\begin{equation}
    \mathcal F_{pq}=\int_{-\infty}^{\infty}\phi_p(x)F(x)\phi_q(x)\,dx .
    \label{eq:app-field-integral}
\end{equation}
Using Eq.~\eqref{eq:app-sho-derivative}, Eq.~\eqref{eq:app-D-full-line} becomes
\begin{equation}
    D_{mn}=-\sqrt n\,\mathcal F_{2m,2n-1}
    +\sqrt{\frac{2n+1}{2}}\,\mathcal F_{2m,2n+1},
    \label{eq:app-D-reduction}
\end{equation}
where the first term is omitted for $n=0$.

\subsection{TC-exp matrix elements}
\label{appsec:tc-exp}

For $u_1(x)=g|x|e^{-\mu|x|}$,
\begin{align}
    F_1(x)&=g\operatorname{sgn}(x)(1-\mu|x|)e^{-\mu|x|},
    \label{eq:app-F-exp}\\
    u''_{1,\mathrm{reg}}(x)&=g e^{-\mu|x|}(-2\mu+\mu^2|x|).
    \label{eq:app-u1-second-reg}
\end{align}
Hence
\begin{equation}
\begin{aligned}
    V_{\mathrm{TC,eff}}^{(1)}(x)
    &=g\mu e^{-\mu|x|}
      -\frac12g\mu^2|x|e^{-\mu|x|}
      -\frac12g^2e^{-2\mu|x|} \\
    &\quad +g^2\mu|x|e^{-2\mu|x|}
      -\frac12g^2\mu^2|x|^2e^{-2\mu|x|}.
\end{aligned}
\label{eq:app-Veff-exp}
\end{equation}
Introduce full-line exponential moments
\begin{align}
    A^{(a)}_{pq}(b)
    &=\int_{-\infty}^{\infty}\phi_p(x)\phi_q(x)|x|^a e^{-b|x|}\,dx,
    \label{eq:app-A-moment}\\
    B^{(a)}_{pq}(b)
    &=\int_{-\infty}^{\infty}\phi_p(x)\phi_q(x)\operatorname{sgn}(x)|x|^a e^{-b|x|}\,dx.
    \label{eq:app-B-moment}
\end{align}
Then the TC-exp scalar correction is
\begin{equation}
\begin{aligned}
    V_{\mathrm{TC,eff},mn}^{(1)}
    &=g\mu A^{(0)}_{2m,2n}(\mu)
      -\frac12g\mu^2 A^{(1)}_{2m,2n}(\mu)
      -\frac12g^2 A^{(0)}_{2m,2n}(2\mu) \\
    &\quad +g^2\mu A^{(1)}_{2m,2n}(2\mu)
      -\frac12g^2\mu^2 A^{(2)}_{2m,2n}(2\mu),
\end{aligned}
\label{eq:app-V-exp-matrix}
\end{equation}
and the field integral is
\begin{equation}
    \mathcal F^{(1)}_{pq}=gB^{(0)}_{pq}(\mu)-g\mu B^{(1)}_{pq}(\mu).
    \label{eq:app-F-exp-matrix}
\end{equation}
The first-order contribution follows from Eq.~\eqref{eq:app-D-reduction}, namely
\begin{equation}
    D^{(1)}_{mn}
    =-\sqrt n\,\mathcal F^{(1)}_{2m,2n-1}
    +\sqrt{\frac{2n+1}{2}}\,\mathcal F^{(1)}_{2m,2n+1}.
    \label{eq:app-D-exp-matrix}
\end{equation}

\subsection{TC-erfc matrix elements}
\label{appsec:tc-erfc}

For
\begin{equation}
    u_2(x)=g\left[|x|\operatorname{erfc}(\mu|x|)
    -\frac{1}{\sqrt\pi\mu}e^{-\mu^2x^2}\right],
    \label{eq:app-u-erfc}
\end{equation}
one obtains
\begin{align}
    F_2(x)&=g\operatorname{sgn}(x)\operatorname{erfc}(\mu|x|),
    \label{eq:app-F-erfc}\\
    u''_{2,\mathrm{reg}}(x)&=-\frac{2g\mu}{\sqrt\pi}e^{-\mu^2x^2},
    \label{eq:app-u2-second-reg}\\
    V_{\mathrm{TC,eff}}^{(2)}(x)
    &=\frac{g\mu}{\sqrt\pi}e^{-\mu^2x^2}
    -\frac{g^2}{2}\operatorname{erfc}^2(\mu|x|).
    \label{eq:app-Veff-erfc}
\end{align}
Introduce the full-line moments
\begin{align}
    G_{pq}(\alpha)
    &=\int_{-\infty}^{\infty}\phi_p(x)\phi_q(x)e^{-\alpha x^2}\,dx,
    \label{eq:app-Gpq}\\
    C_{pq}(\mu)
    &=\int_{-\infty}^{\infty}\phi_p(x)\phi_q(x)\operatorname{erfc}^2(\mu|x|)\,dx,
    \label{eq:app-Cpq}\\
    R_{pq}(\mu)
    &=\int_{-\infty}^{\infty}\phi_p(x)\phi_q(x)\operatorname{sgn}(x)\operatorname{erfc}(\mu|x|)\,dx.
    \label{eq:app-Rpq}
\end{align}
Then
\begin{align}
    V_{\mathrm{TC,eff},mn}^{(2)}
    &=\frac{g\mu}{\sqrt\pi}G_{2m,2n}(\mu^2)-\frac{g^2}{2}C_{2m,2n}(\mu),
    \label{eq:app-V-erfc-matrix}\\
    \mathcal F^{(2)}_{pq}&=gR_{pq}(\mu),
    \label{eq:app-F-erfc-matrix}\\
    D^{(2)}_{mn}
    &=-g\sqrt n\,R_{2m,2n-1}(\mu)
      +g\sqrt{\frac{2n+1}{2}}\,R_{2m,2n+1}(\mu).
    \label{eq:app-D-erfc-matrix}
\end{align}

\subsection{Reduction to scalar Hermite moments}
\label{appsec:moment-reduction}

Write
\begin{equation}
    H_p(x)H_q(x)=\sum_{r=0}^{p+q}c_r^{(p,q)}x^r.
    \label{eq:app-hermite-product}
\end{equation}
The coefficients may be obtained from
\begin{equation}
    H_n(x)=\sum_{r=0}^{n}h_{n,r}x^r,
    \qquad
    h_{n,r}=\begin{cases}
    \displaystyle (-1)^{(n-r)/2}\frac{n!2^r}{((n-r)/2)!r!},
    & n-r\ge0\text{ and }n-r\text{ even},\\[0.8em]
    0,&\text{otherwise},
    \end{cases}
    \label{eq:app-hermite-coefficients}
\end{equation}
and the convolution
\begin{equation}
    c_r^{(p,q)}=\sum_{s=0}^{r}h_{p,s}h_{q,r-s}.
    \label{eq:app-product-convolution}
\end{equation}
The SHO normalization factor is
\begin{equation}
    \mathcal N_{pq}=\frac{1}{\sqrt{2^{p+q}p!q!\pi}}.
    \label{eq:app-Npq}
\end{equation}

For the exponential moments, define
\begin{equation}
    M_k(b)=\int_0^\infty x^k e^{-x^2-bx}\,dx,
    \qquad b>0 .
    \label{eq:app-Mk}
\end{equation}
Using parity only after the full-line expressions have been defined gives
\begin{align}
    A^{(a)}_{pq}(b)
    &=\mathcal N_{pq}\bigl[1+(-1)^{p+q}\bigr]
      \sum_{r=0}^{p+q}c_r^{(p,q)}M_{r+a}(b),
    \label{eq:app-A-eval}\\
    B^{(a)}_{pq}(b)
    &=\mathcal N_{pq}\bigl[1-(-1)^{p+q}\bigr]
      \sum_{r=0}^{p+q}c_r^{(p,q)}M_{r+a}(b).
    \label{eq:app-B-eval}
\end{align}
The scalar moments have the closed form
\begin{equation}
    M_k(b)=\frac{e^{b^2/4}}{2}
    \sum_{j=0}^{k}\binom{k}{j}\left(-\frac b2\right)^{k-j}
    \Gamma\left(\frac{j+1}{2},\frac{b^2}{4}\right),
    \label{eq:app-Mk-closed}
\end{equation}
and the recurrence
\begin{align}
    M_0(b)&=\frac{\sqrt\pi}{2}e^{b^2/4}\operatorname{erfc}\left(\frac b2\right),
    \label{eq:app-M0}\\
    M_1(b)&=\frac12-\frac b2M_0(b),
    \label{eq:app-M1}\\
    M_{k+1}(b)&=\frac{k}{2}M_{k-1}(b)-\frac b2M_k(b),
    \qquad k\ge1.
    \label{eq:app-M-recurrence}
\end{align}
The recurrence follows by integrating the derivative of $x^ke^{-x^2-bx}$ over $[0,\infty)$.

For the Gaussian part of TC-erfc,
\begin{equation}
    G_k(\alpha)=\int_0^\infty x^k e^{-(1+\alpha)x^2}\,dx
    =\frac12(1+\alpha)^{-(k+1)/2}\Gamma\left(\frac{k+1}{2}\right),
    \label{eq:app-Gk}
\end{equation}
and
\begin{equation}
    G_{pq}(\alpha)=\mathcal N_{pq}\bigl[1+(-1)^{p+q}\bigr]
    \sum_{r=0}^{p+q}c_r^{(p,q)}G_r(\alpha).
    \label{eq:app-Gpq-eval}
\end{equation}
For the erfc moments, define
\begin{align}
    E_k(\mu)&=\int_0^\infty x^k e^{-x^2}\operatorname{erfc}(\mu x)\,dx,
    \label{eq:app-Ek}\\
    E_k^{(2)}(\mu)&=\int_0^\infty x^k e^{-x^2}\operatorname{erfc}^2(\mu x)\,dx.
    \label{eq:app-E2k}
\end{align}
Then
\begin{align}
    R_{pq}(\mu)&=\mathcal N_{pq}\bigl[1-(-1)^{p+q}\bigr]
    \sum_{r=0}^{p+q}c_r^{(p,q)}E_r(\mu),
    \label{eq:app-Rpq-eval}\\
    C_{pq}(\mu)&=\mathcal N_{pq}\bigl[1+(-1)^{p+q}\bigr]
    \sum_{r=0}^{p+q}c_r^{(p,q)}E_r^{(2)}(\mu).
    \label{eq:app-Cpq-eval}
\end{align}
The single-erfc moments are initialized by
\begin{align}
    E_0(\mu)&=\frac{\sqrt\pi}{2}-\frac{1}{\sqrt\pi}\arctan\mu
    =\frac{1}{\sqrt\pi}\arctan\left(\frac1\mu\right),
    \label{eq:app-E0}\\
    E_1(\mu)&=\frac12\left(1-\frac{\mu}{\sqrt{1+\mu^2}}\right),
    \label{eq:app-E1}
\end{align}
and obey
\begin{equation}
    E_k(\mu)=\frac{k-1}{2}E_{k-2}(\mu)-\frac{\mu}{\sqrt\pi}G_{k-1}(\mu^2),
    \qquad k\ge2.
    \label{eq:app-E-recurrence}
\end{equation}
For the squared-erfc moments, define
\begin{equation}
    \widetilde E_k(\mu)=\int_0^\infty x^k e^{-(1+\mu^2)x^2}\operatorname{erfc}(\mu x)\,dx,
    \label{eq:app-Etilde}
\end{equation}
and
\begin{equation}
    \widetilde G_k(\mu)=\frac12(1+2\mu^2)^{-(k+1)/2}
    \Gamma\left(\frac{k+1}{2}\right).
    \label{eq:app-Gtilde}
\end{equation}
The modified erfc moments have initial values
\begin{align}
    \widetilde E_0(\mu)&=\frac{1}{\sqrt\pi\sqrt{1+\mu^2}}
    \arctan\left(\frac{\sqrt{1+\mu^2}}{\mu}\right),
    \label{eq:app-Etilde0}\\
    \widetilde E_1(\mu)&=\frac{1}{2(1+\mu^2)}
    \left(1-\frac{\mu}{\sqrt{1+2\mu^2}}\right),
    \label{eq:app-Etilde1}
\end{align}
and recurrence
\begin{equation}
    \widetilde E_k(\mu)=\frac{k-1}{2(1+\mu^2)}\widetilde E_{k-2}(\mu)
    -\frac{\mu}{\sqrt\pi(1+\mu^2)}\widetilde G_{k-1}(\mu),
    \qquad k\ge2.
    \label{eq:app-Etilde-recurrence}
\end{equation}
Finally,
\begin{align}
    E_0^{(2)}(\mu)&=\frac{2}{\sqrt\pi}
    \left[\arctan\sqrt{1+2\mu^2}-\arctan\mu\right],
    \label{eq:app-E20}\\
    E_1^{(2)}(\mu)&=\frac12-
    \frac{2\mu}{\pi\sqrt{1+\mu^2}}
    \arctan\left(\frac{\sqrt{1+\mu^2}}{\mu}\right),
    \label{eq:app-E21}\\
    E_k^{(2)}(\mu)&=\frac{k-1}{2}E_{k-2}^{(2)}(\mu)
    -\frac{2\mu}{\sqrt\pi}\widetilde E_{k-1}(\mu),
    \qquad k\ge2.
    \label{eq:app-E2-recurrence}
\end{align}
At large $p+q$ the convolution \eqref{eq:app-product-convolution} mixes
large terms of alternating sign, so in floating-point arithmetic the sums
\eqref{eq:app-A-eval}, \eqref{eq:app-B-eval}, \eqref{eq:app-Gpq-eval},
\eqref{eq:app-Rpq-eval}, and \eqref{eq:app-Cpq-eval} lose accuracy; they are
best evaluated in exact or extended-precision arithmetic.
Equations~\eqref{eq:app-even-matrix}, \eqref{eq:app-V-exp-matrix}, \eqref{eq:app-D-exp-matrix}, \eqref{eq:app-V-erfc-matrix}, and \eqref{eq:app-D-erfc-matrix} give the complete matrix-element computation.

\section{Details of the solvable-model analysis}

\label{app:sho-delta-details}
This appendix reports the details of the asymptotic convergence analysis
underlying Sec.~\ref{sec:SHO-delta}, together with the assumptions on which it
relies. The relevant prior analyses are those of Grining
\emph{et al.}~\cite{grining_many_2015} for the bare problem and Jeszenszki
\emph{et al.}~\cite{jeszenszki_accelerating_2018} for the transcorrelated
framework. In the two-particle
settings considered in both works, the external potential is either absent or
separable, allowing the Hamiltonian to be decomposed into center-of-mass and
relative-coordinate parts. The interacting relative Hamiltonian can therefore
be viewed as a variant of the solvable model in Sec.~\ref{sec:SHO-delta}. For
Grining \emph{et al.}, it has the same SHO-plus-contact form, up to coordinate,
coupling, and energy-zero conventions, whereas Jeszenszki \emph{et al.} studied
a distinct translation-invariant contact problem with periodic boundary
conditions. Despite these differences in the regular potential and global
boundary conditions, the relative-coordinate subproblems considered in both
works exhibit a local singularity at particle coalescence that is similar in
form to the singularity of our model at $x=0$. This type of local
nonanalyticity underlies the slow convergence of their bare basis expansions.
Grining \emph{et al.} expanded the relative state in a truncated basis of $N$
SHO functions and derived a leading bare error of order $N^{-1/2}$ for both the
Rayleigh quotient of the truncated exact state and the variational eigenvalue
of the projected Hamiltonian. Jeszenszki \emph{et al.} instead retained $M$
single-particle plane waves, whose largest momentum satisfies
$k_{\max}\propto M$, and obtained energy-error rates of $M^{-1}$ for the bare
problem and $M^{-3}$ after TC.

Neither analysis covers the case studied here: the TC Hamiltonian projected onto
the even SHO basis, for which we derive the energy-error rate and its leading
prefactor. To this end, we follow Hill's Rayleigh--Ritz
analysis~\cite{hill_rates_1985}, with the two-sided quotient appropriate to the
non-normal setting~\cite{ostrowski_convergence_1959}, and adopt the weak
formulation and asymmetric estimates of Bachmayr~\cite{bachmayr_hyperbolic_2012}
and Yserentant~\cite{yserentant_mixed_2011}, adapted to the present model. The
analysis rests on a single {assumption} specifying a simple Ritz branch
converging to $E_0$ and a uniform bound on its reduced resolvent
(Assumption~\ref{asm:branch}).

\subsection{Setting, weak formulation, and basic estimates}
\label{app:E:setup}

We first recall the setting. Let
$\hat H=\hat H_{\mathrm{SHO}}+g\delta(x)$ with
$g\in\mathbb R\setminus\{0\}$. We restrict our analysis to the even sector.
The ground-state eigenvalue $E_0$ is simple, and its eigenfunction
$\Psi_0$ is even and satisfies the contact condition $\Psi_0'(0^+)=g\Psi_0(0)$.
We choose $\Psi_0$ real and fix its sign by $\Psi_0(0)>0$.

The exact ground state of $\hat H$ is $\Psi_0(x)=C D_\nu(\sqrt2|x|)$, where
$D_\nu$ is the parabolic-cylinder function of order $\nu=E_0-\tfrac12$
and $C$ is chosen so that $\|\Psi_0\|^2=1$.
Here $E_0$ is the lowest solution of Eq.~\eqref{eq:sho-even-spectrum}.
For the properties of parabolic-cylinder functions, see
DLMF~\cite[\S\S~\href{https://dlmf.nist.gov/12.2}{12.2},
\href{https://dlmf.nist.gov/12.8}{12.8}, and
\href{https://dlmf.nist.gov/12.9}{12.9}]{NIST:DLMF}.
For $x>0$, these give
\[
 \begin{aligned}
 \Psi_0'(x)&=x\Psi_0(x)-\sqrt2 C D_{\nu+1}(\sqrt2x),\\
 D_\nu(\sqrt2x)&\sim(\sqrt2x)^\nu e^{-x^2/2}\qquad(x\to+\infty).
 \end{aligned}
\]
The asymptotic formula also applies at the fixed order $\nu+1$.
Repeated differentiation of $\Psi_0''=(x^2-2E_0)\Psi_0$ expresses every
higher derivative as a linear combination of $\Psi_0$ and $\Psi_0'$
with polynomial coefficients.

For fixed $\mu>0$, let $u$ denote either of the two correlators used in this
work; both of them are real, even, bounded, and Lipschitz, with $u'(0^+)=g$.
In what follows, we shift $u\to u-u(0)$ so that $u(0)=0$; only the erfc form requires this
shift. We set $F=u'$ and denote by $u''_{\mathrm{reg}}$ the ordinary second
derivative away from the origin; we then have $u''=u''_{\mathrm{reg}}+2g\delta$
in the distributional sense. Both $F$ and $u''_{\mathrm{reg}}$ are bounded
on the two half-lines. The TC Hamiltonian can then be written as
\begin{equation}
\hat H_{\mathrm{TC}}= e^{-u} \hat H e^{u} =\hat H_{\mathrm{SHO}}+\hat W,
\end{equation}
with
$V_{\mathrm{TC,eff}}=-\tfrac12u''_{\mathrm{reg}}-\tfrac12F^2$ and
$\hat W=-F\partial_x+V_{\mathrm{TC,eff}}$.
The transformed ground states are $\Phi_0^{R,L}=e^{\mp u}\Psi_0$ with
$\inner{\Phi_0^L}{\Phi_0^R}=\|\Psi_0\|^2=1$.
The right state satisfies $(\Phi_0^R)'(0^+)=0$, whereas the left state
satisfies $(\Phi_0^L)'(0^+)=2g\Phi_0^L(0)=2g\Psi_0(0)$.

Following the notation of
Bachmayr~\cite[Sec.~2]{bachmayr_hyperbolic_2012} and
Yserentant~\cite[Sec.~2]{yserentant_mixed_2011}, we use $a(\cdot,\cdot)$
for the forms in the weak formulation.
All forms below are defined for even functions $f, v \in H^1(\mathbb R)$
with $xf,xv \in L^2(\mathbb R)$; the inner product is that of $L^2(\mathbb R)$,
antilinear in the first argument.
This formulation involves only first derivatives and allows both transformed
states to be treated in the same space, without imposing their contact
conditions on the test functions.
For the SHO Hamiltonian, integration by parts gives us the form
\begin{equation}
a_{\mathrm{SHO}}(f,v)=\tfrac12\inner{f'}{v'}+\tfrac12\inner{xf}{xv}= \inner{f}{\hat H_{\mathrm{SHO}}v},
\label{eq:sho-form}      
\end{equation}
where the last equality holds when $\hat H_{\mathrm{SHO}}v\in L^2$.
 
The half-line equation and the contact condition give the weak equation
of $\hat H$,
\begin{equation}
 a_{\mathrm{SHO}}(f,\Psi_0)+g\overline{f(0)}\Psi_0(0)
 =E_0\inner{f}{\Psi_0},
\end{equation}
for every test function $f$.

For the weak formulation of $\hat H_{\mathrm{TC}}$, we define the associated form
\begin{equation}
\label{eq:E-tc-form}
 \begin{aligned}
 a_{\mathrm{TC}}(f,v)&:=a_{\mathrm{SHO}}(e^{-u}f,e^uv)
                      +g\overline{f(0)}v(0)\\
 &=a_{\mathrm{SHO}}(f,v)+\inner{f}{\hat Wv}.
 \end{aligned}
\end{equation}
The second equality follows by integration by parts,
with the jump in $F$ cancelling the contact term.
As in Bachmayr's TC formulation~\cite[Props.~2.1 and 2.3]{bachmayr_hyperbolic_2012},
the physical weak equation gives the following right and left weak identities
for the exact transformed states:
\begin{equation}
 \begin{aligned}
 a_{\mathrm{TC}}(f,\Phi_0^R)&=E_0\inner{f}{\Phi_0^R},\\
 a_{\mathrm{TC}}(\Phi_0^L,v)&=E_0\inner{\Phi_0^L}{v},
 \end{aligned}
 \label{eq:E-biorth}
\end{equation}
for all $f,v$ in this space.

Finally, we need to estimate the correction $\inner{f}{\hat Wv}$
in Eq.~\eqref{eq:E-tc-form}.
We write $\|\cdot\|$ for the $L^2(\mathbb R)$ norm and $\|\cdot\|_\infty$
for the $L^\infty(\mathbb R)$ norm.
The $L^\infty$ bounds on the coefficients in
$\hat Wv=-Fv'+V_{\mathrm{TC,eff}}v$ reduce the $L^2$ estimate of $\hat Wv$
to control of $\|v'\|$ and $\|v\|$.
We therefore use the energy norm associated with $a_{\mathrm{SHO}}$, defined by
\begin{equation}
 \|v\|_{\mathrm{en}}^2=a_{\mathrm{SHO}}(v,v)+\|v\|^2 .
 \label{eq:E-en-norm}
\end{equation}
This gives $\|v'\|\le\sqrt2\|v\|_{\mathrm{en}}$ and
$\|v\|\le\|v\|_{\mathrm{en}}$.
Combining these bounds with the triangle inequality gives the first estimate
below. Cauchy--Schwarz then gives the second, which requires only $f\in L^2$:
\begin{equation}
 \begin{aligned}
 \|\hat Wv\|&\le\|F\|_\infty\|v'\|
             +\|V_{\mathrm{TC,eff}}\|_\infty\|v\|
             \le C_W\|v\|_{\mathrm{en}},\\
 |\inner{f}{\hat Wv}|&\le\|f\|\,\|\hat Wv\|
             \le C_W\|f\|\,\|v\|_{\mathrm{en}},
 \end{aligned}
 \label{eq:E-K2}
\end{equation}
where $C_W=\sqrt2\|F\|_\infty+\|V_{\mathrm{TC,eff}}\|_\infty$.

\subsection{Exact error formula via Schur-complement reduction}
\label{app:E:exact}

The numerical calculations in Sec.~\ref{sec:SHO-delta} replace the full
operator $\hat H_{\mathrm{TC}}$ by its finite-basis projection, as described
in Sec.~\ref{sec:finite-dimension-tc}.
Our aim is to estimate the energy error $E_0^{(N)}-E_0$ in terms of the
approximation errors of the exact left and right states in this basis.
Following Hill's analysis~\cite[Sec.~III]{hill_rates_1985}, we use the
projected exact right state as a reference direction within the retained
space. For the present non-Hermitian problem, we define the complementary
subspace using the projected exact left state.
Schur elimination of the numerical eigenvector's component in this
subspace gives an energy identity.
The weak identities~\eqref{eq:E-biorth} of the preceding subsection then
express the energy error through the tail components of both exact states.

Let $w_n=\phi_{2n}$ be the complete orthonormal even SHO basis of
Appendix~\ref{app:sho-integrals}, with energies $\varepsilon_n=2n+\tfrac12$.
We retain the first $N=N_{\mathrm{even}}$ functions and let $P_N$ be the
$L^2$-orthogonal projection onto their span, with $Q_N=I-P_N$.
For $f,v$ in the form space, write
$f_n=\inner{w_n}{f}$ and $v_n=\inner{w_n}{v}$.
Orthonormality, $\hat H_{\mathrm{SHO}}w_n=\varepsilon_nw_n$, and
Eq.~\eqref{eq:E-en-norm} give
\begin{equation}
 \|v\|^2=\sum_n|v_n|^2,
 \qquad
 a_{\mathrm{SHO}}(f,v)=\sum_n\varepsilon_n\overline{f_n}v_n,
 \qquad
 \|v\|_{\mathrm{en}}^2=\sum_n(\varepsilon_n+1)|v_n|^2 ,
 \label{eq:E-coeff-norms}
\end{equation}
where the sums run over all $n\ge0$.
The corresponding formulas for $P_Nf,P_Nv$ restrict the sums to $n<N$;
those for $Q_Nf,Q_Nv$ restrict them to $n\ge N$.
We consider the ground-state eigenvalue problem for the finite SHO matrix
$H_{\mathrm{TC}}^{(N)}$:
\[
 \quad H_{\mathrm{TC}}^{(N)}\Phi_0^{(N)}=E_0^{(N)}\Phi_0^{(N)}.
\]

The right eigenvector $\Phi_0^{(N)}$ need not be proportional to the projected
exact state $P_N\Phi_0^R$.
To relate them, we introduce an oblique projection $\Pi_N$ onto the direction
$P_N\Phi_0^R$ within the retained space.
Its complement $\Pi_N^\perp$ selects vectors that pair to zero with
$P_N\Phi_0^L$.
When the projected states have nonzero overlap, define
\[
 \Pi_N=\frac{\ket{P_N\Phi_0^R}\bra{P_N\Phi_0^L}}{\sigma_N},
 \qquad \Pi_N^\perp=P_N-\Pi_N,
\]
where $\sigma_N:=\inner{P_N\Phi_0^L}{P_N\Phi_0^R}=\inner{P_N\Phi_0^L}{\Phi_0^R}$.
At fixed $N$, suppose that the generalized inverse
\[
 A_N=\bigl(\Pi_N^\perp H_{\mathrm{TC}}^{(N)}\Pi_N^\perp
       -E_0^{(N)}\Pi_N^\perp\bigr)^{-1}
 \quad\text{on }\Ran\Pi_N^\perp
\]
exists. This implies $\Pi_N\Phi_0^{(N)}\ne0$; otherwise, the eigenvalue
equation would contradict this invertibility.
Since $\Pi_N\Phi_0^{(N)}$ is a nonzero multiple of $P_N\Phi_0^R$,
we may rescale the eigenvector so that
$\Pi_N\Phi_0^{(N)}=P_N\Phi_0^R$.
The oblique splitting then gives
\[
 \begin{aligned}
 \Phi_0^{(N)}&=P_N\Phi_0^R+s_N,
 &s_N&=\Pi_N^\perp\Phi_0^{(N)},\\
 \inner{P_N\Phi_0^L}{s_N}&=0,
 &\inner{P_N\Phi_0^L}{\Phi_0^{(N)}}&=\sigma_N.
 \end{aligned}
\]

Substituting this decomposition into the eigenvalue equation and applying
$\Pi_N^\perp$, we use $\Pi_N^\perp P_N\Phi_0^R=0$ and
$\Pi_N^\perp s_N=s_N$ to obtain
\[
 \bigl(\Pi_N^\perp H_{\mathrm{TC}}^{(N)}\Pi_N^\perp
       -E_0^{(N)}\Pi_N^\perp\bigr)s_N
 =-\Pi_N^\perp H_{\mathrm{TC}}^{(N)}P_N\Phi_0^R.
\]
The operator on the left is the one inverted by $A_N$, so
$s_N=-A_N\Pi_N^\perp H_{\mathrm{TC}}^{(N)}P_N\Phi_0^R$.
Pairing the eigenvalue equation with $P_N\Phi_0^L$ and substituting this
expression yields the following energy identity for the projected eigenvalue:
\begin{equation}
 \begin{aligned}
 E_0^{(N)}
 &=\frac{\inner{P_N\Phi_0^L}{H_{\mathrm{TC}}^{(N)}P_N\Phi_0^R}}{\sigma_N}
   +\frac{\inner{P_N\Phi_0^L}{H_{\mathrm{TC}}^{(N)}s_N}}{\sigma_N}\\
 &=\frac{\inner{P_N\Phi_0^L}{H_{\mathrm{TC}}^{(N)}P_N\Phi_0^R}}{\sigma_N}
   -\frac{\inner{P_N\Phi_0^L}{H_{\mathrm{TC}}^{(N)}A_N\Pi_N^\perp
                      H_{\mathrm{TC}}^{(N)}P_N\Phi_0^R}}{\sigma_N}.
 \end{aligned}
 \label{eq:E-scalar}
\end{equation}
The first term is the two-sided Rayleigh quotient of the projected exact
states~\cite{ostrowski_convergence_1959}.
The second is the Schur correction from their coupling to the rest of the
retained space, and the second line displays $A_N$ between the two coupling
factors.

To obtain $E_0^{(N)}-E_0$ from Eq.~\eqref{eq:E-scalar}, we first write
$E_0$ with the same denominator $\sigma_N$.
Use the right weak identity~\eqref{eq:E-biorth} with $f=P_N\Phi_0^L$.
We then have
\[
 E_0=\frac{a_{\mathrm{TC}}(P_N\Phi_0^L,\Phi_0^R)}{\sigma_N}.
\]
Subtracting this from Eq.~\eqref{eq:E-scalar} gives
\begin{equation}
 (E_0^{(N)}-E_0)\sigma_N
 =\inner{P_N\Phi_0^L}{H_{\mathrm{TC}}^{(N)}P_N\Phi_0^R}
 -a_{\mathrm{TC}}(P_N\Phi_0^L,\Phi_0^R)
 +\inner{P_N\Phi_0^L}{H_{\mathrm{TC}}^{(N)}s_N}.
 \label{eq:error-formula}
\end{equation}
The finite SHO matrix acts on the retained space as
\[
 H_{\mathrm{TC}}^{(N)}=P_N\hat H_{\mathrm{TC}}P_N
 =P_N(\hat H_{\mathrm{SHO}}+\hat W)P_N.
\]
The projections can be removed inside the pairing because
$P_N^\dagger=P_N$ and $P_N^2=P_N$.
Both projected states are finite SHO expansions, so the SHO form identity
applies with $f=P_N\Phi_0^L$ and $v=P_N\Phi_0^R$. Thus,
\[
 \begin{aligned}
 \inner{P_N\Phi_0^L}{H_{\mathrm{TC}}^{(N)}P_N\Phi_0^R}
 &=\inner{P_N\Phi_0^L}
          {P_N(\hat H_{\mathrm{SHO}}+\hat W)P_N\Phi_0^R}\\
 &=\inner{P_N\Phi_0^L}{\hat H_{\mathrm{SHO}}P_N\Phi_0^R}
   +\inner{P_N\Phi_0^L}{\hat WP_N\Phi_0^R}\\
 &=a_{\mathrm{SHO}}(P_N\Phi_0^L,P_N\Phi_0^R)
   +\inner{P_N\Phi_0^L}{\hat WP_N\Phi_0^R}
   \qquad\text{by Eq.~\eqref{eq:sho-form}}\\
 &=a_{\mathrm{TC}}(P_N\Phi_0^L,P_N\Phi_0^R)
   \qquad\text{by Eq.~\eqref{eq:E-tc-form}.}
 \end{aligned}
\]
Using $\Phi_0^R=P_N\Phi_0^R+Q_N\Phi_0^R$, the first two terms become
\[
 a_{\mathrm{TC}}(P_N\Phi_0^L,P_N\Phi_0^R)
 -a_{\mathrm{TC}}(P_N\Phi_0^L,\Phi_0^R)
 =-a_{\mathrm{TC}}(P_N\Phi_0^L,Q_N\Phi_0^R).
\]
The right argument now contains the tail component.  
It remains to rewrite the projected left component in terms of the left tail.
We therefore write $P_N\Phi_0^L=\Phi_0^L-Q_N\Phi_0^L$ and apply the left
weak identity to the exact left state:
\[
 \begin{aligned}
 &-a_{\mathrm{TC}}(P_N\Phi_0^L,Q_N\Phi_0^R)\\
 &\quad=a_{\mathrm{TC}}(Q_N\Phi_0^L,Q_N\Phi_0^R)
         -a_{\mathrm{TC}}(\Phi_0^L,Q_N\Phi_0^R)
         &&(P_N\Phi_0^L=\Phi_0^L-Q_N\Phi_0^L)\\
 &\quad=a_{\mathrm{TC}}(Q_N\Phi_0^L,Q_N\Phi_0^R)
         -E_0\inner{\Phi_0^L}{Q_N\Phi_0^R}
 &&\text{(left weak identity~\eqref{eq:E-biorth})}\\
 &\quad=a_{\mathrm{TC}}(Q_N\Phi_0^L,Q_N\Phi_0^R)
         -E_0\inner{Q_N\Phi_0^L}{Q_N\Phi_0^R}
 &&(\inner{P_N\Phi_0^L}{Q_N\Phi_0^R}=0).
 \end{aligned}
\]
By Eq.~\eqref{eq:E-tc-form} (now with $f=Q_N\Phi_0^L$ and $v=Q_N\Phi_0^R$),
the last expression is
$T_N^{(0)}+T_N^{(W)}$, with
\[
 \begin{aligned}
 T_N^{(0)}&:=a_{\mathrm{SHO}}(Q_N\Phi_0^L,Q_N\Phi_0^R)
       -E_0\inner{Q_N\Phi_0^L}{Q_N\Phi_0^R},\\
 T_N^{(W)}&:=\inner{Q_N\Phi_0^L}{\hat WQ_N\Phi_0^R}.
 \end{aligned}
\]
Having rewritten the first two terms in Eq.~\eqref{eq:error-formula},
we now turn to the remaining Schur term
\[
 \inner{P_N\Phi_0^L}{H_{\mathrm{TC}}^{(N)}s_N}
 =-\inner{P_N\Phi_0^L}{H_{\mathrm{TC}}^{(N)}A_N\Pi_N^\perp H_{\mathrm{TC}}^{(N)}P_N\Phi_0^R}.
\]
We use a similar strategy as above. We first rewrite the contribution
involving the projected right state, $H_{\mathrm{TC}}^{(N)}P_N\Phi_0^R$, in
terms of the right tail $Q_N\Phi_0^R$, and then rewrite the projected left
state $P_N\Phi_0^L$ in terms of the left tail $Q_N\Phi_0^L$.

Let $f\in\Ran P_N$. Since the SHO operator preserves the retained space,
$\hat H_{\mathrm{SHO}}f\in\Ran P_N$. Hence, for any $v$ in the form space,
Hermitian symmetry of $a_{\mathrm{SHO}}$ and Eq.~\eqref{eq:sho-form} give
\[
 a_{\mathrm{SHO}}(f,Q_Nv)
 =\overline{a_{\mathrm{SHO}}(Q_Nv,f)}
 =\overline{\inner{Q_Nv}{\hat H_{\mathrm{SHO}}f}}=0 .
\]
Using this with $v=\Phi_0^R$, together with $\inner{f}{Q_N\Phi_0^R}=0$,
we obtain
\[
 \begin{aligned}
 \inner{f}{H_{\mathrm{TC}}^{(N)}P_N\Phi_0^R}
 &=a_{\mathrm{TC}}(f,P_N\Phi_0^R)
 &&\text{(finite-basis form identity)}\\
 &=a_{\mathrm{TC}}(f,\Phi_0^R)-a_{\mathrm{TC}}(f,Q_N\Phi_0^R)
 &&\bigl(\Phi_0^R=P_N\Phi_0^R+Q_N\Phi_0^R\bigr)\\
 &=E_0\inner{f}{P_N\Phi_0^R}-\inner{f}{\hat WQ_N\Phi_0^R}
 &&\text{(right weak identity~\eqref{eq:E-biorth};
   }a_{\mathrm{SHO}}(f,Q_N\Phi_0^R)=0\text{)}.
 \end{aligned}
\]
Since this holds for every $f\in\Ran P_N$, and both sides belong to the
retained space,
\[
 H_{\mathrm{TC}}^{(N)}P_N\Phi_0^R
 =E_0P_N\Phi_0^R-P_N\hat WQ_N\Phi_0^R .
\]
Substituting this into
$s_N=-A_N\Pi_N^\perp H_{\mathrm{TC}}^{(N)}P_N\Phi_0^R$ and using
$\Pi_N^\perp P_N\Phi_0^R=0$ gives
\begin{equation}
 s_N=A_N\Pi_N^\perp P_N\hat WQ_N\Phi_0^R ,
 \label{eq:E-sN-tail}
\end{equation}
so $s_N$ is generated by the right tail $Q_N\Phi_0^R$.

We now turn to the projected left state. First, $\inner{\Phi_0^L}{s_N}
=\inner{P_N\Phi_0^L}{s_N}+\inner{Q_N\Phi_0^L}{s_N}=0$, where the first term
vanishes by construction of $s_N$ and the second by $P_N$--$Q_N$
orthogonality. Second, $a_{\mathrm{SHO}}(Q_N\Phi_0^L,s_N)
=\inner{Q_N\Phi_0^L}{\hat H_{\mathrm{SHO}}s_N}=0$, since $s_N\in\Ran P_N$
and the SHO operator preserves the retained space. The Schur term therefore
reduces in weak form to
\[
 \begin{aligned}
 \inner{P_N\Phi_0^L}{H_{\mathrm{TC}}^{(N)}s_N}
 &=a_{\mathrm{TC}}(P_N\Phi_0^L,s_N)
 &&\text{(finite-basis form identity)}\\
 &=a_{\mathrm{TC}}(\Phi_0^L,s_N)-a_{\mathrm{TC}}(Q_N\Phi_0^L,s_N)
 &&\bigl(P_N\Phi_0^L=\Phi_0^L-Q_N\Phi_0^L\bigr)\\
 &=-a_{\mathrm{TC}}(Q_N\Phi_0^L,s_N)
 &&\text{(left weak identity~\eqref{eq:E-biorth};
   }\inner{\Phi_0^L}{s_N}=0\text{)}\\
 &=-\inner{Q_N\Phi_0^L}{\hat Ws_N}
 &&\bigl(a_{\mathrm{SHO}}(Q_N\Phi_0^L,s_N)=0\bigr).
 \end{aligned}
\]
Combining with Eq.~\eqref{eq:E-sN-tail}, the Schur contribution is $-R_N$,
\begin{equation}
 R_N:=\inner{Q_N\Phi_0^L}{\hat Ws_N}
 =\inner{Q_N\Phi_0^L}{\hat WA_N\Pi_N^\perp P_N\hat WQ_N\Phi_0^R}.
 \label{eq:E-Schur-RN}
\end{equation}
Combining these expressions with Eq.~\eqref{eq:error-formula}
gives the exact energy-error identity in terms of the left and right tails:
\begin{equation}
 \boxed{\begin{aligned}
 (E_0^{(N)}-E_0)\sigma_N&=T_N^{(0)}+T_N^{(W)}-R_N,\\[3pt]
 T_N^{(0)}&=a_{\mathrm{SHO}}(Q_N\Phi_0^L,Q_N\Phi_0^R)
       -E_0\inner{Q_N\Phi_0^L}{Q_N\Phi_0^R},\\
 T_N^{(W)}&=\inner{Q_N\Phi_0^L}{\hat WQ_N\Phi_0^R},\\
 R_N&=\inner{Q_N\Phi_0^L}{\hat WA_N\Pi_N^\perp P_N\hat WQ_N\Phi_0^R}.
 \end{aligned}}
 \label{eq:E-error-id}
\end{equation}

Before proceeding to the estimates, we emphasize the distinction between
fixed-$N$ existence and uniform-in-$N$ stability.
The derivation above requires only that $A_N$ exist at each fixed $N$,
whereas the estimates below require a uniform bound on its action as $N$
grows.  Since $R_N$ contains $\hat W s_N$ and $\hat W$ contains a first derivative,
Eq.~\eqref{eq:E-K2} requires control of $s_N$ in the energy norm.
We therefore make the following assumption.
\begin{assumption}[Ritz sequence and uniform reduced-inverse stability]
\label{asm:branch}
There exists $N_0$ such that, for all $N\ge N_0$, the Ritz value
$E_0^{(N)}$ selected in Sec.~\ref{sec:SHO-delta} is real and simple,
and the sequence satisfies
\[
E_0^{(N)}\longrightarrow E_0
\qquad (N\to\infty).
\]
For all $N\ge N_0$, the reduced inverse $A_N$ defined above exists,
and there is a constant $C_A>0$, independent of $N$, such that
\begin{equation}
 \|A_N\Pi_N^\perp P_N z\|_{\mathrm{en}}
 \le C_A\|z\|,
 \qquad
 z\in L^2_{\mathrm{even}}(\mathbb R).
 \label{eq:E-CA}
\end{equation}
\end{assumption}

\begin{remark}
Assumption~\ref{asm:branch} plays a role analogous to the 
reduced-inverse control established by Hill in his Hermitian
Rayleigh--Ritz analysis~\cite[Sec.~III~B and App.~B]{hill_rates_1985}.
There, the bound follows from the variational and positivity structure;
for the non-Hermitian TC projection, we retain the corresponding
stability requirement as an explicit assumption.
\end{remark}

Having established the exact error identity~\eqref{eq:E-error-id} and stated
the additional assumption needed for the Schur term, we now
treat the three contributions separately.
For $T_N^{(0)}$, we first expand both tails in the SHO basis $\{w_n\}$ and set
\[
 c_n^L:=\inner{w_n}{\Phi_0^L},
 \qquad
 c_n^R:=\inner{w_n}{\Phi_0^R}.
\]
With this notation,
\begin{equation}
\begin{aligned}
 T_N^{(0)}
 &=a_{\mathrm{SHO}}(Q_N\Phi_0^L,Q_N\Phi_0^R)
   -E_0\inner{Q_N\Phi_0^L}{Q_N\Phi_0^R}
 &&\text{(definition of $T_N^{(0)}$)}
 \\
 &=\sum_{n\ge N}(\varepsilon_n-E_0)
   \overline{c_n^L}c_n^R
 &&\text{(Eq.~\eqref{eq:E-coeff-norms})}.
\end{aligned}
\label{eq:E-sho-diag}
\end{equation}
For $T_N^{(W)}$, using Eq.~\eqref{eq:E-K2}, we obtain
\begin{equation}
\begin{aligned}
 |T_N^{(W)}|
 &=\bigl|\inner{Q_N\Phi_0^L}{\hat WQ_N\Phi_0^R}\bigr|
 &&\text{(definition of $T_N^{(W)}$)}
 \\
 &\le C_W\|Q_N\Phi_0^L\|\,
       \|Q_N\Phi_0^R\|_{\mathrm{en}}
 &&\text{(Eq.~\eqref{eq:E-K2})}
 \\
 &=C_W
 \left(
   \sum_{n\ge N}|c_n^L|^2
 \right)^{1/2}
 \left(
   \sum_{n\ge N}(\varepsilon_n+1)|c_n^R|^2
 \right)^{1/2}
 &&\text{(Eq.~\eqref{eq:E-coeff-norms})}.
\end{aligned}
\label{eq:E-TW-bound}
\end{equation}
For the Schur term $R_N$, similarly we have:
\begin{equation}
\begin{aligned}
 |R_N|
 &=\left|
 \inner{Q_N\Phi_0^L}
 {\hat WA_N\Pi_N^\perp P_N\hat WQ_N\Phi_0^R}
 \right|
 &&\text{(Eq.~\eqref{eq:E-Schur-RN})}
 \\
 &\le C_W\|Q_N\Phi_0^L\|\,
 \|A_N\Pi_N^\perp P_N\hat WQ_N\Phi_0^R\|_{\mathrm{en}}
 &&\text{(Eq.~\eqref{eq:E-K2})}
 \\
 &\le C_AC_W\|Q_N\Phi_0^L\|\,
 \|\hat WQ_N\Phi_0^R\|
 &&\text{(Eq.~\eqref{eq:E-CA})}
 \\
 &\le C_AC_W^2\|Q_N\Phi_0^L\|\,
 \|Q_N\Phi_0^R\|_{\mathrm{en}}
 &&\text{(Eq.~\eqref{eq:E-K2})}
 \\
 &=C_AC_W^2
 \left(\sum_{n\ge N}|c_n^L|^2\right)^{1/2}
 \left(\sum_{n\ge N}(\varepsilon_n+1)|c_n^R|^2\right)^{1/2}
 &&\text{(Eq.~\eqref{eq:E-coeff-norms})}.
\end{aligned}
\label{eq:E-RN-bound}
\end{equation}
Finally, the normalization
$\inner{\Phi_0^L}{\Phi_0^R}=1$ and $P_N$--$Q_N$ orthogonality give
\begin{equation}
\sigma_N
=\inner{P_N\Phi_0^L}{P_N\Phi_0^R}
=1-\inner{Q_N\Phi_0^L}{Q_N\Phi_0^R}
=1-\sum_{n\ge N}\overline{c_n^L}c_n^R .
\label{eq:E-overlap-tail}
\end{equation}
The tail sum in Eq.~\eqref{eq:E-overlap-tail} vanishes
as $N\to\infty$ by completeness, so $\sigma_N\to1$ and hence
$\sigma_N\ne0$ for sufficiently large $N$.

At this point, the three terms in Eq.~\eqref{eq:E-error-id}, together with
$\sigma_N$, have been reduced to quantities determined by the SHO
coefficient tails of the exact left and right states. The next subsection 
derives the large-$n$ asymptotics of these
coefficients from the local behavior of the exact states at the origin.
Section~\ref{app:E:rate} then compares the resulting asymptotic orders.

\subsection{Local cusp structure and coefficient tails}
\label{app:E:asympt}

To determine the large-$n$ behavior of the coefficients
$c_n^L$ and $c_n^R$, we analyze the local structure of the exact left and
right states near the contact point $x=0$. Since all functions considered here are even, we first expand them on
$x>0$ and then extend the expansions evenly across the origin. Thus odd
powers in the one-sided expansion appear as $|x|^{2k+1}$, whereas even
powers remain $x^{2k}$.

The exact physical ground state of Sec.~\ref{app:E:setup} has
\[
 \frac{\Psi_0(x)}{\Psi_0(0)}
 =1+g|x|-E_0x^2-\frac{gE_0}{3}|x|^3+O(x^4).
\]
For the two correlators, taken in the shifted form $u(0)=0$ of
Sec.~\ref{app:E:setup},
\[
 \begin{aligned}
 u_1(x)
 &=g|x|-g\mu x^2+\frac{g\mu^2}{2}|x|^3+O(x^4),\\
 u_2(x)
 &=g|x|-\frac{g\mu}{\sqrt\pi}x^2+O(x^4).
 \end{aligned}
\]
Note that the second expansion (erfc correlator) contains no odd powers
beyond the linear term:
$u_2-g|x|$ is an even power series in $x$ near the origin.
Substituting these expansions into the exact right state $\Phi_0^R=e^{-u}\Psi_0$ gives
\[
 \begin{aligned}
 \frac{\Phi_{0,\exp}^R(x)}{\Psi_0(0)}
 &=1+\left(g\mu-E_0-\frac{g^2}{2}\right)x^2
      +\beta_{\exp}(\mu)|x|^3+O(x^4),\\
 \frac{\Phi_{0,\operatorname{erfc}}^R(x)}{\Psi_0(0)}
 &=1+\left(\frac{g\mu}{\sqrt\pi}-E_0-\frac{g^2}{2}\right)x^2
      +\beta_{\operatorname{erfc}}|x|^3+O(x^4),
 \end{aligned}
\]
where the subscript identifies the correlator and 
\begin{equation}
 \beta_{\exp}(\mu)=\frac{g(2E_0+g^2)}{3}-\frac{g\mu^2}{2},\qquad
 \beta_{\operatorname{erfc}}=\frac{g(2E_0+g^2)}{3}.
 \label{eq:E-beta}
\end{equation}
For both correlators, applying $e^{-u}$ to $\Psi_0$ cancels the linear
cusp. Since the quadratic term is smooth across
$x=0$, the first potentially nonanalytic term is
$\Psi_0(0)\beta|x|^3$, with $\beta$ given by Eq.~\eqref{eq:E-beta}.
Hence
\[
 (\Phi_0^R)'(0^+)=0,
 \qquad
 (\Phi_0^R)'''(0^+)=6\beta\Psi_0(0).
\]
For the left state, applying $e^u$ to $\Psi_0$ instead reinforces the
linear cusp:
\[
 \frac{\Phi_0^L(x)}{\Psi_0(0)}
 =1+2g|x|+O(x^2).
\]
Thus the TC transformation produces asymmetric contact regularity between
the exact left and right states. Nevertheless, both satisfy the
corresponding weak identities in Eq.~\eqref{eq:E-biorth} with the same
eigenvalue $E_0$; this asymmetry will reappear in their coefficient
decay below.

We next translate the local behavior into the large-$n$ asymptotics of
the SHO coefficients. We first keep the state $G$ general, so that the
bare and TC cases can be analyzed under the same framework.
Since our analysis is restricted to the even sector
(Sec.~\ref{app:E:setup}), both the states and the basis functions considered here
are even. Hence each coefficient can be written as twice its integral
over the positive half-line.
For an even $G$, we therefore write
\begin{equation}
C_n(G)=2\int_0^\infty w_n(x)G(x)\,dx,
\label{eq:coeff_eq}
\end{equation}
so that $c_n^{L,R}=C_n(\Phi_0^{L,R})$.
This half-line representation allows us to relate the coefficients to
the local behavior of $G$ at the contact point. To make this connection explicit, 
it is convenient to introduce the shifted 
half-line operator and use the differential equation satisfied by the SHO
basis functions:
\begin{equation}
   \mathcal L_+
   :=\hat H_{\mathrm{SHO}}-E_0
   =-\tfrac12\partial_x^2+\tfrac12x^2-E_0,
   \qquad x>0,
\end{equation}
with $\lambda_n:=\varepsilon_n-E_0$.
Since
$\mathcal L_+w_n=\lambda_n w_n$,
we can rewrite the coefficient equation as:
\begin{equation}
   \lambda_n C_n(G)
   =2\int_0^\infty \lambda_n w_n(x)G(x)\,dx
   =2\int_0^\infty (\mathcal L_+w_n)(x)G(x)\,dx.
   \label{eq:coeff-eq}
\end{equation}
Integration by parts then gives
\begin{equation}
\lambda_n C_n(G)
 =C_n(\mathcal L_+G)+[w_nG'-w_n'G]_0^\infty.
 \label{eq:ibp-form}
\end{equation} 
For the states considered here, the boundary contribution at infinity
vanishes. Since $w_n$ is even, $w_n'(0)=0$, and hence
\begin{equation}
   \lambda_n C_n(G)
   =C_n(\mathcal L_+G)-w_n(0)G'(0^+).
   \label{eq:E-ibp}
\end{equation}
Now, for the bare ground state $\Psi_0$, Eq.~\eqref{eq:E-ibp} closes immediately,
and we obtain an exact coefficient formula for $\Psi_0$:
\begin{align}
   C_n(\Psi_0)
   &=\frac{C_n(\mathcal L_+\Psi_0)
            -w_n(0)\Psi_0'(0^+)}{\lambda_n}
   \notag\\
   &=-\frac{g\Psi_0(0)w_n(0)}{\lambda_n}
   \qquad
   \text{using }\;
   \substack{
      \mathcal L_+\Psi_0=0 \ \text{on } x>0,\\
      \Psi_0'(0^+)=g\Psi_0(0)
   }
   \notag\\
   &=-\frac{(-1)^n g\Psi_0(0)}{\pi^{1/4}}
     \frac{\sqrt{(2n)!}}
          {2^n n!\left(2n+\tfrac12-E_0\right)}.
   \label{eq:E-bare-coeff}
\end{align}

For the TC states ($\Phi_0^R,\Phi_0^L$), by contrast,
$\mathcal L_+G$ does not vanish in general, so
Eq.~\eqref{eq:E-ibp} does not close immediately.
For the two states considered here, their half-line smoothness and
decay allow Eq.~\eqref{eq:E-ibp} to be applied successively to the
finite orders needed below. After $m$ applications,
\begin{equation}
   C_n(G)
   =
   -w_n(0)\sum_{j=0}^{m-1}
      \frac{(\mathcal L_+^jG)'(0^+)}{\lambda_n^{j+1}}
   +R_{n,m}(G),
   \qquad
   R_{n,m}(G)
   :=
   \frac{C_n(\mathcal L_+^mG)}{\lambda_n^m}.
   \label{eq:E-finite-recursion}
\end{equation}
This finite iterated identity is the even-sector counterpart of the
repeated-integration-by-parts expansion underlying Hill's
Theorem~1~\cite{hill_rates_1985}.

For the right TC state, the first boundary contribution in
Eq.~\eqref{eq:E-finite-recursion} vanishes because
\[
   (\Phi_0^R)'(0^+)=0.
\]
The next contribution is determined by
\begin{align}
   (\mathcal L_+\Phi_0^R)'(0^+)
   &=
   -\tfrac12(\Phi_0^R)'''(0^+)
   -E_0(\Phi_0^R)'(0^+)
   \notag\\
   &=-3\beta\Psi_0(0),
   \label{eq:E-right-boundary}
\end{align}
where we used
$(\Phi_0^R)'''(0^+)=6\beta\Psi_0(0)$.
Thus the $j=1$ boundary term for the right state is proportional to
$\beta$ and vanishes when $\beta=0$.

For the left TC state, by contrast we have:
\[
   (\Phi_0^L)'(0^+)=2g\Psi_0(0),
\]
so the $j=0$ boundary term in
Eq.~\eqref{eq:E-finite-recursion} is already nonzero.

Thus, after retaining the first potentially nonzero boundary contribution,
the three cases (including the bare state $\Psi_0$) have the forms
\begin{equation}
\boxed{
\begin{aligned}
 C_n(\Psi_0)
 &=-\frac{g\Psi_0(0)w_n(0)}{\lambda_n},
 \\[1mm]
 C_n(\Phi_0^L)
 &=-\frac{2g\Psi_0(0)w_n(0)}{\lambda_n}
   +R_n^L,
 \\[1mm]
 C_n(\Phi_0^R)
 &=\frac{3\beta\Psi_0(0)w_n(0)}{\lambda_n^2}
   +R_n^R,
\end{aligned}}
\label{eq:E-coeff-leading-forms}
\end{equation}
Here the bare expression is exact, while the TC remainders are
\begin{equation}
 R_n^L
 :=\frac{C_n(\mathcal L_+\Phi_0^L)}{\lambda_n},
 \qquad
 R_n^R
 :=\frac{C_n(\mathcal L_+^2\Phi_0^R)}{\lambda_n^2}.
 \label{eq:E-coeff-pre-remainders}
\end{equation}
Applying Eq.~\eqref{eq:E-ibp} once more to each remainder gives
\begin{align}
 R_n^L
 &=
 -\frac{w_n(0)(\mathcal L_+\Phi_0^L)'(0^+)}{\lambda_n^2}
 +\frac{C_n(\mathcal L_+^2\Phi_0^L)}{\lambda_n^2},
 \label{eq:E-left-remainder}
 \\
 R_n^R
 &=
 -\frac{w_n(0)(\mathcal L_+^2\Phi_0^R)'(0^+)}{\lambda_n^3}
 +\frac{C_n(\mathcal L_+^3\Phi_0^R)}{\lambda_n^3}.
 \label{eq:E-right-remainder}
\end{align}

To estimate the remainder , we first use the elementary bound.
For each fixed $h\in L^2(0,\infty)$, Cauchy--Schwarz gives
\begin{equation}
   |C_n(h)|
   \leq \sqrt{2}\,\|h\|_{L^2(0,\infty)}
   =O(1),
   \label{eq:E-Cn-bound}
\end{equation}
uniformly in $n$.

Moreover,
\begin{equation}
   w_n(0)
   =\frac{(-1)^n}{\pi^{1/4}}
     \frac{\sqrt{(2n)!}}{2^n n!}
   =\frac{(-1)^n}{\sqrt{\pi}}\,n^{-1/4}
     \bigl[1+O(n^{-1})\bigr],
   \qquad
   \lambda_n
   =2n+\tfrac12-E_0
   =2n\bigl[1+O(n^{-1})\bigr].
   \label{eq:E-wn-lambda-asympt}
\end{equation}
It follows from
Eqs.~\eqref{eq:E-left-remainder}--\eqref{eq:E-right-remainder} that
\[
   R_n^L=O(n^{-2}),
   \qquad
   R_n^R=O(n^{-3}).
\]
Substituting these estimates into
Eq.~\eqref{eq:E-coeff-leading-forms} gives
\begin{equation}
\boxed{
\begin{aligned}
   C_n(\Psi_0)
   &=
   -\frac{g\Psi_0(0)}{2\sqrt{\pi}}
    (-1)^n n^{-5/4}
    +O(n^{-9/4}),
   \\[1mm]
   C_n(\Phi_0^L)
   &=
   -\frac{g\Psi_0(0)}{\sqrt{\pi}}
    (-1)^n n^{-5/4}
    +O(n^{-2}),
   \\[1mm]
   C_n(\Phi_0^R)
   &=
   \frac{3\beta\Psi_0(0)}{4\sqrt{\pi}}
   (-1)^n n^{-9/4}
   +O(n^{-3}).
\end{aligned}}
\label{eq:E-coeff-asymptotics}
\end{equation}

Note that for $\beta\neq0$ the last expression has the displayed
$n^{-9/4}$ leading term. If $\beta=0$, this contribution vanishes and
the next nonzero boundary term must be further examined.
These coefficient asymptotics imply the following tail estimates lemma:
\begin{lemma}[Tail estimates]
\label{lem:E-tail-estimates}
The coefficient asymptotics in Eq.~\eqref{eq:E-coeff-asymptotics}
imply
\begin{equation}
\begin{aligned}
   \|Q_N\Phi_0^L\|
   &=O(N^{-3/4}),
   &
   \|Q_N\Phi_0^R\|
   &=O(N^{-7/4}),
   \\
   \|Q_N\Phi_0^R\|_{\mathrm{en}}
   &=O(N^{-5/4}),
   &
   \sigma_N
   &=1+O(N^{-5/2}).
\end{aligned}
\label{eq:E-tailnorms}
\end{equation}
\end{lemma}

\begin{proof}
By the coefficient identities in Eq.~\eqref{eq:E-coeff-norms},
\begin{equation}
   \|Q_N G\|^2
   =
   \sum_{n\geq N}|C_n(G)|^2,
   \qquad
   \|Q_N G\|_{\mathrm{en}}^2
   =
   \sum_{n\geq N}
   (\varepsilon_n+1)|C_n(G)|^2 .
   \label{eq:E-tail-coeff-norms}
\end{equation}
Using
\[
   \sum_{n\geq N} n^{-p}=O(N^{1-p}),
   \qquad p>1,
\]
together with
\[
   |C_n(\Phi_0^L)|=O(n^{-5/4}),
   \qquad
   |C_n(\Phi_0^R)|=O(n^{-9/4}),
\]
we obtain
\begin{align}
   \|Q_N\Phi_0^L\|^2
   &=
   O\!\left(
      \sum_{n\geq N}n^{-5/2}
   \right)
   =
   O(N^{-3/2}),
   \label{eq:E-left-L2-tail}
   \\
   \|Q_N\Phi_0^R\|^2
   &=
   O\!\left(
      \sum_{n\geq N}n^{-9/2}
   \right)
   =
   O(N^{-7/2}).
   \label{eq:E-right-L2-tail}
\end{align}
Since $\varepsilon_n+1=O(n)$,
\begin{equation}
   \|Q_N\Phi_0^R\|_{\mathrm{en}}^2
   =
   O\!\left(
      \sum_{n\geq N}n\,n^{-9/2}
   \right)
   =
   O\!\left(
      \sum_{n\geq N}n^{-7/2}
   \right)
   =
   O(N^{-5/2}).
   \label{eq:E-right-en-tail}
\end{equation}
Taking square roots yields the first three estimates in
Eq.~\eqref{eq:E-tailnorms}.

Finally, using
\[
   1-\sigma_N
   =
   \inner{Q_N\Phi_0^L}{Q_N\Phi_0^R},
\]
Cauchy--Schwarz gives
\begin{equation}
   |1-\sigma_N|
   \leq
   \|Q_N\Phi_0^L\|\,
   \|Q_N\Phi_0^R\|
   =
   O(N^{-3/4})O(N^{-7/4})
   =
   O(N^{-5/2}),
\end{equation}
and therefore
\[
   \sigma_N=1+O(N^{-5/2}).
\]
\end{proof}

\subsection{Energy convergence and correlator dependence}
\label{app:E:rate}

\begin{theorem}[Rate and prefactor]
\label{thm:rate}
Under {Assumption}~\ref{asm:branch}, either correlator at fixed $\mu>0$ satisfies
\begin{equation}
 E_0^{(N)}-E_0=-\frac{g\Psi_0(0)^2\beta}{\pi}N^{-3/2}+O(N^{-2}).
 \label{eq:E-final}
\end{equation}
\end{theorem}
\begin{proof}
$T_N^{(0)}=\sum_{n\ge N} \lambda_n\overline{c_n^L}c_n^R$.
Inserting the coefficient
asymptotics~\eqref{eq:E-coeff-asymptotics} and using
$\sum_{n\ge N}n^{-5/2}=\tfrac23N^{-3/2}+O(N^{-5/2})$ gives the first line
below; evaluating the asymmetric bounds~\eqref{eq:E-TW-bound}
and~\eqref{eq:E-RN-bound} with Lemma~\ref{lem:E-tail-estimates} gives the
other two:
\begin{equation*}
 \begin{aligned}
 T_N^{(0)}
 &=-\frac{g\Psi_0(0)^2\beta}{\pi}N^{-3/2}+O(N^{-9/4}),
 \\[1mm]
 |T_N^{(W)}|
 &\le C_W\|Q_N\Phi_0^L\|\,\|Q_N\Phi_0^R\|_{\mathrm{en}}
  =O(N^{-3/4})\,O(N^{-5/4})=O(N^{-2}),
 \\[1mm]
 |R_N|
 &\le C_AC_W^2\|Q_N\Phi_0^L\|\,\|Q_N\Phi_0^R\|_{\mathrm{en}}
  =O(N^{-3/4})\,O(N^{-5/4})=O(N^{-2}).
 \end{aligned}
\end{equation*}
Finally, $\sigma_N^{-1}=1+O(N^{-5/2})$ by Lemma~\ref{lem:E-tail-estimates},
so Eq.~\eqref{eq:E-error-id} gives the stated result.
\end{proof}

\begin{remark}[Bare convergence]
For the bare self-adjoint Hamiltonian
$\hat H=\hat H_{\mathrm{SHO}}+g\delta(x)$,
Eq.~\eqref{eq:E-coeff-asymptotics} gives
$|C_n(\Psi_0)|=O(n^{-5/4})$, so that
\[
   \|Q_N\Psi_0\|_{\mathrm{en}}^2
   =
   \sum_{n\geq N}(\varepsilon_n+1)|C_n(\Psi_0)|^2
   =
   O(N^{-1/2}).
\]
The standard Rayleigh--Ritz estimate with the trial function
$P_N\Psi_0$ then gives
$E_{0,\mathrm{bare}}^{(N)}-E_0=O(N^{-1/2})$,
in agreement with Grining
\emph{et al.}~\cite{grining_many_2015}.
\end{remark}

The correlator dependence enters through the cubic
coefficients in Eq.~\eqref{eq:E-beta}.
The rescaled errors in Sec.~\ref{sec:SHO-delta} have limits
$A_j(\mu):=\lim_{N\to\infty}A_j(\mu;N)=|g|\Psi_0(0)^2|\beta_j|/\pi$,
where $j\in\{\exp,\operatorname{erfc}\}$. Their ratio is
\begin{equation}
 \frac{A_{\exp}(\mu)}{A_{\operatorname{erfc}}(\mu)}
 =\left|1-\frac{3\mu^2}{2(2E_0+g^2)}\right|.
 \label{eq:E-ratio}
\end{equation}
This gives the correlator dependence of Fig.~\ref{fig:prefactor_tc} under
the same {assumption}.

\paragraph{Possible cancellation of the leading term.}
For TC-exp, the correlator parameter can be chosen to cancel
the leading $N^{-3/2}$ energy-error term. This occurs at the unique positive parameter
\begin{equation}
 \mu_\ast^2=\tfrac23(2E_0+g^2)
 \label{eq:E-mustar}
\end{equation}
for which $\beta_{\exp}=0$.
Positivity follows from
\[
2E_0+g^2
=
\|\Psi_0'-g\operatorname{sgn}(x)\Psi_0\|^2
+\|x\Psi_0\|^2
>0.
\]
TC-erfc has no corresponding cancellation because
$\beta_{\operatorname{erfc}}$ is nonzero and independent of $\mu$.

At $\mu_\ast$, the $|x|^3$ term vanishes and the next
potentially nonanalytic term is $\Psi_0(0)\beta_5|x|^5$, where
\begin{equation}
 \beta_5=
 \frac{g}{270}
 \left(124E_0^2+160E_0g^2+49g^4-9\right).
 \label{eq:E-beta5}
\end{equation}
With $\beta=0$, the first two boundary data in
Eq.~\eqref{eq:E-finite-recursion} vanish, and carrying the same
calculation one order further gives
\[
   (\mathcal L_+^2\Phi_0^R)'(0^+)
   =30\beta_5\Psi_0(0).
\]
Hence
\begin{equation}
   c_n^R
   =
   -\frac{30\beta_5\Psi_0(0)w_n(0)}{\lambda_n^3}
   +\frac{a_n}{\lambda_n^3},
   \qquad
   a_n:=C_n(\mathcal L_+^3\Phi_0^R).
   \label{eq:E-tuned-cn}
\end{equation}
Since $\mathcal L_+^3\Phi_0^R\in L^2(0,\infty)$, Parseval implies
$(a_n)\in\ell^2$. Together with the asymptotics in
Eq.~\eqref{eq:E-wn-lambda-asympt},
Eq.~\eqref{eq:E-tuned-cn} yields
\begin{align*}
   \|Q_N\Phi_0^R\|_{\mathrm{en}}^2
   &=
   O\!\left(\sum_{n\geq N}n^{-11/2}\right)
   +
   O\!\left(
      N^{-5}\sum_{n\geq N}|a_n|^2
   \right)
   \\
   &=O(N^{-9/2}),
\end{align*}
and therefore
\begin{equation}
   \|Q_N\Phi_0^R\|_{\mathrm{en}}
   =O(N^{-9/4}).
   \label{eq:E-tuned-energy-tail}
\end{equation}
Consequently, Eqs.~\eqref{eq:E-TW-bound} and
\eqref{eq:E-RN-bound} improve to
\[
   T_N^{(W)},\,R_N
   =O(N^{-3})
   =o(N^{-5/2}).
\]
The same square-summability bounds the $a_n$ contribution to the diagonal
tail,
\[
 \sum_{n\ge N}
 \frac{|c_n^L a_n|}{\lambda_n^2}
 \leq
 C'
 \left(\sum_{n\ge N}|a_n|^2\right)^{1/2}
 \left(\sum_{n\ge N}n^{-13/2}\right)^{1/2}
 =
 o(N^{-11/4}),
\]
so that this contribution is subleading as well. The explicit boundary term
in Eq.~\eqref{eq:E-tuned-cn} therefore gives, by the same calculation as in
Theorem~\ref{thm:rate},
\begin{equation}
 E_0^{(N)}-E_0
 =
 \frac{3g\Psi_0(0)^2\beta_5}{\pi}N^{-5/2}
 +o(N^{-5/2}),
 \qquad
 \mu=\mu_\ast .
 \label{eq:E-tuned-rate}
\end{equation}
For $\beta_5\ne0$, this gives $N^{-5/2}$ convergence.
If $\beta_5=0$, this term also cancels and the error is $o(N^{-5/2})$.
The cancellation is specific to the target state and uses its exact $E_0$;
$\mu_\ast$ is held fixed as $N$ changes.
At the level of local coefficients, this is a multiplicative analogue of
Hill's additive comparison functions~\cite[Eqs.~(4.21)--(4.24)]{hill_rates_1985}.
Figure~\ref{fig:appendix-e-tuning} illustrates the generic and tuned behavior.

\begin{figure}[ht!]
 \centering
 \includegraphics[width=0.85\linewidth]{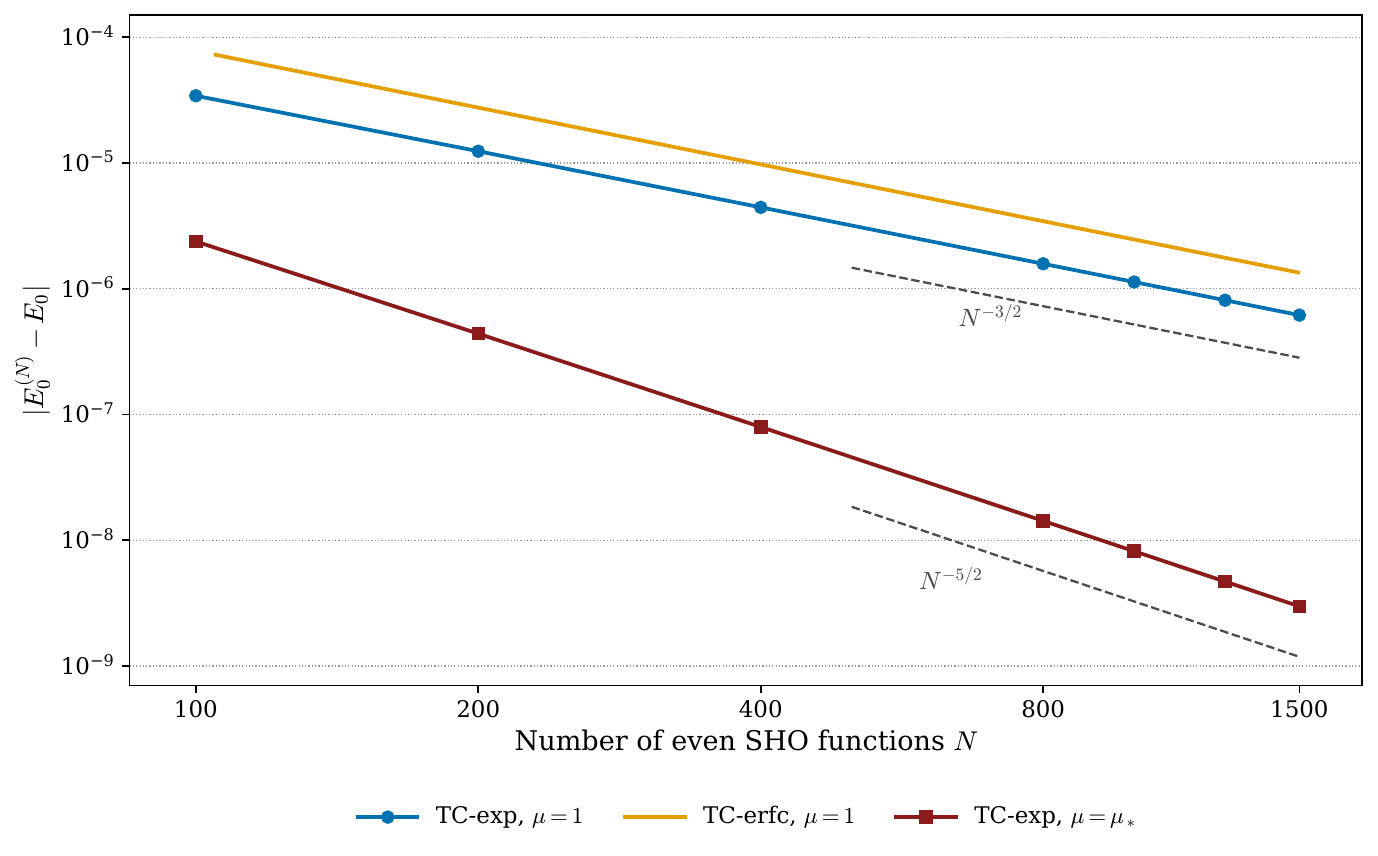}
 \caption{Basis convergence at $g=1$ for TC-exp with $\mu=1$ and fixed
 $\mu_\ast\simeq1.362715$, and TC-erfc with $\mu=1$.
 Here $N$ counts retained even SHO functions; dashed lines indicate
 $N^{-3/2}$ and $N^{-5/2}$ slopes.}
 \label{fig:appendix-e-tuning}
\end{figure}

\section{Re-analysis of published TC and xTC Hamiltonians}
\label{app:uvarov-reanalysis}

We test the sensitivity of solver-returned condition numbers to
eigenbasis choice on an independent set of operators. We re-analyze the
published full-TC and xTC qubit Hamiltonians of Uvarov and
Izmaylov~\cite{uvarov_accuracy_2025}. The data comprise Li, Be,
B, C, and N, with ten variational-Monte-Carlo correlator realizations for each
atom and each of the full-TC and xTC constructions. Their Hamiltonians use the
same Jordan--Wigner qubit representation as the molecular calculations in
Sec.~\ref{sec:molecular} and are analyzed on the full qubit space.
We use the untruncated operators and compare them with the corresponding
reported exact-operator condition numbers.

We evaluate these operators using the conditioning conventions of
Sec.~\ref{sec:cond}. Because their optimized correlators differ from the
fixed forms used here, this comparison concerns the eigenbasis convention rather
than physical conditioning magnitudes. Except for most full-TC Li
realizations and one full-TC N realization, $\kappa(V_{\mathrm{os}})$ remains
close to unity; all xTC values are at most $1.21$. The most extreme
full-TC Li realization has the maximum published raw value over the full data
set, $\kappa_S=33.8$, and retains a large convention-defined upper bound
after eigenspace orthonormalization, with
$\kappa(V_{\mathrm{os}})=19.45$ (Fig.~\ref{fig:uvarov-reanalysis} and
Table~\ref{tab:uvarov-reanalysis}).

\begin{figure*}[!ht]
  \centering
  \includegraphics[width=\linewidth]{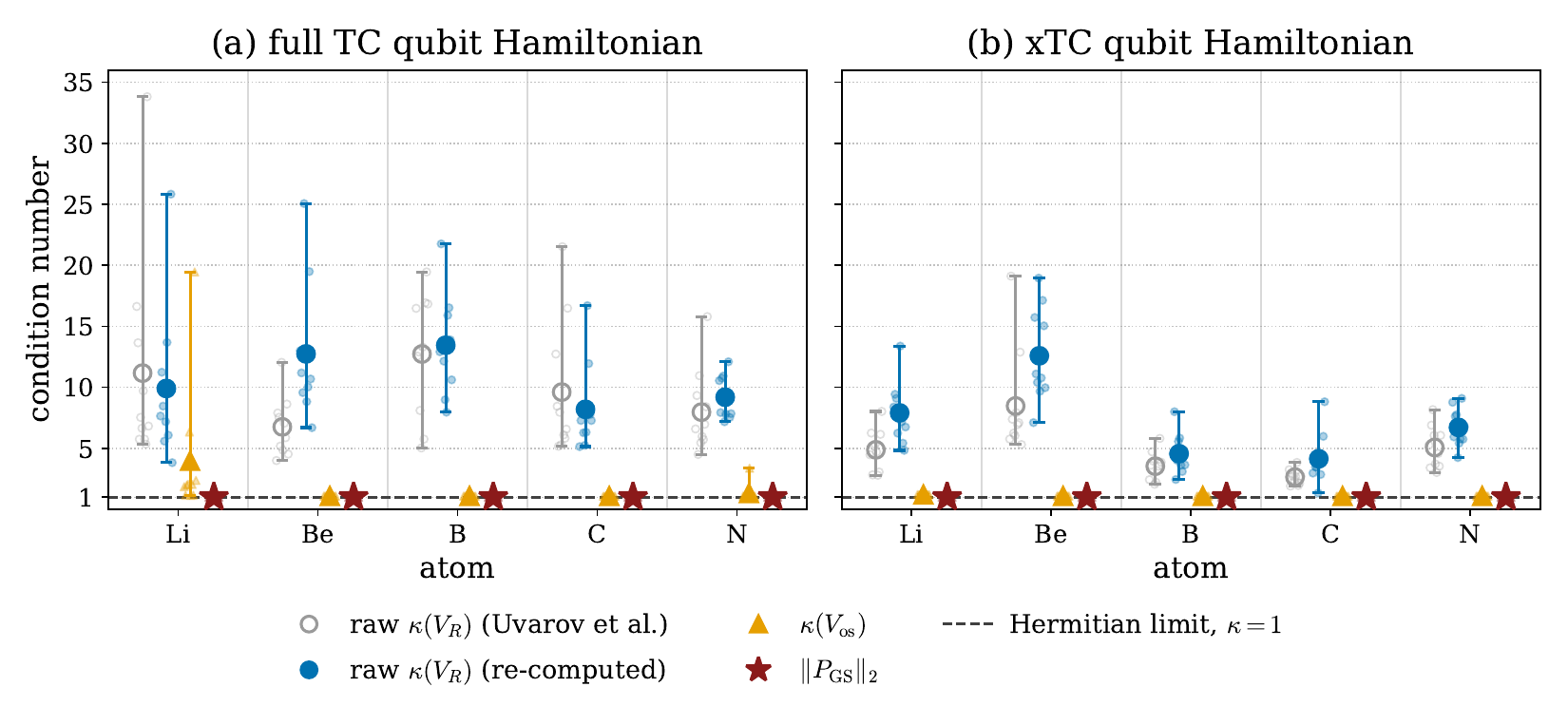}
  \caption{\label{fig:uvarov-reanalysis}%
  Re-analysis of the Li--N transcorrelated qubit Hamiltonians of
  Ref.~\cite{uvarov_accuracy_2025} (full TC, left; xTC, right).
  Large markers and error bars show the mean and full $[\min,\max]$ range
  over ten realizations; faint symbols show individual values.
  Gray and blue circles denote the published and re-computed raw condition
  numbers, while orange triangles and red stars denote
  $\kappa(V_{\mathrm{os}})$ and $\|P_{\mathrm{GS}}\|_2$; the basis-fixed quantities
  lie closer to the Hermitian limit except for most full-TC Li realizations and
  one full-TC N realization.}
\end{figure*}

\begin{table*}[ht!]
\caption{\label{tab:uvarov-reanalysis}%
Per-atom means and full $[\min,\max]$ ranges over ten realizations for the
published operators of Ref.~\cite{uvarov_accuracy_2025}.
Definitions and eigenspace conventions follow Sec.~\ref{sec:cond};
$\kappa_S$ and $\kappa(V_R)$ denote the published and re-computed raw values,
respectively. The final three columns are basis-fixed quantities, with
$\max_i\|P_i\|_2\le\kappa^*\le\kappa(V_{\mathrm{os}})$.}
\squeezetable
\begin{ruledtabular}
\begin{tabular}{lccccc}
Atom & $\kappa_S$~\cite{uvarov_accuracy_2025} & $\kappa(V_R)$ & $\kappa(V_{\mathrm{os}})$ & $\max_i\|P_i\|_2$ & $\|P_{\mathrm{GS}}\|_2$\\
\colrule
\multicolumn{6}{c}{full TC}\\
Li & $11.2\,[5.3,33.8]$ & $9.9\,[3.8,25.8]$  & $3.96\,[1.16,19.45]$ & $1.172\,[1.011,1.605]$ & $1.0006\,[1.0006,1.0006]$\\
Be & $6.8\,[4.0,12.0]$  & $12.7\,[6.7,25.1]$ & $1.09\,[1.09,1.09]$  & $1.003\,[1.003,1.004]$ & $1.0018\,[1.0016,1.0020]$\\
B  & $12.7\,[5.0,19.5]$ & $13.5\,[8.0,21.8]$ & $1.09\,[1.08,1.09]$  & $1.003\,[1.003,1.003]$ & $1.0014\,[1.0011,1.0014]$\\
C  & $9.6\,[5.2,21.5]$  & $8.2\,[5.1,16.7]$  & $1.08\,[1.07,1.08]$  & $1.003\,[1.003,1.003]$ & $1.0007\,[1.0006,1.0007]$\\
N  & $8.0\,[4.5,15.8]$  & $9.2\,[7.2,12.1]$  & $1.32\,[1.09,3.37]$  & $1.004\,[1.003,1.004]$ & $1.0000\,[1.0000,1.0000]$\\
\colrule
\multicolumn{6}{c}{xTC}\\
Li & $4.9\,[2.8,8.0]$   & $7.9\,[4.8,13.4]$  & $1.20\,[1.20,1.21]$  & $1.017\,[1.016,1.018]$ & $1.0006\,[1.0006,1.0006]$\\
Be & $8.5\,[5.3,19.1]$  & $12.6\,[7.1,18.9]$ & $1.09\,[1.09,1.09]$  & $1.003\,[1.003,1.004]$ & $1.0018\,[1.0016,1.0020]$\\
B  & $3.5\,[2.1,5.8]$   & $4.6\,[2.4,8.0]$   & $1.09\,[1.08,1.09]$  & $1.003\,[1.003,1.003]$ & $1.0014\,[1.0011,1.0014]$\\
C  & $2.6\,[1.9,3.8]$   & $4.1\,[1.3,8.8]$   & $1.08\,[1.07,1.08]$  & $1.003\,[1.003,1.003]$ & $1.0007\,[1.0006,1.0007]$\\
N  & $5.1\,[3.0,8.2]$   & $6.7\,[4.2,9.1]$   & $1.09\,[1.09,1.09]$  & $1.004\,[1.003,1.004]$ & $1.0000\,[1.0000,1.0000]$\\
\end{tabular}
\end{ruledtabular}
\end{table*}

\newpage

\bibliographystyle{apsrev4-2}
\bibliography{Refs}

\end{document}